\documentclass[amsmath,amssymb,a4paper,10pt]{article}
\usepackage{graphicx}
\usepackage{amssymb}
\usepackage{amsmath}
\usepackage{amsthm}
\usepackage{framed} 
\usepackage{a4wide}
\usepackage{enumerate}

\newtheorem{theorem}{Theorem}[section]

\newtheorem{lemma}[theorem]{Lemma}
\newtheorem{corollary}[theorem]{Corollary}
\newtheorem{remark}[theorem]{Remark}
\newtheorem{definition}[theorem]{Definition}

\newcounter{itm}
\newenvironment{myprotocol}[1]
  {\begin{minipage}{\columnwidth}
    \begin{framed}\hspace{0ex}
     \begin{minipage}{0.99\columnwidth}
       {\bf #1:}
       \setcounter{itm}{1}
 \begin{list}{}{\usecounter{itm}}}
   {    \end{list}
       \vspace{-1.5ex}
       \end{minipage}
     \end{framed}
    \end{minipage}\vspace{-0.6ex}}

\newcommand*{\textfrac}[2]{{{#1}/{#2}}}

\newcommand*{\bbN}{\mathbb{N}}

\newcommand*{\cA}{\mathcal{A}} 

\newcommand*{\cG}{\mathcal{G}}
\newcommand*{\cH}{\mathcal{H}}

\newcommand*{\cS}{\mathcal{S}}

\newcommand*{\cX}{\mathcal{X}}
\newcommand*{\cY}{\mathcal{Y}}
\newcommand*{\cZ}{\mathcal{Z}}

\newcommand*{\Pt}{\tilde{P}}

\newcommand*{\Ext}{\mathsf{Ext}}
\newcommand*{\Cond}{\mathsf{Cond}}
\newcommand*{\uniform}{\mathcal{U}}

\newcommand*{\mhalf}{{-\textfrac{1}{2}}}
\newcommand*{\half}{{\textfrac{1}{2}}}
\newcommand*{\id}{\mathsf{id}}
\newcommand*{\ket}[1]{|#1\rangle}
\newcommand*{\bra}[1]{\langle #1|}
\newcommand*{\proj}[1]{\ket{#1}\bra{#1}}
\newcommand*{\tr}{\mathsf{tr}}
\newcommand*{\sbin}{\{0,1\}}
\newcommand*{\rhob}{\widehat{\rho}}

\newcommand*{\spr}[2]{\langle #1|#2 \rangle}
\newcommand*{\supp}{\mathrm{supp}}
\newcommand*{\myspan}{\mathrm{span}}

\newcommand*{\BAD}{\mathcal{B}}

\newcommand*{\fr}[2]{\frac{#1}{#2}}
\newcommand*{\h}[5]{H(#1|#2)^{#3}_{\fr{#4}{#5}}}

\newcommand*{\inc}[2]{#1_{>#2}}
\newcommand*{\dec}[2]{#1_{\leq #2}}
\newcommand*{\hcond}[3]{H(\emptyset|#1)_{\fr{#2}{#3}}}
\newcommand*{\bPsi}{\widehat{\Psi}}

\newcommand*{\mytree}[1]{\mathbb{T}_{#1}}
\newcommand*{\treevsimple}[1]{\mathbf{v}_{#1}}
\newcommand*{\treev}[2]{\mathbf{v}^{#1}_{#2}}
\newcommand*{\treew}[2]{\mathbf{w}^{#1}_{#2}}

\newcommand*{\hmin}{H_{\min}}
\newcommand*{\hminrate}{R_{\min}}
\newcommand*{\hbmin}{\mathbf{h}_{\min}}

\newcommand*{\ExpE}{\mathop{\mathbb{E}}}
\newcommand*{\matrixsampler}{{parallel sampler}}
\newcommand*{\matrixsamplers}{{parallel samplers}}

\begin{document}

\title{Sampling of min-entropy relative to quantum knowledge}


\author{Robert K\"onig and Renato Renner\\
\textsf{rkoenig@caltech.edu, renner@phys.ethz.ch} }

\maketitle

\begin{abstract}
  Let $X_1, \ldots, X_n$ be a sequence of $n$ classical random
  variables and consider a sample $X_{s_1}, \ldots, X_{s_r}$ of $r
  \leq n$ positions selected at random. Then, except with
  (exponentially in $r$) small probability, the min-entropy
  $H_{\min}(X_{s_1} \cdots X_{s_r})$ of the sample is not smaller
  than, roughly, a fraction $\frac{r}{n}$ of the overall entropy
  $H_{\min}(X_1 \cdots X_n)$, which is optimal.
  
  Here, we show that this statement, originally proved in [S.~Vadhan,
  LNCS~2729, Springer, 2003] for the purely classical case, is still
  true if the min-entropy $H_{\min}$ is measured relative to a quantum
  system.  Because min-entropy quantifies the amount of randomness
  that can be extracted from a given random variable, our result can
  be used to prove the soundness of locally computable extractors in a
  context where side information might be quantum-mechanical.  In
  particular, it implies that key agreement in the bounded-storage
  model|using a standard \emph{sample-and-hash} protocol|is fully
  secure against quantum adversaries, thus solving a long-standing
  open problem.
\end{abstract}

\maketitle

\tableofcontents

\section{Introduction}

Let $X$ be a classical random variable and let $E$ be a (generally quantum-mechanical) system whose state might be correlated to $X$. The \emph{min-entropy of $X$ given $E$}, denoted $\hmin(X|E)$, is a natural measure for the uncertainty on the value of $X$ given access to the \emph{side information} $E$. More precisely, $\hmin(X|E)$ corresponds to the maximum length of a bitstring $R$ which is (a)~uniquely determined by $X$ and (b)~virtually uniform and independent of $E$.\footnote{See Lemma~\ref{lem:hminextractablekeylength} of Section~\ref{sec:basicdefinitionsh}  for a mathematically precise statement.}  

Here, we study the following question initiated by Nisan and Zuckerman~\cite{NisZuc96}.\footnote{Nisan and Zuckerman considered the special case where $E$ is classical.}   Given a sequence $X_1, \ldots, X_n$ of $n$ classical random variables with min-entropy (relative to side information $E$) at least $\hmin(X_1 \cdots X_n|E) \geq n \nu$, for some $\nu \geq 0$, what is the min-entropy $\hmin(X_{s_1} \cdots X_{s_r}|E)$ of a randomly selected sample $X_{s_1}, \ldots, X_{s_r}$ of $r$~positions? In other words, we are starting with a sequence $X_1, \ldots, X_n$ which contains at least $n \nu$ bits of uniform (relative to $E$) randomness, and we are interested in the amount of uniform (again relative to $E$) randomness of the subsequence $X_{s_1}, \ldots, X_{s_r}$.

As a main result, we show that the min-entropy per position is preserved under sampling, i.e., 
\[
  \frac{1}{n} \hmin(X_1 \cdots X_n|E) \geq \nu \quad \text{implies} \quad
  \frac{1}{r} \hmin(X_{s_1} \cdots X_{s_r}|E) \geq \nu + o(1) 
\]
(except with probability exponentially small in $r$). This generalizes a result by Vadhan~\cite{Vadhan03} who considered the case where $E$ is purely classical.\footnote{If the system $E$ is purely classical, it can generally be omitted in the analysis, as explained in Section~\ref{sec:extracionpriorcinfo}.}

A main application of this result is in the context of \emph{randomness extraction}. It relies on the \emph{leftover-hash lemma}~\cite{ILL89} (see also~\cite{BBCM95}), or, more precisely, its quantum generalization~\cite{Ren05} (see also~\cite{KoMaRe03, RenKoe05}), saying that the randomness of a classical random variable $X$, measured in terms of the min-entropy, can be extracted by applying a suitable hash function.  That is, $X$ can be mapped to a string $Z$ of size (roughly) $\hmin(X|E)$ which is virtually uniform and independent of $E$. Our result now implies that, given a long sequence $X_1, \ldots, X_n$ with sufficient min-entropy, random bits can be obtained by the \emph{sample-and-hash technique}, i.e., first sampling a subsequence $X_{s_1}, \ldots, X_{s_r}$ and then applying a two-universal hash function. 
 
The sample-and-hash technique is of interest in cryptography, in particular in the context of the \emph{bounded storage model}~\cite{Maurer92b}. Here, the security of cryptographic schemes is based on the assumption that a string of random variables $X_1, \ldots, X_n$, called \emph{randomizer}, is temporarily available for public access, but too long to be stored on a computer, even by a potential adversary. The idea then is to use this string as a source of secret randomness. 

Based on the original work by Maurer~\cite{Maurer92b}, various schemes for \emph{key expansion} in the bounded storage model have been proposed~\cite{DziMau02,DziMau04a,Lu02,Vadhan03}. These are mostly based on the sample-and-hash technique described above. More precisely, a short initial string is used for selecting positions of the randomizer $X_1, \ldots, X_n$. Then a hash function is applied to extract a key $Z$. 

Because the min-entropy of the randomizer $X_1, \ldots, X_n$ given the information $E$ stored by an adversary, $\hmin(X_1 \cdots X_n | E)$, is necessarily large, our result implies that the final key $Z$ is indeed uniform relative to~$E$ and, hence, secret. In other words, our result proves that key expansion in the bounded storage model is possible in the context of a quantum adversary. It generalizes previous results~\cite{DziMau04a,Lu02,Vadhan03} where security has been proved under the assumption that the adversary is purely classical.

\subsubsection*{Outline}
The paper is organized as follows: We first cover some background material on randomness extraction in Section~\ref{sec:basicdefsknownresults}. In Section~\ref{sec:ourcontribution}, we discuss our main result and its relation to prior work. Section~\ref{sec:proofsketch} provides an informal overview of the central ideas involved in the proof. The remainder of the paper is devoted to a formal derivation of our main results; in Section~\ref{sec:rulesandtoolsentropy}, we establish the required properties of min-entropy. We subsequently apply these to the problem at hand in Section~\ref{sec:entropysampling}, where we derive our main result. We conclude in Section~\ref{sec:bsmapplicationdetailed} by giving explicit parameters for key expansion in the bounded storage model.

\section{Basic definitions and known results\label{sec:basicdefsknownresults}}
\subsection{Randomness extractors}
\emph{Randomness extraction}, i.e., the process of transforming
partially random data $X$ into a uniformly distributed string $Z$,
plays an important role in computer science and, in particular,
cryptography.  For example, it is used to generate secure keys, given
only partially secret raw data. One of the most fundamental results in
the area of randomness extraction is the \emph{leftover-hash
  lemma}~\cite{ILL89}.  It states that the number of uniform bits
that can be extracted from a given random variable $X$ by two-universal
hashing (i.e., by applying a function chosen at random from a
two-universal set of hash functions) is roughly equal
to the \emph{min-entropy}\footnote{In the literature, the quantity
  $H_{\min}$ is also denoted $H_{\infty}$ and called \emph{R\'enyi
    entropy of order $\infty$}.} of $X$ defined by
\begin{equation} \label{eq:classminentropy}
  H_{\min}(X):=-\log \max_x P_X(x) \ .
\end{equation} 
We can express this result more formally by saying that two-universal hashing is an {\em extractor}. A $(k,\varepsilon)$-extractor is a function $\Ext:\cX\times\cY\rightarrow\cZ$ with the property that the random variable $Z=\Ext(X,Y)$  is $\varepsilon$-close to uniform\footnote{The $L_1$-norm of a function $f:\cZ\rightarrow \mathbb{R}$ is defined as $\|f\|:=\sum_{z\in\cZ}|f(z)|$.}, i.e.,
\begin{align*}
\frac{1}{2}\|P_{\Ext(X,Y)}-P_{\uniform_Z}\|\leq \varepsilon\ ,
\end{align*}
 whenever $X$ is a random variable $X$ with min-entropy at least $\hmin(X)\geq k$ and $Y$ is an independent and uniform seed, i.e., $P_Y\equiv P_{\uniform_{\cY}}$. (Here $P_{\uniform_\cZ}$ denotes the uniform distribution on $\cZ$.) A strengthening of this notion is the concept of a {\em strong extractor}, whose output is required to be uniform even conditioned on the seed $Y$. A strong $(k,\varepsilon)$-extractor satisfies the inequality
\begin{align}\label{eq:strongextractordefinition}
\|P_{\Ext(X,Y)Y}-P_{\uniform_Z}\cdot P_{\uniform_{Y}}\|\leq \varepsilon\ 
\end{align}
for all $P_{XY}\equiv P_X\cdot P_{\uniform_{\cY}}$ with $\hmin(X)\geq k$. Two-universal hashing corresponds to a strong $(k,\varepsilon)$-extractor with $\ell$ bits of output, for any $\varepsilon\geq 0$ and $k\geq \ell+2\log\textfrac{1}{\varepsilon}$.

While two-universal hashing is optimal in the number $H_0(Z):=\log |\cZ|$ of bits it can extract, it is not usable in certain applications. For example, computing the output $Z=\Ext(X,Y)$ might be infeasible, e.g., if the initial number $H_0(X)=n$ of bits is too large to be processed by a limited computational device. Also, in cryptographic scenarios, the seed $Y$ is sometimes a (secret) key of limited size (e.g., $H_0(Y)=O(\log n)$) compared to the length of $X$. Thus it is natural to try to find extractors with additional properties, such as efficient computability or limited seed length. An example of such a requirement which is important for applications in the bounded storage model is {\em local computability}; in other words, if $X=(X_1,\ldots,X_n)$ consists of a large number $n$ of blocks (or bits), the output $\Ext(X,Y)$ should only depend on a small subset $X_{\cS}=(X_{s_1},\ldots,X_{s_{r}})$ 
of these values, where $\cS=\{s_1,\ldots,s_{r}\}=\cS(y)\subset [n]=\{1,\ldots,n\}$ specifies the subset  for every $y\in\cY$. In other words, these extractors are of the form $\Ext(X,Y)=f(X_{\cS(Y)},Y)$.

\subsection{Randomness condensers}
With the aim of finding other constructions of extractors, it is natural to consider weaker notions of randomness generation. One natural way to generalise the concept of a randomness extractor is to require that the output is only close to a random variable with high min-entropy (instead of being close to a uniform random variable). This leads to the definition of a $(k,k',\varepsilon)$-condenser: This is a function $\Cond:\cX\times\cY\rightarrow Z$ such that for all random variables $X$ with $\hmin(X)\geq k$, there is a random variable $\bar{Z}$ with $\hmin(\bar{Z})\geq k'$ such that
\begin{align*}
\frac{1}{2}\|P_{\Cond(X,Y)}-P_{\bar{Z}}\|\leq \varepsilon\ ,
\end{align*}
where $Y$ is a uniform and independent seed on $\cY$. In terms of the so-called {\em smooth min-entropy}\footnote{The supremum ranges over all {\em subnormalised probability distributions} $P_{\bar{Z}}$, that is functions $P_{\bar{Z}}:\cZ\rightarrow [0,1]$ satisfying $\sum_{z\in\cZ} P_{\bar{Z}}(z)\leq 1$.} $\hmin^\varepsilon(Z):=\sup_{P_{\bar{Z}}:\|P_Z-P_{\bar{Z}}\|\leq \varepsilon} \hmin(\bar{Z})$
 this requirement is simply expressed by
\[
\hmin^\varepsilon(\Cond(X,Y))\geq k'\ .
\]
 The notion of a condenser is a strict generalisation of the notion of an extractor. Indeed, a $(k,\varepsilon)$-extractor $\Ext:\cX\times\cY\rightarrow\cZ$ is a $(k,\log|\cZ|,\varepsilon)$-condenser and vice versa.

Again, a stronger version of condensers is obtained by requiring that $\Cond(X,y)$ has high smooth entropy {\em with high probability} over $y$. The analog of~\eqref{eq:strongextractordefinition} defining a strong $(k,k',\varepsilon)$-condenser then is the requirement that for every $X$ with $\hmin(X)\geq k$, there exists a joint distribution $P_{\bar{Z}\bar{Y}}$ such that
\begin{align*}
\frac{1}{2}\|P_{\Cond(X,Y)Y}-P_{\bar{Z}\bar{Y}}\|\leq \varepsilon\ ,
\end{align*}
where $Y$ is independent of $X$ with uniform distribution $P_{Y}\equiv P_{\uniform_{\cY}}$  on $\cY$, and $\hmin(\bar{Z}|\bar{Y})\geq k'$. Here, the {\em conditional min-entropy}  is defined as
\[
\hmin(Z|Y):=-\log\sum_{y\in\cY}P_Y(y)\max_{z} P_{Z|Y=y}(z)\ .
\]
As before, this requirement is equivalent to demanding  that
\[
\hmin^\varepsilon(\Cond(X,Y)|Y)\geq k'\ ,
\]
where $\hmin^\varepsilon(Z|Y):=\sup_{\|P_{\bar{Z}{\bar{Y}}}-P_{ZY}\|\leq\varepsilon} \hmin(\bar{Z}|\bar{Y})$ is the {\em conditional smooth min-entropy}. With this definition, a function $\Ext:\cX\times\cY\rightarrow\cZ$ is a strong $(k,\varepsilon)$-extractor if and only if it is a strong $(k,\log|\cZ|,\varepsilon)$-condenser.

\subsection{Constructing locally computable extractors: The sample-and-hash approach}
Condensers can be used as a building block for constructing extractors. A possible way of obtaining a new construction is by applying an extractor to the output of a condenser. More precisely, suppose that
\begin{align*}
\Cond:&\cX_C\times\cY_C\rightarrow \cX_E\qquad&\textrm{ is a }(k_C,k_E,\varepsilon_C)-\textrm{condenser, and }\\
\Ext:&\cX_E\times\cY_E\rightarrow\cZ_E\qquad&\textrm{ is a }(k_E,\varepsilon_E)-\textrm{extractor}\ .
\end{align*}
It is easy to see that in this situation, the function
\begin{align*}
\widehat{\Ext}:\cX_C\times (\cY_C\times\cY_E)&\rightarrow\cZ_E\\
(x_C,(y_C,y_E))&\mapsto \Ext(\Cond(x_C,y_C),y_E)
\end{align*}
is a $(k_C,\varepsilon_C+\varepsilon_E)$-extractor. This is because the condenser $\Cond$ generates a random variable with a sufficient amount of min-entropy for $\Ext$. This conclusion is also true for the strong versions of these notions: if $\Cond$ and $\Ext$ are a strong condenser and a strong extractor, respectively, then the function $\widehat{\Ext}$ is a strong extractor.

Let us now return to the problem of constructing locally computable extractors. Clearly, if $\Cond(X,Y)=\Cond((X_1,\ldots,X_n),Y)$ is of the form $\Cond(X,Y)=X_{\cS(Y)}$, where $\cS(y)\subset [n]$ is a subset of indices for every $y\in\cY$, then the previous construction results in an extractor of the form
$\widehat{\Ext}(X,Y)=\widehat{\Ext}((X_1,\ldots,X_n),(Y_C,Y_E))=\Ext(X_{\cS(Y_C)},Y_E)$. This extractor is clearly locally computable.  This way of building a locally-computable extractor by first sampling a few indices specified by $\cS(Y_C)$ at random and then applying an extractor is called the {\em sample-and-hash approach}. Building locally computable extractors is thus reduced to the problem of constructing condensers of the form $\Cond(X,Y)=X_{\cS(Y)}$.

\subsection{Averaging samplers are condensers:  preservation of min-entropy rates\label{sec:averagingsamplersclassical}}
Consider a sequence of random variable $X=(X_1, \ldots, X_n)$ on $\cX^n$ and assume
that the \emph{min-entropy rate} $R_{\min}(X):=\frac{\hmin(X)}{H_0(X)}$ is lower bounded by $\mu$, i.e.,
\[
\mu\leq R_{\min}(X)=\frac{1}{n\log|\cX|} H_{\min}(X_1 \cdots X_n)\ .
\]
We will call the quantity $\frac{\hmin(X)}{H_0(X)}$ on the lhs the {\em min-entropy rate} of $X$. Suppose further that we select $r$ of these random variables at random, resulting in a subset $X_\cS=(X_{s_1},\ldots,X_{s_r})$ corresponding to indices $\cS=\{s_1,\ldots,s_r\}$.
Intuitively, one would expect that with high probability over the choice of $\cS$, the amount of randomness contained
in such a sample is proportional to its size $r=|\cS|$, i.e.,
\begin{equation} \label{eq:mainbound}
\mu-\delta\leq  R_{\min}(X_{\cS})=\frac{1}{|\cS|\log|\cX|}  H_{\min}(X_{\cS}) 
\end{equation}
for some small $\delta>0$.
In other words, we expect the min-entropy rate to be preserved under sampling.
Indeed, as shown by Vadhan~\cite{Vadhan03} (improving on previous work
by Nisan and Zuckerman~\cite{NisZuc96}),
inequality~\eqref{eq:mainbound} is correct with high probability (over
the choice of the sample $\cS=\{s_1,\ldots,s_r\}$). In the terminology of condensers, this is saying  that the function
\begin{align*}
\Cond:\cX^n\times\binom{[n]}{r}&\rightarrow \cX^r\\
((X_1,\ldots,X_n),\cS)&\mapsto X_\cS
\end{align*}
is a $(\mu n\log |\cX|,(\mu-\delta)r\log|\cX|,\varepsilon)$-condenser for some small $\delta,\varepsilon>0$. We call this function the {\em $\binom{[n]}{r}$-subset condenser}. 

Neglecting issues related to computational complexity, a condenser of the form $\Cond(X,Y)=X_{\cS(Y)}$ is fully specified by the distribution $P_{\cS}\equiv P_{\cS(Y)}$ over subsets of $[n]$. The $\binom{[n]}{r}$-subset condenser is simply represented by the uniform distribution over all subsets $\cS\subset [n]$ of size $|\cS|=r$. 

It is natural to ask which distributions over subsets $\cS$ give rise to good condensers.  Intuitively, a necessary condition is that the set of subsets $\cS$ covers $[n]$ well in some sense. In fact, Vadhan~\cite{Vadhan03} showed that it suffices for $\cS$ to be a so-called {\em averaging sampler}; i.e., a distribution over subsets of $[n]$ which can be used to approximate the average of any $n$~values. Formally, such a sampler is defined as follows:

\begin{definition}\label{def:averagingsampler}
An $(n,\xi,\varepsilon)$-sampler is a probability distribution $P_{\cS}$ over subsets $\cS\subset [n]$ with the property that
\begin{align}
\Pr_{\cS}\left[\frac{1}{|\cS|}\sum_{i\in\cS}\beta_i\leq \frac{1}{n}\sum_{i=1}^n\beta_i-\xi\right]\leq \varepsilon\textrm{ for all }(\beta_1,\ldots,\beta_n)\in [0,1]^n\ . \label{eq:averagingsamplerdef}
\end{align}
For simplicity, we will assume that $P_\cS$ is completely supported on subsets of the same size, and refer to this as $|\cS|\leq n$.
\end{definition}
Observe that we only consider a one-sided error.\footnote{We point out that the notion of {\em samplers} is usually defined differently in the computer science literature. There, a {\em sampler} is an algorithm which efficiently approximates the average of a large number of values. The aim is to give an estimate of the average $\frac{1}{n}\sum_{i=1}^n \beta_i$ of an (arbitrary) vector $(\beta_1,\ldots,\beta_n)\in [0,1]^n$ whose entries are accessible in the form of an oracle. Here we restrict our attention to so-called {\em averaging} samplers: These output the value $\frac{1}{|\cS|}\sum_{i\in\cS}\beta_i$ of a (randomly) chosen subset $\cS\subset [n]$ of values. For a more detailed discussion of samplers and their computational aspects, see~\cite{goldreich97}.}
We will call $\xi$ the {\em accuracy} of the sampler, and $\varepsilon$ its {\em failure probability}. Returning to our example, the uniform distribution  over subsets of a fixed size is an averaging sampler with the following parameters. 
\begin{lemma}\label{lem:subsetsampler}
Let $r<n$ and let $P_\cS$ be the uniform distribution over subsets $\cS\subset [n]$ of size $|\cS|=r$. This defines a $(n,\xi,e^{-\textfrac{r\xi^2}{2}})$-sampler for every $r>0$ and $\xi\in [0,1]$.
\end{lemma}
\noindent This statement is a consequence of the Hoeffding-Azuma inequality and given as Lemma~5.5 of~\cite{babaihayes05}. We call this sampler simply the {\em $\binom{[n]}{r}$-subset sampler}. It will be sufficient for our purposes, but our results hold more generally for arbitrary averaging samplers.

Vadhan showed that in the same way as the $\binom{[n]}{r}$-subset sampler gives rise to the $\binom{[n]}{r}$-condenser, any averaging sampler defines a corresponding condenser (with appropriate parameters). In other words, a probability distribution $P_\cS$ over subsets of $[n]$ with the sampler property~\eqref{eq:averagingsamplerdef} preserves the min-entropy rate when picking a random subset, in the sense of~\eqref{eq:mainbound}.

\subsection{Extractors, condensers and prior classical information~\label{sec:extracionpriorcinfo}}
In cryptographic settings, it is often desirable to generate randomness which is not only (close to) uniform, but also {\em independent} of an adversary's prior information. We first consider the case where the adversary is classical, such that her information is described by a  random variable~$E$. In other words, the task is to generate a  key $Z$ satisfying 
$\frac{1}{2}\|P_{Z\tilde{E}}-P_{\uniform_{\cZ}}\cdot P_{\tilde{E}}\|\leq\varepsilon$, where $\tilde{E}$ summarises the adversary's knowledge. 

Suppose the initial situation is described by a joint distribution $P_{XE}$, where $X$ is held by the honest parties, and the adversary holds $E$. We will assume that the adversary's information about $X$ is limited; this  is expressed by a lower bound on  the conditional entropy $\hmin(X|E)$. Conveniently, a strong $(k,\varepsilon)$-extractor achieves key extraction in this setup, when invoked with (public) independent randomness $Y$. That is, we have
\begin{align}\label{eq:extractoruniform}
\frac{1}{2}\|P_{\Ext(X,Y)YE}-P_{\uniform_{\cZ}}\cdot P_{\uniform_{\cY}}\cdot P_E\|\leq 2\varepsilon
\end{align}
for all $P_{XE}$ with $\hmin(X|E)\geq k+\log\textfrac{1}{\varepsilon}$. In other words, if the adversary's initial prior information $E$ about $X$ is limited, the extracted key $Z=\Ext(X,Y)$ will look uniform to the adversary even if he is given the seed of the extractor (i.e., $\tilde{E}=(E,Y)$). This procedure of using public (independent) randomness to generate secret keys from partially secret information is well-known as {\em privacy amplification}~\cite{BBCM95} (usually in conjunction with two-universal hashing as an extractor).

Inequality~\eqref{eq:extractoruniform} is a trivial application of Markov's inequality; it is obtained by applying the extractor property to the conditional distributions $P_{X|E=e}$. A similar conclusion holds more generally for any strong $(k,k',\varepsilon)$-condenser: Here we have
\begin{align}
\hmin^{2\varepsilon}(\Cond(X,Y)|YE)\geq k'\label{eq:condclassicalprior}
\end{align}
for all joint distributions $P_{XE}$  with $\hmin(X|E)\geq k+\log\textfrac{1}{\varepsilon}$.  This means that the problem of randomness extraction in the context of prior classical information essentially reduces to the randomness generation problem without any side-information.

\subsection{Extractors, condensers and prior quantum information}
\label{sec:extracionpriorqinfo}
The mentioned property of extractors and condensers fails to be true in cases where the adversary's prior information $E$ is quantum. Indeed, in this case, the conditional distributions $P_{X|E=e}$ are no longer defined, and the analysis of randomness extraction has to be done differently. 

The relevant concepts in this modified setup are sufficiently straightforward to define: Consider a classical random variable $X$ and a quantum system $E$ which is correlated to this variable. This situation is completely described by a classical-quantum state $\rho_{XE}=\sum_{x\in\cX} P_X(x)\proj{x}\otimes\rho^x_E$ (where $\{\ket{x}\}_{x\in\cX}$ is an orthonormal basis), or equivalently the ensemble $\{P_X(x),\rho^x_E\}_{x\in\cX}$ on $E$. For the purpose of randomness extraction, the relevant measure of min-entropy is the conditional min-entropy 
$H_{\min}(X|E)$ introduced in~\cite{Ren05}; this quantity is defined by\footnote{This definition is meaningful arbitrary bipartite states $\rho_{XE}$ even with non-classical part $X$.}
\[
  \hmin(X|E) := - \log \min_{\sigma_E} \min \{\lambda: \rho_{XE} \leq \lambda \cdot \id_X \otimes \sigma_E \} \ .
\]
 The conditional min-entropy generalizes the classical min-entropy~\eqref{eq:classminentropy}. For classical-quantum states $\rho_{XE}$, the min-entropy $\hmin(X|E)$ characterises the  amount of uniform randomness $Z=f(X)$ that can be extracted from $X$ such that $Z$ is independent of $E$.

In terms of this measure of prior information, a {\em $(k,\varepsilon)$-strong quantum extractor} is a function $\Ext:\cX\times\cY\rightarrow\cZ$ with the property that (cf.~\eqref{eq:extractoruniform})
\begin{align}
\frac{1}{2}\|\rho_{\Ext(X,Y)YE}-\rho_{\uniform_\cZ}\otimes\rho_{\uniform_{\cY}}\otimes\rho_E\|\leq \varepsilon\ \label{eq:strongextractorquantumdef}
\end{align}
for all classical-quantum-states $\rho_{XE}$ with $\hmin(X|E)\geq k$. In this expression, $Y$ is an independent and uniform seed on $\cY$, and $\rho_{\uniform_\cZ}$ denotes the completely mixed state on $\cZ$, i.e., the state $\frac{1}{|\cZ|}\sum_{z\in\cZ} \proj{z}$. Clearly, a $(k,\varepsilon)$-strong quantum extractor is a $(k,\varepsilon)$-strong extractor in the original (classical) sense. The converse is not true in general (see~\cite{gavinskyetal07} for a particularly striking example in the bounded storage model). However, the left-over hash lemma can be generalised to the quantum case: the two-universal hashing construction $\Ext:\sbin^n\times\sbin^n\rightarrow\sbin^\ell$ is a $(k,\varepsilon)$-strong quantum extractor for any $k\geq \ell+2\log\textfrac{1}{\varepsilon}$, as shown by Renner~\cite{Ren05}. (The optimality of this extractor with respect to the number of extracted bits is shown below in Lemma~\ref{lem:hminextractablekeylength}.) As with classical extractors, an important goal is to find constructions which are more randomness-efficient, and satisfy additional properties such as local computability.

Similarly, a {\em $(k,k',\varepsilon)$-strong quantum condenser} $\Cond:\cX\times\cY\rightarrow\cZ$ is defined by the requirement (cf.~\eqref{eq:condclassicalprior})
\begin{align*}
\hmin^\varepsilon(\Cond(X,Y)|YE)\geq k'\ 
\end{align*}
for all $\rho_{XE}$ with $\hmin(X|E)\geq k$. In this expression, the {\em smooth min-entropy} $\hmin^\varepsilon(X|E)$ is defined by a maximisation over a set of operators in the vicinity of $\rho_{XE}$. Note that there is a certain freedom in these definitions (the only constraint is the preservation of the desirable composability properties). We choose to define the smooth min-entropy as 
\begin{align*}
\hmin^\varepsilon(X|E)=\sup_{\substack{\bar{\rho}_{XE}:\|\bar{\rho}_{XE}-\rho_{XE}\|\leq\varepsilon\\
\tr(\bar{\rho}_{XE})\leq 1}} \hmin(X|E)_{\bar{\rho}_{XE}}\ ,
\end{align*}
where the maximisation is over all subnormalised nonnegative operators $\bar{\rho}_{XE}$ in an $\varepsilon$-ball around $\rho_{XE}$, and the quantity on the rhs~is the  min-entropy of the corresponding operator (see below for a formal definition). As shown in~\cite{Ren05}, if $X$ is classical, this supremum is achieved by an operator $\bar{\rho}_{XE}$ which is classical on $\cX$. To guarantee compatibility of quantum condensers and extractors, we require a $(k,\varepsilon)$-strong quantum extractor to satisfy~\eqref{eq:strongextractorquantumdef} for all subnormalised nonnegative operators  $\rho_{XE}$ with classical part $X$ and $\hmin(X|E)\geq k$. This is true for two-universal hashing, as the analysis in~\cite{Ren05} shows.

\section{Our contribution\label{sec:ourcontribution}}
\subsection{Main result: samplers are quantum condensers}
Our main result states that samplers can be used to ``condense'' min-entropy even in a quantum context, in the same way as they give rise to randomness condensers for classical distributions (as discussed in Section~\ref{sec:averagingsamplersclassical}). More precisely, we consider an $n$-tuple $X^n=(X_1,\ldots,X_n)$ of random variables on $\cX^n$, where $\cX$ is a (large) alphabet. We show that relative to a quantum system $E$, the {\em min-entropy rate} is preserved when picking a random subset $X_{\cS}$ (using a sampler). 

To express this in a concise form, we introduce the {\em min-entropy rates}
\begin{align}
\hminrate^\varepsilon(A|B)_\rho&:=\frac{\hmin^\varepsilon(A|B)_\rho}{H_0(A)_\rho}\ ,\label{eq:smoothminentropyrate}
\end{align}
where $H_0(A)=\log |\cA|$ is the alphabet size of $A$. Our main result states that this quantity is approximately preserved under sampling. Clearly, when applied to a $(n,\xi,\varepsilon)$-sampler, such a statement must depend on the accuracy $\xi$ of the sampler and its failure probability~$\varepsilon$. For such a sampler and the situation described above, our main result is given by the inequality
\begin{align}
\hminrate^{\varepsilon'}(X_{\cS}|\cS E)_\rho&\geq \hminrate(X^n|E)_\rho-3\xi-2\kappa \log\textfrac{1}{\kappa}\ ,\label{eq:samplerlowerbound}
\end{align}
where the parameters $\varepsilon'$ and $\kappa$ are equal to
\begin{align*}
\varepsilon'=2\cdot 2^{-\xi n\log|\cX|}+3\varepsilon^{\textfrac{1}{4}}\qquad\textrm{ and }\qquad
\kappa =\frac{n}{|\cS|\log |\cX|}\ .
\end{align*}
(This result is stated as Corollary~\ref{cor:samplersminentropypreservation} below.) 
This inequality shows that (for appropriate alphabet sizes) the min-entropy rate is preserved, up to the accuracy of the sampler. As expected, the failure probability $\varepsilon$ of the sampler is reflected in the distance (i.e., the smoothness parameter $\varepsilon'$). In fact, this distance mainly depends on the failure probability of the sampler, and the term $2\cdot 2^{-\xi n\log|\cX|}$ is usually negligible.

Observe that the expression $2\kappa\log\textfrac{1}{\kappa}$ on the lhs~of~\eqref{eq:samplerlowerbound} goes to zero as $\kappa\rightarrow 0$. The parameter $\kappa$ captures the alphabet sizes in the problem; our result applies to regions where $\kappa$ is small. As $|\cS|\leq n$, this is equivalent to demanding that $\cX$ is a large alphabet. Thus we will henceforth assume that the random variables $X_i$ are large ``blocks''(instead of individual bits).

It is instructive to apply this result to the $\binom{[n]}{r}$-subset sampler: Here the error probability $\varepsilon$ decays exponentially with $r$ for any fixed $\xi\in [0,1]$. 
More precisely, the following reformulation of~\eqref{eq:samplerlowerbound} is obtained by setting $\Delta=3\xi+2\kappa\log\textfrac{1}{\kappa}$. 
We then have
\begin{align*}
\hminrate^{\varepsilon}(X_{\cS}|\cS E)_\rho&\geq
\hminrate(X^n|E)_\rho-\Delta\qquad\textrm{ for any }\Delta\geq 2\kappa\log\textfrac{1}{\kappa}\ , \textrm{ where }\\
\varepsilon&=e^{-\Omega\left(r(\Delta-2\kappa\log\textfrac{1}{\kappa})^2\right)} 
\end{align*}Thus (smooth) min-entropy-rate is preserved up to a constant, with an exponentially small error~$\varepsilon$.

\subsection{Related work}

We briefly explain how our contribution relates to other known results.  We stress that giving a comprehensive review of all the relevant areas is not the aim of this section. Nor do we attempt to provide a complete list of references; the pointers given here are mainly intended to facilitate access to further literature. We identify the following broad points of contact with previous work:

\subsubsection*{Quantum information about classical random variables: Random access encodings}

Our main result is an upper bound on the amount of information a quantum system gives about certain classical values. As such, it fits into a long line of work, the most prominent example of which is Holevo's  upper bound on the accessible information~\cite{holevo73}. 

More specifically, our result bounds the information about a (randomly selected) substring $X_{\cS}=(X_{s_1},\ldots,X_{s_r})$ of a classical string $X^n=(X_1,\ldots,X_n)$. In this sense, it is structurally identical to the {\em random access encodings} studied by Ambainis, Nayak, Ta-Shma and Vazirani~\cite{ANTV99}. Formally, an {\em $\binom{[n]}{1}\overset{p}{\mapsto}m$ random access encoding} maps $n$-bit strings $X^n=(X_1,\ldots,X_n)$ into $m$-qubit states $\rho_X$  in a way that allows to retrieve any (single) bit $X_i$ with probability at least $p$ by a measurement.\footnote{The notation used here is slightly different from these original papers.} Strengthening the result of~\cite{ANTV99}, Nayak~\cite{Nayak99} showed that at least  $m\geq (1-h(p))n$ qubits are needed  for this kind of encoding. (Here $h(\cdot)$ is the binary entropy function.) This can be understood as a precise expression of the qualitative statement that $m$~qubits cannot be used to store more than $m$~classical bits.

Recently, this result has been significantly generalized by Ben-Aroya, Regev and de Wolf~\cite{aroyaetal07}. They studied $\binom{[n]}{r}\overset{p}{\mapsto} m$ encodings, where the aim is to be able to retrieve {\em each substring} $X_\cS$ of length $r=|\cS|$ with probability at least~$p$ from the $m$-qubit state. They showed that the success probability $p$ decreases exponentially in $r$ when $m<0.7n$. The result~\cite{aroyaetal07} of Ben-Aroya et al.\ is of the same form as ours. Indeed, as explained below, in terms of entropies, it expresses the fact that in the studied situation, the  {\em entropy-rate} is preserved. However, there are at least three major differences to our work.

Firstly, \cite{aroyaetal07} provides an upper bound on the \emph{guessing probability} $p(X_{\cS} | E)$, i.e., the probability of retrieving the correct value $X_{\cS}$ given quantum information $E$, which is the figure of merit in the context of random access encodings. By virtue of the identity $p(X_{\cS}|E) = 2^{-H_{\min}(X_{\cS}|E)}$ (see~\cite{KoReSc08} for more details), their result implies a lower bound on the min-entropy (and, hence, also on the smooth min-entropy for any $\varepsilon \geq 0$). In contrast, we derive a lower bound on the \emph{smooth} min-entropy $H_{\min}^{\varepsilon}(X_{\cS}|E)$, which is the relevant quantity in the context of randomness extraction (e.g., in the bounded storage model). This, in turn, implies an upper bound on the guessing probability $p(X_{\cS} | E) \leq
2^{-H^{\varepsilon}_{\min}(X_{\cS} | E)} + \varepsilon$. Because our result is not optimized for very small $\varepsilon$, the upper bound on the guessing probability following from our result might be far below the bound of~\cite{aroyaetal07}. On the other hand, the bound on the smooth min-entropy implied by the result of~\cite{aroyaetal07} is below our bound, which is asymptotically optimal.

A second, apparently insignificant yet important difference between~\cite{aroyaetal07} and our work is the alphabet size of the random variables $X_i$  in the tuple $X^n=(X_1,\ldots,X_n)\in \cX^n$. While these are single bits in~\cite{aroyaetal07}, they may be random variables over a large alphabet in our work, i.e., every $X_i\in\sbin^c$ is itself a $c$-bit string for some (usually large\footnote{Note that our main result as stated in Corollary~\ref{cor:samplersminentropypreservation} does not directly apply to cases where the alphabet of the random variables $X_i$ is too small, e.g., if they are single bits. However, our result can be extended to these cases in the following way. Given, for instance, a bitstring $B = (B_1, \ldots, B_N)$, the permuted string, $B_{\pi} := (B_{\pi(1)}, \ldots, B_{\pi(N)})$, for any permutation $\pi \in S_N$, has the same min-entropy as $B$. We can therefore apply Corollary~\ref{cor:samplersminentropypreservation} to the permuted string $B_{\pi}$, for a randomly chosen $\pi$, and appropriately chosen partitioning $B_{\pi} = (X_1, \ldots, X_n)$ into $n$ blocks, resulting in a substring $B' = X_{\cS}$ with high min-entropy. Since, after undoing the permutation on $B'$, this string is identically distributed as a bitstring chosen at random from $B$, we conclude that the min-entropy rate is essentially conserved under random sampling of bits.}) $c$. In the latter case, choosing a random subset $\cS\subset [n]$ of size $r=|\cS|$ effectively generates a substring $X_{\cS}=(X_{s_1},\ldots,X_{s_r})$ of length $\ell = c r$ by blockwise sampling. For $c \gg \log n$, this procedure consumes only $\log \binom{n}{r} \leq r \log n \ll \ell$ random bits, in contrast to $\log\binom{n}{c r} \geq \ell$ when the individual bits are chosen at random. When applied to the bounded storage model, this means that we can extract more bits than the number of initial (shared) key bits. On the other hand, while the sample-and-hash approach can in principle be applied using the result of~\cite{aroyaetal07}, the number of extracted bits is much smaller than the number of initial key bits, that is, no significant key expansion can be achieved. 

Thirdly, the result of~\cite{aroyaetal07} measures the initial quantum information about the string $X$ in terms of the number $m$ of qubits used in the encoding. More precisely, it is assumed that $X$ is uniformly distributed and that at most $m$ qubits containing information about $X$ are stored in a quantum system $E$ (formally, $H_0(E) \leq m$, where $H_0(E)$ denotes the logarithm of the dimension of $E$). In contrast, our result applies more generally to situations where merely a lower bound on the quantity $H_{\min}(X|E)$ is known, while the quantum system $E$ may be arbitrarily large. The above special case where the dimension of $E$ is bounded follows from the general fact that $H_{\min}(X|E)\geq H_{\min}(X)-H_0(E)$.

\subsubsection*{Key extraction: Extractors and privacy amplification}

The study of key extraction in the presence of a classical adversary is, as argued above, equivalent to the question of constructing randomness extractors (see~\cite{Shal02} for a survey of this intensely studied subject). More specifically, two-universal hashing was first 
applied to privacy amplification in~\cite{BeBrRo88,BBCM95}. Maurer and Dziembowski~\cite{DziMau02,DziMau04a} obtained optimal protocols for key extraction in the (classical) bounded storage model. Lu~\cite{Lu02} made the connection to locally (or on-line) computable  {\em strong extractors}. Vadhan subsequently gave essentially optimal constructions by showing that sampling preserves min-entropy~\cite{Vadhan03}; the sampling approach for extracting randomness can be traced back to the work of Nisan and Zuckerman~\cite{NisZuc96} and abounds in the randomness extractor literature.

The situation in the presence of an adversary with prior quantum information is more intricate, and much less is known to date.  On the negative side, Gavinsky, Kempe, Kerenidis, Raz and de Wolf~\cite{gavinskyetal07} gave a surprising example of a classical extractor which fails to extract randomness in the presence of a quantum adversary (with a similar amount of quantum memory). 
On the positive side, Renner~\cite{Ren05} showed that two-universal hashing is optimal in the amount of extracted key (see also~\cite{RenKoe05}). K\"onig and Terhal~\cite{KoeTer06} showed that strong extractors with binary output also extract secure bits against quantum adversaries; this provides quantum extractors with short seeds, but does not achieve significant key expansion in the bounded storage model. Recently, new constructions of quantum extractors were proposed by Fehr and Schaffner~\cite{fehrscha07}. While these extractors can be used for privacy amplification, their parameters are not suitable for the bounded storage model.

\section{Proof sketch\label{sec:proofsketch}}
In this section, we give an informal overview of the main ideas involved in the proof of the result~\eqref{eq:samplerlowerbound}. In Section~\ref{sec:proofidea}, we give a simple proof of an analogous statement for the (classical) Shannon entropy. Our proof for (quantum) min-entropy mimics this line of argument, but differs in a few major points, as discussed below.

A few of our techniques may be of independent interest. A central idea is the {\em splitting} of a state into several components based on {\em conditional operators}; it leads to a modified chain-rule for min-entropies. We explain this in Section~\ref{sec:towardsmodifiedchainrule}. The converse procedure which we call {\em recombining} is especially interesting when only subsets of the split states are used in the recombination. The outcome of such a partial recombination is a state which approximates the original state. By selecting split states in a systematic fashion, we can single out the high-entropy components of a state. As we 
explain in Section~\ref{sec:partialrecombination}, this  is a fundamental tool for showing that a given state has a certain amount of (smooth) min-entropy.

We will conclude this part of the paper with an overview of how these two procedures -- the splitting and the recombining -- can be combined with an argument about samplers to give the result we seek. 

We stress that this section is introductory in nature, and the technical details are left to later sections. In particular, we will only argue qualitatively,
and the formulas in Sections~\ref{sec:towardsmodifiedchainrule} and~\ref{sec:partialrecombination} are not meant to be taken literally. However, the basic structure of our arguments will be exactly as sketched here.

\subsection{Proof idea\label{sec:proofidea}}
We show how to derive a modified statement related to~\eqref{eq:samplerlowerbound},  where we restrict our attention to probability distributions
and
where the min-entropy $H_{\min}(A|B)$ is replaced by the (conditional) Shannon entropy $H(A|B)=H(AB)-H(B)$. (Here $H(X)=-\sum_{x\in\cX} P_X(x)\log P_X(x)$ denotes the usual Shannon entropy.) This kind of proof is sketched in~\cite{NisZuc96}  to give an intuition why samplers are good condensers. However, neither the proofs in~\cite{NisZuc96} nor Vadhan's proof~\cite{Vadhan03} proceed along these lines.

The essential properties of the Shannon entropy used are the subadditivity property
\begin{align}\label{eq:shannonentropysubadditivity}
H(A|BC)\leq H(A|B)\ ,
\end{align}
i.e., the fact that further conditioning can only reduce the entropy, and the chain-rule
\begin{align}\label{eq:chainruleshannonentropy}
H(AB|C)=H(A|BC)+H(B|C)\ .
\end{align}

Consider  a probability distribution  $P_{X^nE}$, where $X^n=(X_1,\ldots,X_n)$ is an $n$-tuple of random variables. Our aim is to show that with high probability over a randomly chosen subset $\cS\subset [n]$ of size $|\cS|=r$, the entropy of $H(X_{\cS}|E)$ is approximately equal to $\frac{r}{n}H(X^n|E)$.

To abbreviate the notation, we will define 
\begin{align*}
\inc{X}{j}&=X_{j+1}X_{j+2}\cdots X_n\\
\dec{X}{j}&=X_1X_2\cdots X_j\qquad\textrm{ for }j\in [n]\\
\inc{X}{n}&=\emptyset\ 
\end{align*}
for any such $n$-tuple. The first step is what we call a {\em splitting step: } The chain-rule~\eqref{eq:chainruleshannonentropy} implies that the entropy $H(X^n|E)$
 can be decomposed into its
constituents, 
\begin{align*}
H(X^n|E) = \sum_{i=1}^n \alpha_i\qquad\textrm{ where }\qquad \alpha_i := H(X_i | \inc{X}{i}E)\textrm{ for any } i\in [n]\ .
\end{align*}
In other words, we have split the entropy into a sum of individual components.

If we now select a subset  $\cS\subset [n]$ of $r=|\cS|$ indices at random, then
Chernoff's inequality implies that the inequality \vspace{-1.5ex}
\begin{equation} \label{eq:chernoff} 
\frac{1}{r}  \sum_{s\in\cS} \alpha_{s} \geq \frac{1}{n} H(X^n|E) - O(\textfrac{1}{\sqrt{r}}) \
  ,  \vspace{-1ex}
\end{equation}
holds except with probability exponentially small in $r$. Note that this holds more generally for any $(n,\xi,\varepsilon)$-sampler $\cS$ with corresponding adaptations.

By strong subadditivity, we have 
\begin{align}
\alpha_{j} = H(X_{j} | \inc{X}{j}E) \leq H(X_{j} |\inc{X}{j\cap \cS}E)\qquad\textrm{ for any }j\in [n]\ ,\label{eq:strongsubapp}
\end{align}
where $X_{>j\cap \cS}$ is the concatenation of all
variables $X_{i}$ with $i>j$ and $i\in\cS$. With this inequality we can  essentially eliminate all variables $X_i$ with $i\not\in\cS$ from our inequalities.

The final step is what we call a {\em recombination step}: Using the chain
rule once again, we obtain with~\eqref{eq:strongsubapp}
\begin{align*}
H(X_{\cS}|E) = \sum_{s\in\cS}
H(X_{s} | \inc{X}{s\cap \cS}E) \geq \sum_{s\in\cS}\alpha_{s}\ .
\end{align*}
In other words, we can get a lower bound on the joint entropy $H(X_\cS|E)$ by combining the individual contributions $s\in\cS$. 

With~\eqref{eq:chernoff}, we conclude
that  with all but exponentially small probability, the (Shannon)-entropy rate is preserved when selecting a random subset.

The proof of our main result for min-entropy follows the same lines, with a modified chain-rule for min-entropies. 
Notice that the chain-rule in the form~\eqref{eq:chainruleshannonentropy} can be seen as the combination of two inequalities,
\begin{align}
H(AB|C)&\leq H(A|BC)+H(B|C)\ ,\label{eq:chainruleshannonsplitting}\\
H(AB|C)&\geq H(A|BC)+H(B|C)\label{eq:chainruleshannonrecombination}
\end{align}
both of which are used in the proof sketch. Indeed, the first inequality~\eqref{eq:chainruleshannonsplitting} allows us to divide the joint entropy $H(X^n|E)$ into a sum of individual contributions, whereas the second inequality~\eqref{eq:chainruleshannonrecombination} provides a lower bound on the joint entropy $H(X_{\cS}|E)$ in terms of its components. We refer to the first application as  a {\em splitting} and the second application as a {\em recombination step}.  For the min-entropy, these two steps are more involved; we do not only split and recombine entropies, but corresponding quantum states, as explained in the next section.

\subsection{Towards a modified chain-rule: Entropy-splitting\label{sec:towardsmodifiedchainrule}}

The subadditivity property~\eqref{eq:shannonentropysubadditivity} is
easily shown to hold for the min-entropy. Similarly, a
recom\-bi\-na\-tion-chain-rule~\eqref{eq:chainruleshannonrecombination}
can be proved for min-entropy. However, the splitting-chain
rule~\eqref{eq:chainruleshannonsplitting} is no longer true for
min-entropies and has to be replaced by a more subtle statement. This
can be seen as a quantum version of the \emph{entropy splitting lemma}
proposed in~\cite{Wullsc07}. It is a major component of our proof and
may be of independent interest.

To state this modified splitting-chain-rule, consider a state $\rho_{ABC}$ with purification $\ket{\Psi_{ABCD}}$. We will construct a decomposition
\begin{align}
\ket{\Psi_{ABCD}}=\sum_{\alpha} \ket{\Psi^\alpha_{ABCD}}\label{eq:entropysplitstatesdecomposition}
\end{align}
of $\ket{\Psi_{ABCD}}$ into mutually orthogonal subnormalised states $\{\ket{\Psi^\alpha_{ABCD}}\}_{\alpha}$ such that 
\begin{align}\label{eq:inequalitysplittingexact} 
\hmin(A|BC)_{\rho^\alpha}+\hmin(B|C)_{\rho^\alpha}\geq \hmin(AB|C)_\rho\ \end{align}
for every $\alpha$. In contrast to~\eqref{eq:chainruleshannonsplitting}, this statement splits the entropy into a sum of individual entropies of states which are {\em different} from the original state $\ket{\Psi_{ABCD}}$. They are, however, directly related to $\ket{\Psi_{ABCD}}$ by~\eqref{eq:entropysplitstatesdecomposition}; we call these states {\em split states}.

For technical reasons, it will be convenient to have a version
of~\eqref{eq:inequalitysplittingexact} which decomposes
$\ket{\Psi_{ABCD}}$ into a fixed number $m\in\mathbb{N}$ of
states. The indices $\alpha$ are then from the set
$[m]:=\{1,\ldots,m\}$, and~\eqref{eq:inequalitysplittingexact} is
replaced by
\begin{align}\label{eq:inequalitysplittingapproximate} 
\hmin(A|BC)_{\rho^\alpha}+\hmin(B|C)_{\rho^\alpha}\geq \hmin(AB|C)_\rho-\frac{\Delta}{m}\qquad\textrm{ for all }\alpha\in [m]\ ,
 \end{align}
 where $\Delta$ is function of $\ket{\Psi_{ABCD}}$ which can be
 bounded in situations of interest. (The exact statement is given as
 Corollary~\ref{cor:split} below.) An important property of the split
 states is that each $\ket{\Psi^\alpha_{ABCD}}$ is the result of
 applying a projection $Q_{AD}^\alpha$ to $\ket{\Psi_{ABCD}}$, where
 $Q_{AD}^\alpha$ only acts non-trivially on systems $A$ and $D$.

\subsection{(Partial) recombination of split states\label{sec:partialrecombination}}
Decomposing a state $\ket{\Psi_{ABC}}$ into a sum of mutually orthogonal states $\{\ket{\Psi^\alpha_{ABC}}\}_{\alpha\in [m]}$ gives us a convenient way of bounding the (smooth) min-entropy of $\ket{\Psi_{ABC}}$. The general procedure is as follows: Suppose for example that our aim is to bound the quantity $H(A|B)_\rho$ from below. 
We will show that if the entropy $H(A|B)_{\rho^\alpha}$ is large for every split state $\ket{\Psi^{\alpha}}$, then the same is true for the quantity $\hmin(A|B)_\rho$ (up to a correction of size $\log m$, see Lemma~\ref{lem:recombinationbasic}  for a precise statement).

We can use this fact to show that a state $\ket{\Psi_{ABC}}$ is close to a state $\ket{\widehat{\Psi}_{ABC}}$ with large min-entropy $H(A|B)_{\widehat{\rho}}$. We start from an arbitrary orthogonal decomposition of $\ket{\Psi_{ABC}}$ of the form~\eqref{eq:entropysplitstatesdecomposition} into $m$ states $\{\ket{\Psi^\alpha}\}_{\alpha\in [m]}$. We then identify a subset $\Gamma(\lambda)\subset [m]$ with the property that
\begin{align*}
\hmin(A|B)_{\rho^\alpha}\geq \lambda\qquad\textrm{ for all }\alpha\in \Gamma(\lambda)\ .
\end{align*}
We define the {\em partially recombined state }
\begin{align*}
\ket{\widehat{\Psi}_{ABC}}=\sum_{\alpha\in \Gamma(\lambda)} \ket{\Psi^\alpha}\ .
\end{align*}
We can show that  $H(A|B)_{\widehat{\rho}}\gtrsim \lambda$ is large. Moreover, since the states $\{\ket{\Psi^\alpha_{ABC}}\}_{\alpha\in [m]}$ are assumed to be orthogonal, we can bound the distance of $\ket{\widehat{\Psi}_{ABC}}$ to the original state $\ket{\Psi_{ABC}}$ by an expression of the form $2\sqrt{1-\omega(\Gamma(\lambda))}$, where $\omega(\Gamma(\lambda))$ is the weight of $\Gamma(\lambda)$ under the probability distribution $\omega(\alpha)=\tr \proj{\Psi^\alpha_{ABC}}$ on $[m]$. In this way, showing that the smooth min-entropy $H^\varepsilon(A|B)_{\rho}$ of $\ket{\Psi_{ABC}}$ is lower bounded by a value $\lambda$ reduces to showing that the corresponding set $\Gamma(\lambda)$ has a large weight under~$\omega$.

\subsection{Putting it together: splitting, sampling and recombining\label{sec:splittingsamplingrecombining}}
Let us now return to our original problem: Given a quantum state $\rho_{X^nE}=\rho_{X_1\cdots X_nE}$ with purification $\ket{\Psi}$, we would like to show that $\hmin^\varepsilon(X_{\cS}|E)_\rho$ is large with high probability over the choice of $\cS\subset [n]$. To illustrate the  required steps in the proof, let us consider a simple example where $n=4$.

The first step is to apply the splitting rule to $\ket{\Psi}$, dividing the joint entropy $\hmin(X^n|E)$ into a contribution from $X_1$ and the remainder. This gives $m$ states $\ket{\Psi^{\alpha_1}}_{\alpha_1\in [m]}$ with the property that for all $\alpha^1\in [m]$, 
\begin{align}
\hmin(X^4|E)_\rho\lesssim \hmin(X_1|\inc{X}{1}E)_{\rho^{\alpha_1}}+\hmin(\inc{X}{1}|E)_{\rho^{\alpha_1}}\ .\label{eq:hsplitfirst}
\end{align}
(Here $\rho^{\alpha_1}=\proj{\Psi^{\alpha_1}}$ denotes the density operator corresponding to $\ket{\Psi^{\alpha_1}}$.) We then apply the splitting-chain-rule to each of these states in order to split $\hmin(\inc{X}{1}|E)_{\rho^{\alpha_1}}=\hmin(X_2X_3X_4|E)_{\rho^{\alpha_1}}$ into the contribution of $X_2$ and the remaining part. This results, for each $\alpha_1\in [m]$, in a collection of states $\{\ket{\Psi^{\alpha_1\alpha_2}}\}_{\alpha_2\in [m]}$ satisfying
\begin{align}
\hmin(\inc{X}{1}|E)_{\rho^{\alpha_1}}\lesssim\hmin(X_2|\inc{X}{2}E)_{\rho^{\alpha_1\alpha_2}}+\hmin(\inc{X}{2}|E)_{\rho^{\alpha_1\alpha_2}}\ .\label{eq:hsplitsecond}
\end{align}
Finally, dividing the last term into contributions from $X_3$ and $X_4$, we get, for each $(\alpha_1,\alpha_2)\in [m]^2$, a family of states $\{\ket{\Psi^{\alpha_1\alpha_2\alpha_3}}\}_{\alpha_3\in [m]}$ such that
\begin{align}
\hmin(\inc{X}{2}|E)_{\rho^{\alpha_1\alpha_2}}\lesssim\hmin(X_3|\inc{X}{3}E)_{\rho^{\alpha_1\alpha_2\alpha_3}}+\hmin(X_4|E)_{\rho^{\alpha_1\alpha_2\alpha_3}}\ .\label{eq:hsplitthird}
\end{align}
This completes the splitting step. Summarising, we have obtained a collection of states starting from $\ket{\Psi}$: Those states $\{\ket{\Psi^{\alpha_1}}\}_{\alpha_1\in [m]}$ obtained by applying the splitting-chain-rule once, the states $\{\ket{\Psi^{\alpha_1\alpha_2}}\}_{\alpha_1\alpha_2\in [m]^2}$ corresponding to states that are the result of splitting twice and so on. 

A useful geometric visualisation (which is, however, not essential for the proof) is obtained by placing these states at the vertices of an $m$-ary tree (in this case of depth $3$). We place $\ket{\Psi}$ at the root, and the descendants of each vertex are the split states obtained by splitting. Thus every $3$-tuple $(\alpha_1,\alpha_2,\alpha_3)\in [m]^3$ specifies a path with vertex labels $(\ket{\Psi},\ket{\Psi^{\alpha_1}},\ket{\Psi^{\alpha_1\alpha_2}},\ket{\Psi^{\alpha_1\alpha_2\alpha_3}})$ from the root to a leaf.

 Let us combine inequalities~\eqref{eq:hsplitfirst}--\eqref{eq:hsplitthird} into
\begin{align*}
\hmin(X^4|E)_\rho&\lesssim \hmin(X_1|\inc{X}{1}E)_{\rho^{\alpha_1}}+\hmin(X_2|\inc{X}{2}E)_{\rho^{\alpha_1\alpha_2}}\\
&\ \ \ +\hmin(X_3|\inc{X}{3}E)_{\rho^{\alpha_1\alpha_2\alpha_3}}+\hmin(X_4|E)_{\rho^{\alpha_1\alpha_2\alpha_3}}\qquad\textrm{for all }\alpha^3=(\alpha_1,\alpha_2,\alpha_3)\in [m]^3\ .
\end{align*}
 By attaching the entropies of interest to the edges of the mentioned tree, we can interpret this inequality as expressing the fact that the sum of the values of the edges along each path of the tree from the root to a leaf is lower bounded by $\hmin(X^4|E)_\rho$.

The next things to consider are the sampling- and recombination step. Our aim is to show that the smooth entropy $\hmin^\varepsilon(X_{\cS}|E)_{\rho}$ is large (with high probability over the choice of the subset $\cS\subset [4]$). We follow the procedure outlined in the previous section. That is, we define the recombined state
\begin{align*}
\ket{\widehat{\Psi}}&=\sum_{\alpha^3\in\Gamma(\lambda,\cS)} \ket{\Psi^{\alpha^3}}
\end{align*}
where $\Gamma(\lambda,\cS)$ is the set of paths $\alpha^3\in [m]^3$ with the property that 
\begin{align*}
\delta_{1\in\cS}\hmin(X_1|\inc{X}{1}E)_{\rho^{\alpha_1}}+\delta_{2\in\cS}\hmin(X_2|\inc{X}{2}E)_{\rho^{\alpha_1\alpha_2}}&\\
\qquad +\delta_{3\in\cS}\hmin(X_3|\inc{X}{3}E)_{\rho^{\alpha_1\alpha_2\alpha_3}}+\delta_{4\in\cS}\hmin(X_4|E)_{\rho^{\alpha_1\alpha_2\alpha_3}}&\geq \lambda\ .
\end{align*}
In other words, we restrict our attention to paths (and corresponding states) which (when restricted to $\cS$), have large entropy. We then need to show the following:
\begin{enumerate}[(i)]
\item\label{it:samplercorrectness}
with high probability over the choice~$\cS$, the state $\ket{\widehat{\Psi}}$  is close to $\ket{\Psi}$
\item\label{it:entropypsihat}
the entropy $H(X_{\cS}|E)_{\hat{\rho}}$ is large.
\end{enumerate}
The proof of~\eqref{it:samplercorrectness} again involves  a bound of the form
\[
\frac{1}{2}\big\|\proj{\Psi}-\proj{\widehat{\Psi}}\big\|\leq \sqrt{1-\omega(\Gamma(\lambda,\cS))}\ ,
\]
where $\omega(\alpha^3)=\tr\proj{\Psi^{\alpha^3}}$, $\alpha^3\in [m]^3$ is a (fixed) probability distribution on the leaves. We will show that a sampler has the following property, when applied to the situation described above (see Section~\ref{sec:averagingsamplersmatrix}): With high probability over the choice of $\cS$, the weight $\omega(\Gamma(\lambda,\cS))$ is large. More generally, we show how the sampler-property extends from a single sequence of values to the case of a matrix of values (in our case corresponding to edges of a tree).

The proof of~\eqref{it:entropypsihat} is done inductively using subadditivity, the recombination-chain-rule, and the recombination argument outlined above. For concreteness, suppose for example that $\cS=\{2,4\}$. Then we have
\begin{align*}
\hmin(X_2|\inc{X}{2}E)_{\rho^{\alpha_1\alpha_2}}+\hmin(X_4|E)_{\rho^{\alpha_1\alpha_2\alpha_3}}\geq \lambda\ 
\end{align*}
for all $\alpha^3=(\alpha_1,\alpha_2,\alpha_3)\in\Gamma(\lambda,\cS)\subset [m]^3$. It is convenient to rephrase this as follows, writing $\alpha^3=(\alpha^2,\alpha_3)$. We then have for all $\alpha^2\in [m]^2$
\begin{align*}
\hmin(X_4|E)_{\rho^{(\alpha^2,\alpha_3)}}\geq \lambda-\hmin(X_2|\inc{X}{2}E)_{\rho^{\alpha^2}}\qquad\textrm{ for all }\alpha_3\textrm{ with }(\alpha^2,\alpha_3)\in\Gamma(\lambda,\cS)\ .
\end{align*}
In particular, when we apply this to the (intermediate) partially recombined states
\begin{align}\label{eq:intermpartrecomb}
\ket{\widehat{\Psi}^{\alpha^2}}=\sum_{\alpha_3:(\alpha^2,\alpha_3)\in\Gamma(\lambda,\cS)} \ket{\Psi^{(\alpha^2,\alpha_3)}}\ ,
\end{align}
we obtain
\begin{align*}
\hmin(X_4|E)_{\widehat{\rho}^{\alpha^2}}\gtrsim \lambda-\hmin(X_2|\inc{X}{2}E)_{\rho^{\alpha^2}}\qquad\textrm{ for all }\alpha^2\ .
\end{align*}
We will also use the fact that the recombined states satisfy $\hmin(X_2|\inc{X}{2}E)_{\widehat{\rho}^{\alpha^2}}\geq\hmin(X_2|\inc{X}{2}E)_{\rho^{\alpha^2}}$ (see Lemma~\ref{lem:subsetgammastates}\eqref{eq:rhobentropy}).
Subadditivity gives $\hmin(X_2|X_4E)_{\widehat{\rho}^{\alpha^2}}\geq \hmin(X_2|\inc{X}{2}E)_{\widehat{\rho}^{\alpha^2}}$ for all $\alpha^2\in [m]$. With the previous two inequalities, we therefore get
\begin{align*}
\hmin(X_4|E)_{\widehat{\rho}^{\alpha^2}}+\hmin(X_2|X_4E)_{\widehat{\rho}^{\alpha^2}}\gtrsim\lambda\ .
\end{align*}
This in turn implies
\begin{align*}
\hmin(X_2X_4|E)_{\widehat{\rho}^{\alpha^2}}\gtrsim \lambda\qquad\textrm{for all }\alpha^2\in [m]^2
\end{align*}
by the recombination-chain-rule. 
Because $\ket{\bPsi}$ can be written as sum of the states~\eqref{eq:intermpartrecomb}, the recombination-procedure then gives
\begin{align*}
\hmin(X_2X_4|E)_{\widehat{\rho}}\gtrsim\lambda\ ,
\end{align*}
as claimed.

This line of argument can be followed more generally for a general subset $\cS\subset [n]$. We will need intermediate (partially) recombined states $\{\ket{\widehat{\Psi}^{\alpha^j}}\}_{\alpha^j\in [m]^j}, j\in [n]$; these can again be thought of as being attached to the vertices of a tree. They are defined recursively, by recombining ``good'' states (i.e., those corresponding to prefixes of elements in $\Gamma(\lambda,\cS)$). In other words, when recombining, we work our way up the tree (omitting ``bad'' states, i.e., those with small entropies.)

This concludes our sketch proof; it is now time to elaborate on the details.

\section{Rules and tools for min-entropy\label{sec:rulesandtoolsentropy}}
In this section, we set the ground for our result concerning samplers. In particular, we formally introduce the conditional min-entropy $\hmin(A|B)_{\rho}$ in  Section~\ref{sec:basicdefinitionsh}. This will be done via an intermediate quantity $H(A|B)_{\fr{\rho}{\sigma}}$.  Most of our rules for min-entropy, the most basic of which are stated in Section~\ref{sec:basicrules}, apply to these intermediate quantities; they will be our main object of study. In Section~\ref{sec:entropysplitting}, we establish our central splitting-chain-rule.

\subsection{Preliminaries\label{sec:preliminaries}}

Throughout, we consider nonnegative operators acting on finite-dimensional Hilbert spaces (or {\em systems}) $\cH_A,\cH_B,\ldots$ and their tensor products. We use subscripts to indicate which systems an operator acts on. We also use subscripts when we trace out systems, but sometimes make use of superscripts to denote ``tracing out everything but'', in the
following sense: for a tripartite state $\rho_{ABC}$, we write
$\tr_{BC}(\rho_{ABC})=\tr_{\overline{A}}(\rho_{ABC})=\rho_A$ for the reduced density operator on $A$. As explained above, we sometimes abuse notation by omitting identities. For example, we will write operator inequalities such as
\begin{align*}
\rho_{AB}\leq \sigma_B\ ,
\end{align*}
for a bipartite operator $\rho_{AB}$ and an operator $\sigma_B$ on $\cH_B$. By this inequality, we simply mean $\rho_{AB}\leq\id_{A}\otimes\sigma_B$ (which is defined by the condition that $\id_A\otimes\sigma_B-\rho_{AB}$ is a nonnegative operator). More generally, when writing operators on multipartite systems, we omit identities whenever a unique meaningful statement can be obtained by tensoring corresponding identities to the operators. To give an example, we will write expressions such as
\begin{align*}
Q_B\rho_{AB} Q_B\leq  P_D\rho_{ABCD}P_D\ ,
\end{align*}
where the operators act on the spaces indicated by subscripts, instead of
\begin{align*}
(\id_{A}\otimes Q_B\otimes\id_{CD})(\rho_{AB}\otimes\id_{CD})(\id_{A}\otimes Q_B\otimes\id_{CD})\leq (\id_{ABC}\otimes P_D)\rho_{ABCD}(\id_{ABC}\otimes P_D)\ .
\end{align*}

Basic properties of operator inequalities we need are their preservation under partial traces and the application of operators, i.e., the  fact that $\rho_{AB}\leq \sigma_{AB}$ implies that 
\begin{align*}
\rho_{A}\leq \sigma_{A}\
\end{align*}
and 
\begin{align*}
T_{AB}\rho_{AB}T^\dagger_{AB}\leq T_{AB}\sigma_{AB}T^\dagger_{AB}
\end{align*}
for any operator $T_{AB}$ on $\cH_A\otimes\cH_B$.

For two operators $\rho_{AB}$ on $\cH_A\otimes\cH_B$ and $\sigma_B$ on $\cH_B$ such that the support of $\rho_B$ is contained in the support of $\sigma_B$, the {\em conditional operator} $\fr{\rho_{AB}}{\sigma_B}$ is defined as\footnote{Note that for $\sigma_B=\rho_B$, definition~\eqref{eq:deifnitionconditionaloperator} coincides with the conditional operator $\rho_{A|B}=\fr{\rho_{AB}}{\rho_B}$ discussed, e.g., in~\cite{leifer07}.}
\begin{align}
  \fr{\rho_{AB}}{\sigma_{B}}:=  \sigma_B^\mhalf\rho_{AB}\sigma_B^\mhalf \ .\label{eq:deifnitionconditionaloperator}
\end{align}\
Here $\sigma_B^\mhalf := \sqrt{\sigma_B^{-1}}$, where $\sigma_B^{-1}$ is the {\em generalised inverse}\footnote{The \emph{generalised inverse} $\sigma^{-1}$ of an operator $\sigma$ is defined as the operator which has the same eigenspaces as $\sigma$ with zero eigenvalue on the null eigenspace of $\sigma$ and eigenvalues $\lambda^{-1}$ on the eigenspace of $\sigma$ corresponding to the eigenvalue $\lambda > 0$.} of $\sigma_B$. An important property of conditional operators is that 
\begin{align}\label{eq:divpartialtraceconsistency}
  \frac{\rho_{AB}}{\sigma_B}=\tr_C\bigl(\frac{\rho_{ABC}}{\sigma_B}\bigr)\ 
\end{align}
for any tripartite operator $\rho_{ABC}$.

We will say that a bipartite operator $\rho_{AE}$ on $\cH_A\otimes
\cH_E$ is {\em classical on $A$ (relative to an orthonormal basis
  $\{\ket{a}\}_a$ of $\cH_A$)} if it has the form $\rho_{AE}=\sum_a
\proj{a}_A\otimes \rho^a_E$. Clearly, if $\rho_{AE}$ is classical on
$A$ relative to $\{\ket{a}\}_a$, then so is
$\rho'_{AE}=O_{AE}\rho_{AE}O_{AE}^\dagger$, for any operator of the
form $O_{AE}=\sum_{a}\proj{a}\otimes O^a_E$. It is easy to verify that
this statement is still true when considering purifications and
additional classical systems: If $\ket{\Psi_{ABEF}}$ is such that the
reduced density operator $\rho_{ABE}$ is classical on both $A$ and $B$
(relative to some orthonormal bases) then the same is
true\footnote{This can be seen by decomposing the state as
  $\ket{\Psi_{ABEF}}=\sum_{a,b}\ket{a}\ket{b}\ket{\varphi^{a,b}_{EF}}$. Classicality
  of the state $\rho_{ABE}$ on $A$ and $B$ then implies that
  $\tr_{F}(\ket{\varphi^{a,b}_{EF}}\bra{\varphi^{a',b'}_{EF}})=0$
  whenever $(a,b)\neq (a',b')$. The claim can then be deduced from the
  fact that
  $\tr_{F}(O^{a}_E\ket{\varphi^{a,b}_{EF}}\bra{\varphi^{a',b'}_{EF}}(O^{a'}_E)^\dagger)=O^{a}_E\tr_{F}(\ket{\varphi^{a,b}_{EF}}\bra{\varphi^{a',b'}_{EF}})(O^{a'}_E)^\dagger$.}
for the state $O_{AE}\ket{\Psi_{ABEF}}$.

\subsection{Definition of min-entropy\label{sec:basicdefinitionsh}}

As already mentioned, every pair of operators $\rho_{AB}$ on $\cH_A\otimes\cH_B$ and $\sigma_B$ on $\cH_B$ such that the support of $\rho_B$ is contained in the support of $\sigma_B$ give rise to a conditional operator $\fr{\rho_{AB}}{\sigma_B}$.\footnote{In the following, we will always assume that the support of $\rho_B$ is contained in the support of $\sigma_B$, such that the conditional operator is well defined.} We define the quantity  $H(A|B)_{\fr{\rho}{\sigma}}$  as minus the logarithm\footnote{All logarithms $\log$ are binary; natural logarithms will be denoted by $\ln$.} of the maximal eigenvalue of this conditional operator, that is
\begin{align*}
\h{A}{B}{}{\rho}{\sigma}:=-\log \lambda_{\max} (\fr{\rho_{AB}}{\sigma_B})\ .
\end{align*}
In some sense, this can be read as ``the entropy of $\rho_A$ when it is conditioned on $\sigma_B$''; in the case where $A$ is classical, the operator $\sigma_B$ is related to a measurement on $\cH_B$ (which is supposed to reproduce the value on $A$, see~\cite{KoReSc08}).

Maximising this quantity over all nonnegative trace-one operators $\sigma_B$ whose support contains the support of $\rho_B$ gives the {\em min-entropy of $A$ given $B$}, defined as\footnote{In Section~\ref{sec:extracionpriorqinfo}, the quantity $\hmin(A|B)_\rho$ was introduced without explicit reference to the intermediate quantities $H(A|B)_{\fr{\rho}{\sigma}}$.} 
\begin{align}\label{eq:minentropyconvdef}
\hmin(A|B)_\rho:=\sup_{\sigma_B} H(A|B)_{\fr{\rho}{\sigma}}\ .
\end{align}
This quantity has a simple operational interpretation, as will be shown in a forthcoming publication~\cite{KoReSc08}: it is equivalent to the maximal probability of guessing $A$ given $B$, in the case where $A$ is classical.

While~\eqref{eq:minentropyconvdef} is ultimately the quantity of interest, the intermediate quantities $H(A|B)_{\fr{\rho}{\sigma}}$ are easier to manipulate, and satisfy various useful rules. As we will see below, most of these follow more or less directly from the alternative characterisation
\begin{align}
  \h{A}{B}{}{\rho}{\sigma}&:=-\log\min \{\lambda\ :\ \frac{\rho_{AB}}{\sigma_B} \leq \lambda \cdot \id_{AB}\}\label{eq:operatorinequalityhmin}
\end{align}
in terms of a family of operator inequalities.

For consistency reasons, it is convenient to set
\begin{align*}
H(A|B)_{\rho}&:=H(A|B)_{\fr{\rho}{\rho}}\\
\hcond{B}{\rho}{\sigma}&:=-\log\min \{\lambda\ :\ \frac{\rho_{B}}{\sigma_B} \leq \lambda \cdot \id \}\ ,
\end{align*}
where $\fr{\rho_B}{\sigma_B}=\sigma_B^\mhalf\rho_B\sigma_B^\mhalf$. Note that the latter quantity is equal to $H(A|B)_{\fr{\rho}{\sigma}}$ if the Hilbert space $\cH_A$ corresponding to system $A$ is trivial, i.e., $\cH_A\cong\mathbb{C}$. We can think of the quantity $H(B)_{\fr{\rho}{\sigma}}$ as  a conditional entropy obtained by adjoining a trivial system to $B$ using the isomorphism $\cH_B\cong \cH_B\otimes\mathbb{C}$. Informally, this corresponds to a situation where we condition ``nothing'' on $\sigma_B$; formally, it will turn out to be convenient to define $H(\emptyset|B)_{\fr{\rho}{\sigma}}:=\hcond{B}{\rho}{\sigma}$.

Finally, we will also (formally) encounter situations where $\rho_{AB}=0$; in these cases, we formally set $H(A|B)_{\rho}=\infty$, $H(A|B)_{\fr{\rho}{\sigma}}=\infty$, meaning that an arbitrarily large value can be assigned to these quantities in any identity where they appear.

For a parameter $\varepsilon \geq 0$, the $\varepsilon$-smooth min-entropy of $A$ given $B$ is equal to (cf.~\cite{Ren05}) 
\begin{align*}
\hmin^\varepsilon(A|B)_{\rho}=\sup_{\substack{\bar{\rho}_{AB}: \|\bar{\rho}_{AB}-\rho_{AB}\|\leq \varepsilon\\
\tr(\bar{\rho}_{AB})\leq 1}}\hmin(A|B)_{\bar{\rho}}\ , 
\end{align*}
where the supremum is over all nonnegative operators $\bar{\rho}_{AB}$ with trace bounded by~$1$ in an $\varepsilon$-ball around $\rho_{AB}$. (Here $\|A\|=\tr\sqrt{A^\dagger A}$ is the $L_1$-norm.) 

As already mentioned, the (smooth) entropy $\hmin^\varepsilon(X|E)$ captures the 
number of secret bits extractable from $X$ with respect to an adversary holding~$E$. The following lemma justifies this operational interpretation.
\begin{lemma}\label{lem:hminextractablekeylength}
Consider a state $\rho_{XE}$ where $X$ is classical. 
Let $\rho_{\uniform_{\sbin^\ell}}$ denote the completely mixed state on $\sbin^\ell$.
Then
\begin{enumerate}[(i)]
\item\label{it:leftover} For any $\ell\leq
  \hmin(X|E)-2\log\textfrac{1}{\varepsilon}$, there is a function 
  $f:\cX\times\cX\rightarrow\sbin^\ell$ (independent of $\rho_{X E}$) which extracts an $\ell$-bit
  string $Z=f(X,Y)$ from $X$, such that $Z$ is $\varepsilon$-close to
  uniform and independent of $(E,Y)$, where $Y$ is a uniform and
  independent seed. In formulae, we have
\begin{align*}
\frac{1}{2}\|\rho_{f(X,Y)YE}-\rho_{\uniform_{\sbin^\ell}}\otimes \rho_{Y}\otimes\rho_E\|\leq \varepsilon\ .
\end{align*}
\item\label{it:optimalityleftover}
For any function $f:\cX\rightarrow\sbin^\ell$ and $\varepsilon\geq 0$, the inequality
  \[
  \frac{1}{2}\bigl\| \rho_{f(X) E} - \rho_{\uniform_{\sbin^\ell}} \otimes \rho_E \bigr\| \leq \varepsilon
  \]
  implies
  \[
    \hmin^{2\varepsilon}(X|E)_{\rho} \geq \ell \ .
  \]
\end{enumerate}
\end{lemma}
\begin{proof}
Statement~\eqref{it:leftover} is a reformulation of the fact that the two-universal hashing construction is a quantum extractor, as shown by Renner~\cite{Ren05}.

For the proof of~\eqref{it:optimalityleftover}, let $\bar{\rho}_{S E} := \rho_{\uniform_{\sbin^\ell}} \otimes \rho_E$. Then, obviously 
  \[
    \hmin(S|E)_{\bar{\rho}} \geq H(S|E)_{\frac{\bar{\rho}}{\bar{\rho}}} = \ell \ .
  \]
  Because $\frac{1}{2}\| \bar{\rho}_{S E} - \rho_{f(X) E} \| \leq \varepsilon$, this implies
  \[
    \hmin^{2\varepsilon}(f(X)|E)_{\rho} \geq \hmin(S|E)_{\bar{\rho}} \geq \ell \ .
  \]
  Since the min-entropy can only decrease when applying a function (see~\cite{Ren05}), we conclude that
  \[
    \hmin^{2\varepsilon}(X | E)_{\rho} \geq \hmin^{2\varepsilon}(f(X) | E)_{\rho}  \geq \ell \ ,
  \]
as desired.
\end{proof}

\subsection{Some basic rules and properties\label{sec:basicrules}}
We now summarise a few basic rules for the quantities $\h{A}{B}{}{\rho}{\sigma}$ which directly follow from~\eqref{eq:operatorinequalityhmin} using standard properties of operator inequalities, as described in Section~\ref{sec:preliminaries}.

\begin{lemma}[Properties of min-entropy]\label{lem:hminprop}
The min-entropy satisfies the following.
\begin{enumerate}[(i)]
\item {\bf (Positivity for classical systems)}\label{it:classicalcond}
Let $\rho_{AB}=\sum_{a} \proj{a}\otimes\rho_B^a$ be classical on $A$. Then
$H(A|B)_\rho\geq 0$.
\item {\bf (Dimension bound)}\label{it:hzerobound}
For any $\rho_{AB}$ and $\sigma_B$ with $\sigma_B\leq \rho_B$, we have
$H(A|B)_{\fr{\rho}{\sigma}}\leq H_0(A)$, where $H_0(A)=\log \dim \cH_A$. In particular, $H(A|B)_{\rho}\leq H_0(A)$.
 More generally $H(B|C)_{\fr{\rho}{\sigma}}\geq  H(AB|C)_{\fr{\rho}{\sigma}}-H_0(A)$
for any  $\rho_{ABC}$ and $\sigma_{C}$.

\item {\bf (Subadditivity)}\label{it:simplesubadditivity}
  $H(A|B)_{\fr{\rho}{\sigma}}\geq H(A|B C)_{\fr{\rho}{\sigma}}$ for  any $\rho_{A B C}$ and $\sigma_{B C}$.

\item {\bf (Recombination-chain-rule)}\label{it:chainrule}
  $H(AB|C)_{\fr{\rho}{\sigma}}\geq  H(A|B
  C)_{\rho}+H(B|C)_{\fr{\rho}{\sigma}}$ for any $\rho_{A B C}$ and $\sigma_C$.
\end{enumerate}
\end{lemma}
\begin{proof}
\eqref{it:classicalcond} directly follows from
$\proj{a}\otimes\rho_B^a\leq \rho_B$ for all $a$.

For the proof of the first part of~\eqref{it:hzerobound}, we simply take the trace on both sides of the inequality $\rho_{AB}\leq 2^{-H(A|B)_{\fr{\rho}{\sigma}}}\sigma_B$ to get $\tr(\rho_{AB})\leq 2^{H_0(A)-H(A|B)_{\fr{\rho}{\sigma}}}\tr(\sigma_B)$, which gives the claim because $\tr(\sigma_B)\leq \tr(\rho_B)=\tr(\rho_{AB})$. For the proof of the second part of~\eqref{it:hzerobound}, observe that we have
$\rho_{ABC}\leq 2^{-H(AB|C)_{\fr{\rho}{\sigma}}}\sigma_C$ by definition. Tracing out the system $A$ gives
\begin{align*}
\rho_{BC}\leq 2^{-H(AB|C)_{\fr{\rho}{\sigma}}+H_0(A)}\sigma_C\ .
\end{align*}
The claim~\eqref{it:hzerobound} then follows from the definition of $H(B|C)_{\fr{\rho}{\sigma}}$.

Similarly,~\eqref{it:simplesubadditivity} directly follows by tracing
out $C$ from the inequality   \[
    \rho_{A B C}\leq 2^{-H(A|B C)_{\fr{\rho}{\sigma}}}\sigma_{B C}\ .
  \]

For the proof of~\eqref{it:chainrule}, observe that 
 \begin{align*}
     \rho_{ABC}&\leq 2^{-H(A|B C)_{\rho}}\rho_{B C}
  \leq 
    2^{-H(A|B C)_{\rho}-H(B|C)_{\fr{\rho}{\sigma}}} \sigma_C \ .
  \end{align*}
The claim follows from the definition of $H(AB|C)_{\fr{\rho}{\sigma}}$.
\end{proof}
We point out  that~\eqref{it:hzerobound} and~\eqref{it:simplesubadditivity} directly translate into the statements
\begin{align*}
\hmin(B|C)_\rho&\geq \hmin(AB|C)_\rho-H_0(A)\\
\hmin(A|B)_{\rho}&\geq \hmin(A|BC)_{\rho}
\end{align*}
for the min-entropy. An analogous statement cannot be made for the recombination-chain-rule~\eqref{it:chainrule}, and we will have to retain the dependence on $\sigma_C$ in our arguments.

Having established subadditivity and a recombination-chain-rule, we will address the problem of finding a converse splitting-chain-rule in the next section. Before doing so, however, we will mention another property of the min-entropy which will be important for our purposes. This is the fact the entropy of a state $\ket{\Psi'}=Q\ket{\Psi}$ obtained by applying a projection to a state $\ket{\Psi}$ is lower bounded by the entropy of the original state. We will later see that this allows us to retain information about the entropy  when going from a state to its split descendants. 

Note that this statement is not generally true, but depends crucially on where the projection acts.
\begin{lemma}[Monotony under
  projections]\label{lem:projectionmonotony}
Let $\ket{\Psi_{ABC}}$ be a pure state, let $Q_C$ be an operator on~$C$ and
let $\ket{\Psi_{ABC}'}=Q_C\ket{\Psi_{ABC}}$.  Let $\rho_{ABC}$ and
$\rho'_{ABC}$ be the corresponding density operators. 
Then 
\begin{align*}
H(A|BC)_{\rho'}\geq H(A|BC)_\rho\ .
\end{align*}
Furthermore, if $Q_C$ is a projector, then 
\begin{align*}
H(A|B)_{\fr{\rho'}{\sigma}}\geq H(A|B)_{\fr{\rho}{\sigma}}\qquad\textrm{ and }\qquad
\hcond{B}{\rho'}{\sigma}\geq \hcond{B}{\rho}{\sigma}
\end{align*}
for arbitrary $\sigma=\sigma_C$.
\end{lemma}
\begin{proof}
To prove the first inequality, let $Q_C$ be arbitrary. Applying $Q_C$
from the left and $Q_C^\dagger$ from the right to both sides of the inequality
\begin{align*}
\rho_{ABC}\leq 2^{-H(A|BC)_{\rho}}\rho_{BC}\ 
\end{align*}
gives
\begin{align*}
\rho'_{ABC}\leq  2^{-H(A|BC)_{\rho}}\rho'_{BC}
\end{align*}
by definition of $\ket{\Psi'_{ABC}}$ and the properties of the partial
trace. This proves the first inequality.

Let now $Q_C$ be a projector, and let $\ket{\varphi_{AB}}\in AB$ be arbitrary. Then 
\begin{align*}
\tr(\proj{\varphi_{AB}}\rho'_{AB})&=\tr\left((\proj{\varphi_{AB}}\otimes\id_C)\rho'_{ABC}\right)\\
&=\tr\left((\proj{\varphi_{AB}}\otimes Q_C)\rho_{ABC}\right)\ 
\end{align*}
by the cyclicity of the trace and the fact that $Q_C$ is a projector.
In particular, with $Q_C^\bot=\id_C-Q_C$ denoting the projector onto the
orthogonal complement of the image of $Q_C$, we have
\begin{align*}
\tr\left(\proj{\varphi_{AB}}(\rho_{AB}-\rho'_{AB})\right)&=\tr\left((\proj{\varphi_{AB}}\otimes
  Q_C^\bot)\rho_{ABC}\right)\geq 0\ .
\end{align*}
We conclude that $\rho_{AB}'\leq \rho_{AB}$. In particular, 
\begin{align*}
\rho_{AB}'\leq \rho_{AB}\leq 2^{-H(A|B)_{\fr{\rho}{\sigma}}}\sigma_B\qquad \textrm{ and }\qquad \rho'_B\leq\rho_B\leq 2^{-\hcond{B}{\rho}{\sigma}}\sigma_B
\end{align*}
which implies the claim.
\end{proof}

\subsection{Entropy-splitting: A splitting-chain-rule for min-entropy\label{sec:entropysplitting}}
To introduce our splitting-chain-rule, we proceed  in two steps: In Section~\ref{sec:warmup}, we show a simplified version which does not restrict the number of states the original state is split into. As this is irrelevant for the remainder of our proof, this section can be skipped; however, it nicely illustrates the relevant features. The case of interest, where we split a given state into a fixed number $m$ of states, can be seen as a coarse-graining of the former. It will be the topic of Section~\ref{sec:coarsegraining}.

\subsubsection{A warm-up\label{sec:warmup}}
The chain-rule we will prove in this section concerns a tripartite state $\rho_{ABC}$ with purification $\rho_{ABCD}=\proj{\Psi_{ABCD}}$ and an operator $\sigma_C$. We will show that we can split $\ket{\Psi_{ABCD}}$ into a sum of states $\{\ket{\Psi^\alpha_{ABCD}}\}_{\alpha}$ as in~\eqref{eq:entropysplitstatesdecomposition}, in a way that
\begin{align}
H(A|BC)_{\rho^{\alpha}}+H(B|C)_{\fr{\rho^\alpha}{\sigma}}\geq H(AB|C)_{\fr{\rho}{\sigma}}\qquad\textrm{ for all }\alpha\ ,\label{eq:exactsumcondoperators}
\end{align}
where $\rho^\alpha_{ABCD}=\proj{\Psi^\alpha_{ABCD}}$. Note that by taking the supremum over $\sigma_B$, we immediately obtain the inequality
\begin{align*}
\hmin(A|BC)_{\rho^{\alpha}}+\hmin(B|C)_{\rho^\alpha}\geq \hmin(AB|C)_{\rho}\qquad\textrm{ for all }\alpha\ 
\end{align*}
from~\eqref{eq:exactsumcondoperators}. However,~\eqref{eq:exactsumcondoperators} makes a stronger assertion, and we will generally deal with statements of this form.

For the proof of~\eqref{eq:exactsumcondoperators}, consider the eigendecomposition
\begin{align*}
\fr{\rho_{BC}}{\sigma_C}=\sum_{\alpha} \alpha P^\alpha_{BC}\ 
\end{align*}of the conditional operator,
where $P^\alpha_{BC}$ is the projector onto the eigenspace corresponding to the eigenvalue $\alpha$. 

We will use the operators $P^\alpha_{BC}$ to define our split states,
which will be labeled by the spectrum of
$\fr{\rho_{BC}}{\sigma_C}$. Clearly, if we apply $P^\alpha_{BC}$ on
both sides of the operator $\fr{\rho_{BC}}{\sigma_C}$, we end up with
an operator which has a single non-zero eigenvalue $\alpha$. While
$P^\alpha\fr{\rho_{BC}}{\sigma_C}P^\alpha$ thus has a very simple
form, it is very different in nature from the original
``unconditional'' operator $\rho_{BC}$. Intuitively, it therefore
makes sense to multiply by $\sigma_C$. The appropriate definition of
$\ket{\Psi^\alpha_{ABCD}}$ turns out to be just the result of this,
i.e., we can define\footnote{As above, we assume that the support of
  $\sigma_C$ contains the support of $\rho_C$, hence, $\sigma_C$ is
  invertible on the relevant subspace.}
\begin{align*}
\ket{\Psi^\alpha_{ABCD}}=\sigma_C^\half P^\alpha_{BC}\sigma_C^\mhalf\ket{\Psi_{ABCD}}\ .
\end{align*}
It is easy to check that these states decompose $\ket{\Psi_{ABCD}}$ as in~\eqref{eq:entropysplitstatesdecomposition}. They also satisfy~\eqref{eq:exactsumcondoperators}, as we will show now. First observe that
$\rho^\alpha_{BC}=\sigma_C^\half P^\alpha_{BC}\fr{\rho_{BC}}{\sigma_C}P^\alpha_{BC}\sigma_C^\half$ by their very definition, and thus 
\begin{align*}
\rho^\alpha_{BC}=\alpha\cdot \sigma_C^\half P^\alpha_{BC}\sigma_C^\half\ .
\end{align*} 
In particular,
we have 
\begin{align}\label{eq:rhoalphabcbnd}
\fr{\rho^\alpha_{BC}}{\sigma_C}=\alpha P^\alpha_{BC}\ ,
\end{align}
which implies that
\begin{align}\label{eq:hminbcbnd}
\hmin(B|C)_{\fr{\rho^\alpha}{\sigma}}=-\log\alpha\ .
\end{align}
Combining~\eqref{eq:rhoalphabcbnd} and~\eqref{eq:hminbcbnd} gives the statement 
\begin{align}\label{eq:PalphaBCbound}
P^\alpha_{BC}\leq 2^{\hmin(B|C)_{\fr{\rho^\alpha}{\sigma}}} \fr{\rho^\alpha_{BC}}{\sigma_C}\ .
\end{align}

By definition of the quantity $\h{AB}{C}{}{\rho}{\sigma}$, we also have
$\fr{\rho_{ABC}}{\sigma_C}\leq 2^{-\h{AB}{C}{}{\rho}{\sigma}}$. Applying the projector $P^\alpha_{BC}$ on both sides of this inequality leads to
\begin{align*}
P^\alpha_{BC}\fr{\rho_{ABC}}{\sigma_C}P^\alpha_{BC}\leq 2^{-\h{AB}{C}{}{\rho}{\sigma}} P^\alpha_{BC}\leq 2^{\hmin(B|C)_{\fr{\rho^\alpha}{\sigma}}-\h{AB}{C}{}{\rho}{\sigma}}\fr{\rho^\alpha_{BC}}{\sigma_C}\ .
\end{align*}
Multiplying this inequality from both sides by $\sigma_C^\half$ immediately gives the desired statement~\eqref{eq:exactsumcondoperators}, since 
$\rho^\alpha_{ABC}=\sigma_C^\half P^\alpha_{BC}\fr{\rho_{ABC}}{\sigma_C}P^\alpha_{BC}\sigma_C^\half$.

This concludes the proof of our simplified statement, where a state $\ket{\Psi_{ABCD}}$ is split into a family $\{\ket{\Psi^\alpha_{ABCD}}\}_{\alpha}$, each of which obeys the splitting-chain-rule inequality~\eqref{eq:exactsumcondoperators}. The number of states is determined by the number of different eigenvalues of the operator $\fr{\rho_{BC}}{\sigma_C}$; indeed, each state $\ket{\Psi^\alpha_{ABCD}}$ corresponds to an eigenvalue $\alpha$.

Before continuing, let us show the following useful properties of the states $\{\ket{\Psi^\alpha_{ABCD}}\}_{\alpha}$: They are mutually orthogonal, and each state $\ket{\Psi^\alpha_{ABCD}}=Q^\alpha_{AD}\ket{\Psi}$ is the result of applying a projection $Q^\alpha_{AD}$ (which acts non-trivially only on $A$ and $D$) to $\ket{\Psi_{ABCD}}$. This statement is the result of using the complementarity property that is inherent in quantum states.

First observe that $\sigma_C^\mhalf\ket{\Psi_{ABCD}}$ is a purification of the conditional operator $\fr{\rho_{BC}}{\sigma_C}$. Using the Schmidt-decomposition, we can write
\begin{align*}
\sigma_C^\mhalf\ket{\Psi_{ABCD}}=\sum_\alpha \sqrt{\alpha}\ket{\alpha}_{AD}\ket{\alpha}_{BC}\ ,
\end{align*}
where $\{\ket{\alpha}_{AD}\}$ and $\{\ket{\alpha}_{BC}\}$ are eigenvectors with eigenvalue $\alpha$ of $\fr{\rho_{AD}}{\sigma_C}$ and $\fr{\rho_{BC}}{\sigma_C}$, respectively (slightly abusing notation, we omit multiplicities). We can define $Q^\alpha_{AD}$ as the projector onto the eigenspace of $\fr{\rho_{AD}}{\sigma_C}$ corresponding to the eigenvalue $\alpha$. We then clearly have
\begin{align*}
P^\alpha_{BC}\sigma_C^\mhalf\ket{\Psi_{ABCD}}=Q^\alpha_{AD}\sigma_C^\mhalf\ket{\Psi_{ABCD}}
\end{align*}
for every $\alpha$. Since $\sigma_C^\half$ and $Q^\alpha_{AD}$ act on different systems, they commute, and we obtain
\begin{align*}
\ket{\Psi^\alpha_{ABCD}}=\sigma_C^\half Q^\alpha_{AD}\sigma_C^\mhalf\ket{\Psi_{ABCD}}=Q^\alpha_{AD}\ket{\Psi_{ABCD}}\ ,
\end{align*}
as claimed. The orthogonality of these states is now immediate.

\subsubsection{Splitting into a fixed number of states\label{sec:coarsegraining}}
In this section, we show that the construction discussed in Section~\ref{sec:warmup} can be adapted to yield a fixed number $m$~of states. This is quite straightforward: We simply divide the spectrum of the conditional operator $\fr{\rho_{BC}}{\sigma_C}$ into $m$ different intervals $]\mu_{\alpha-1},\mu_{\alpha}]$, for $\alpha\in [m]$. Instead of projecting onto the eigenspace corresponding to a single eigenvalue, we use projectors $P^\alpha_{BC}$ onto the direct sum of eigenspaces associated with eigenvalues in the corresponding interval.
\begin{lemma}[Entropy splitting] \label{lem:splitfirst}
  Let $\rho_{ABCD}=\proj{\Psi_{ABCD}}$ be a pure state and  let
  $\sigma_C$ be a nonnegative operator. Let   $h_0\leq h_1\leq \cdots \leq h_m$ be an $(m+1)$-tuple of   monotonically increasing real values with minimum and maximum given  by
  \begin{align*}
    h_0 & := H(B|C)_{\frac{\rho}{\sigma}} \\
    h_m & := H(A B| C)_{\frac{\rho}{\sigma}} - H(A | B C)_{\rho} \ .
  \end{align*}
Then there are mutually orthogonal projectors
$\{Q^\alpha_{AD}\}_{\alpha\in [m]}$ with the property that
  \begin{align}
    H(A|B C)_{\rho^\alpha} 
  & \geq  H(A B|C)_{\frac{\rho}{\sigma}} - h_{\alpha} \label{eq:splitABC} 
  \\
    H(B|C)_{\frac{\rho^{\alpha}}{\sigma}} 
  & \geq 
    h_{\alpha-1} \label{eq:splitBC}
  \end{align}
  where $\rho^{\alpha}_{ABCD}=\proj{\Psi^\alpha_{ABCD}}$ is defined as 
  \begin{align}\label{eq:sigmadefinition}
\ket{\Psi_{ABCD}^\alpha}:=Q^\alpha_{AD}\ket{\Psi_{ABCD}}\  .
  \end{align}
An alternative expression for these states is
\begin{align}
\ket{\Psi_{ABCD}^\alpha}:=\sigma_C^\half P^\alpha_{BC}\sigma_C^\mhalf\ket{\Psi_{ABCD}}\  ,\label{eq:sigmadefinitionalternative}
\end{align}
where $\{P^\alpha_{BC}\}_{\alpha\in [m]}$ are mutually orthogonal
projectors. They satisfy 
\begin{align}\label{eq:identityresolutionpsi}
\sum_{\alpha\in [m]}\ket{\Psi^{\alpha}}=\ket{\Psi}
\end{align}
\end{lemma}

  Note that, according to the chain rule (Lemma~\ref{lem:hminprop}~\eqref{it:chainrule}), $h_0 \leq h_m$,
  i.e., there always exists a tuple of reals as defined in the lemma.

\begin{proof}
  The proof of this statement is almost identical to the proof given
  in Section~\ref{sec:warmup}, but given here for completeness.  For
  any $\alpha \in [m-1]$, define $\mu_\alpha := 2^{-h_\alpha}$, and
  let $\mu_{0} := \infty$, $\mu_{m} := -\infty$.  Note that
  $(\mu_{\alpha})_{\alpha=0}^m$ is a monotonically decreasing sequence
  of values.

Consider the  Schmidt-decomposition
\begin{align}\label{eq:schmidtconditionalstate}
\sigma_C^\mhalf\ket{\Psi_{ABCD}}=\sum_\lambda
\sqrt{\lambda}\ket{\lambda}_{BC}\ket{\lambda}_{AD}\ 
\end{align}
of the ``conditional'' state $\sigma_C^\mhalf\ket{\Psi_{ABCD}}$ (the sum may include multiplicities). 
For
every $\alpha\in [m]$, we
define the projectors $P^\alpha_{BC}$ and $Q^\alpha_{AD}$ as
\begin{align*}
P^\alpha_{BC}&=\sum_{\lambda\in
  ]\mu_\alpha,\mu_{\alpha-1}]}\proj{\lambda}_{BC}\\
Q^\alpha_{AD}&=\sum_{\lambda\in
  ]\mu_\alpha,\mu_{\alpha-1}]}\proj{\lambda}_{AD}\ .
\end{align*}
By definition, these operators satisfy~\eqref{eq:sigmadefinition} and~\eqref{eq:sigmadefinitionalternative}. Moreover, using the fact that $Q^\alpha_{AD}$ commutes with
$\sigma_C^\mhalf$, we conclude that~\eqref{eq:sigmadefinition}
and~\eqref{eq:sigmadefinitionalternative} define the same state $\ket{\Psi^\alpha_{ABCD}}$.

Since $P^{\alpha}_{B C}$ is the
  projector onto the eigenspaces of $\fr{\rho_{BC}}{\sigma_C}$ which
  belong to the eigenvalues in $]\mu_\alpha,\mu_{\alpha-1}]$ for every
  $\alpha\in [m]$, we have
  \begin{align}
  \mu_{\alpha}  P^\alpha_{B C}\leq     P^\alpha_{B C} \fr{\rho_{BC}}{\sigma_C} P^\alpha_{B C}\leq
  \mu_{\alpha-1}  P^\alpha_{B C}\ . \label{eq:halphapfirst}
  \end{align}

We show that 
  \begin{align}
P^\alpha_{B C}\fr{\rho_{BC}}{\sigma_C} P^\alpha_{B C}\leq  2^{-h_{\alpha-1}}  P^\alpha_{B C}\qquad\textrm{for all }\alpha\in [m]\ .\label{eq:halphap}
  \end{align}
This follows directly from~\eqref{eq:halphapfirst} for $\alpha\geq 2$ since $2^{-h_{\alpha-1}}=\mu_{\alpha-1}$; for $\alpha=1$, it is a consequence of the fact that the eigenvalues of $\fr{\rho_{BC}}{\sigma_C}$ are upper bounded by $2^{-h_0}$ by definition of $h_0$.

 Claim~\eqref{eq:splitBC} now
directly follows from~\eqref{eq:halphap} and
the fact that
$\fr{\rho^\alpha_{BC}}{\sigma_C}=P^\alpha_{BC}\fr{\rho_{BC}}{\sigma_C}P^\alpha_{BC}$.

Next we show that
 \begin{align}
    P^\alpha_{B C} \frac{\rho_{A B C}}{\sigma_C} P^\alpha_{B C}
  \leq
    2^{-H(A B | C)_{\frac{\rho}{\sigma}}+h_\alpha}  P^\alpha_{B C} \fr{\rho_{B C}}{\sigma_C}  P^\alpha_{B C} \qquad\textrm{ for all }\alpha\in [m]\  .\label{eq:halphapsecond}
  \end{align}
We distinguish two cases: For $\alpha=m$, identity~\eqref{eq:halphapsecond} is equivalent to 
\begin{align*}
    P^\alpha_{B C} \frac{\rho_{A B C}}{\sigma_C} P^\alpha_{B C}
  \leq
    2^{-H(A|BC)_\rho}  P^\alpha_{B C} \fr{\rho_{B C}}{\sigma_C}  P^\alpha_{B C} 
\end{align*}
because of the definition of $h_m$. But this directly follows from $\rho_{ABC}\leq 2^{-H(A|BC)_\rho}\rho_{BC}$ by multiplication from both sides with $P^\alpha_{BC}\sigma_C^{\mhalf}$ and its adjoint.

For $1\leq \alpha< m$, we use the fact that the first inequality of~\eqref{eq:halphapfirst} is equivalent to
\begin{align*}
2^{-h_\alpha}P^\alpha_{BC}\leq     P^\alpha_{B C} \frac{\rho_{B C}}{\sigma_C} P^\alpha_{B C}
\end{align*}
since $2^{-h_\alpha}=\mu_\alpha$. Substituting this into the inequality
\begin{align*}
    P^\alpha_{B C} \frac{\rho_{A B C}}{\sigma_C} P^\alpha_{B C}
  \leq
    2^{-H(A B | C)_{\frac{\rho}{\sigma}}} P^\alpha_{B C} \ .
\end{align*}
(which directly follows from $\frac{\rho_{A B C}}{\sigma_C}  \leq
    2^{-H(A B | C)_{\frac{\rho}{\sigma}}}$) immediately gives~\eqref{eq:halphapsecond} for all  $1\leq \alpha< m$. This concludes the proof of the auxiliary statement~\eqref{eq:halphapsecond}.

The proof of~\eqref{eq:splitABC}  is now straightforward, based on~\eqref{eq:halphapsecond}. Multiplying the latter inequality by $\sigma_C^\half$ from both the left and
  the right yields
  \[
    \rho^\alpha_{A B C} 
  \leq 
    2^{-H(A B|C)_{\frac{\rho}{\sigma}}+h_\alpha} \rho^\alpha_{B C} \ .
  \]
which implies~\eqref{eq:splitABC}.
\end{proof}

In the previous lemma, we did not specify the intervals $]h_{\alpha-1},h_\alpha]$ that are used to partition the spectrum of the conditional operator $\fr{\rho_{BC}}{\sigma_C}$. A simple choice is to partition the spectrum into $m$ intervals of equal length. This results in the following splitting-chain-rule, which will be our basic tool in what follows.

\begin{corollary} \label{cor:split}
Let $\rho_{ABCD}=\proj{\Psi_{ABCD}}$ be a pure state and let  $\sigma_C$ be a nonnegative operator.  Then for any $m\in\bbN$
  \[
    H(A|B C)_{\rho^\alpha} + H(B|C)_{\frac{\rho^{\alpha}}{\sigma}}   
  \geq
    H(A B|C)_{\frac{\rho}{\sigma}} - \frac{\Delta}{m}
  \]
  where
  \[
    \Delta := H(A B| C)_{\frac{\rho}{\sigma}} - H(A | B C)_{\rho} -
    H(B|C)_{\frac{\rho}{\sigma}} \ . 
  \]
and where   $\rho^{\alpha}_{ABCD}=\proj{\Psi^\alpha_{ABCD}}$ is defined by~\eqref{eq:sigmadefinition} or~\eqref{eq:sigmadefinitionalternative} in terms of families of mutually orthogonal projectors $\{Q^\alpha_{AD}\}_{\alpha\in [m]}$ and $\{P^\alpha_{BC}\}_{\alpha\in [m]}$, as in Lemma~\ref{lem:splitfirst}.
\end{corollary}
We are usually able to obtain a bound on $\Delta$; for a comparatively
large value of $m$, we therefore get an approximation
of~\eqref{eq:exactsumcondoperators}, which is a converse to the
recombination-chain-rule (Item~\eqref{it:chainrule} of
Lemma~\ref{lem:hminprop}).
\begin{proof}
Here we choose $h_\alpha=h_0+\alpha\frac{\Delta}{m}$ for all $\alpha\in\{0,\ldots,m\}$.
\end{proof}
We point out that the statement of Corollary~\ref{cor:split} is also valid with $B$ removed from all expressions. This is because we can always adjoin a trivial system~$B$ with Hilbert space $\cH_B\cong\mathbb{C}$.

\begin{figure}
\begin{center}
\centerline{\includegraphics[scale=1.0]{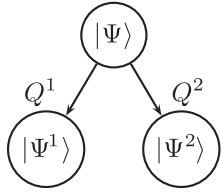}}
\end{center}
\caption{This figure illustrates the basic building block for our arguments (cf.~Corollary~\ref{cor:split}). A state $\ket{\Psi}$ can be decomposed
  into a sum of $m$~orthogonal states; in the figure, $m=2$. We
  illustrate this by a tree; the original state sits at the root,
  whereas the split states sit at nodes labeled by $\alpha\in
  [m]$. The state at the root is the sum of its descendants, which
  are  identical to the leaves in this case. Going from a node
  to its descendants is achieved by applying corresponding projection
  operators. \label{fig:basicsplitting}}
\end{figure}

For later use, we establish a few additional properties of the states $\ket{\Psi^\alpha_{ABCD}}$. We first show that the states $\ket{\Psi^\alpha_{ABCD}}$ have the same classicality properties as the original state $\ket{\Psi_{ABCD}}$.
\begin{remark}[Preservation of classicality properties]\label{rem:classicality}
Suppose that $D=D_1D_2$ is bipartite, and that $\rho_{ABCD_1}$ is classical on $A$, $B$, and $D_1$ (relative to some orthonormal bases of these subsystems). Then $\rho^\alpha_{ABCD_1}$ is classical on $A$, $B$, and $D_1$ (relative to the same bases), for any $\alpha\in [m]$.
\end{remark}
\begin{proof}
  According to the discussion at the end of
  Section~\ref{sec:preliminaries} about classical states
  and~\eqref{eq:sigmadefinitionalternative}, it suffices to show that
  the operator $\sigma_C^\half P^\alpha_{BC}\sigma_C^\mhalf$ has the
  form
\begin{align}
\sigma_C^\half P^\alpha_{BC}\sigma_C^\mhalf=\sum_b \proj{b}\otimes O^b_C\ ,\label{eq:sigmaCpBC}
\end{align}
for some operators $\{O^b_C\}_b$ on $C$, where $\{\ket{b}\}_b$ is the eigenbasis of $\rho_B$. 

Because $\sigma_C^\mhalf$ acts only on $C$, the state 
$\sigma_C^\mhalf\ket{\Psi_{ABCD}}$ is classical on $B$ when tracing out $A$ and $D$, i.e., 
\begin{align*}
\fr{\rho_{BC}}{\sigma_C}=\sum_b\proj{b}_B\otimes\theta^b_C\ ,
\end{align*}
for some nonnegative operators $\{\theta^b_C\}_b$ on $C$, because this
is true for the original state $\ket{\Psi_{ABCD}}$ by assumption. In
particular, the eigenvectors of $\fr{\rho_{BC}}{\sigma_C}$ are of the
form $\ket{b}_B\ket{\varphi}_C$. Since $P^\alpha_{BC}$ is a projector
onto an eigenspace of this operator, this proves that $P^\alpha_{BC}$
has the form $P^\alpha_{BC}=\sum_b \proj{b}_B\otimes T^b_C$ for some
operators $\{T^b_C\}_b$ on $C$. This immediately gives the
claim~\eqref{eq:sigmaCpBC}.
\end{proof}

As explained in Section~\ref{sec:splittingsamplingrecombining}, we will later apply the splitting-chain-rule recursively. In particular, we will further split up split states. Conveniently, orthogonality properties are preserved under such successive splitting operators, as we now explain.

For concreteness, suppose that we split a state $\ket{\Psi_{A_1B_1CD_1}}$ into states $\{\ket{\Psi^{\alpha_1}_{A_1B_1CD_1}}\}_{\alpha_1}$ satisfying
\begin{align*}
H(A_1|B_1C)_{\rho^{\alpha_1}}+\h{B_1}{C}{}{\rho^{\alpha_1}}{\sigma}\gtrsim \h{A_1B_1}{C}{}{\rho^{\alpha_1}}{\sigma}\ .
\end{align*}
Assume further that $B_1=A_2B_2$ is bipartite. We can then split each $\ket{\Psi^{\alpha_1}_{A_1B_1C_1D_1}}$ further into a family of states $\{\ket{\Psi^{\alpha_1\alpha_2}_{A_1B_1C_1D_1}}\}_{\alpha_2}$ such that
\begin{align*}
H(A_2|B_2C)_{\rho^{\alpha_1\alpha_2}}+\h{B_2}{C}{}{\rho^{\alpha_1\alpha_2}}{\sigma}\gtrsim \h{A_2B_2}{C}{}{\rho^{\alpha_1}}{\sigma}=\h{B_1}{C}{}{\rho^{\alpha_1}}{\sigma}\ 
\end{align*}
for all $(\alpha_1,\alpha_2)$. Diagrammatically, the grouping/splitting of systems can be drawn as
\begin{align*}
B_1&\ \Big\{ \begin{matrix}A_2\\ B_2\end{matrix}\\
C &\  \ \ C\\
\begin{matrix}
A_1\\
D_1
\end{matrix}&\ \Big\} D_2\ .
\end{align*}
Clearly, a desirable property is that these states are orthogonal, such that 
\begin{align*}
\ket{\Psi_{A_1B_1CD_1}}=\sum_{(\alpha_1,\alpha_2)} \ket{\Psi^{\alpha_1\alpha_2}_{A_1B_1CD_1}}
\end{align*}
is a decomposition of $\ket{\Psi_{A_1B_1CD_1}}$ into mutually orthogonal states.

We will prove this statement by considering the corresponding projection operators $\{Q^{\alpha_1}_{A_1D_1}\}_{\alpha_1}$ and $\{Q^{\alpha_2\alpha_2}_{A_2D_2}\}_{(\alpha_1,\alpha_2)}$ (where $D_2=D_1A_1$) defined by the splitting-chain-rule; i.e., these are operators satisfying
\begin{align*}
\ket{\Psi^{\alpha_1}_{A_1B_1CD_1}}&=Q^{\alpha_1}_{A_1D_1}\ket{\Psi_{A_1B_1CD_1}}\\
\ket{\Psi^{\alpha_1\alpha_2}_{A_2B_2CD_2}}&=Q^{\alpha_1\alpha_2}_{A_2D_2}\ket{\Psi^{\alpha_1}_{A_2B_2CD_2}}=Q^{\alpha_1\alpha_2}_{A_2D_2}Q^{\alpha_1}_{A_1D_1}\ket{\Psi_{A_1B_1CD_1}}\ .
\end{align*}
By definition, for every $\alpha_1$, the operators $\{Q^{\alpha_2\alpha_2}_{A_2D_2}\}_{(\alpha_1,\alpha_2)}$ are mutually orthogonal for different $\alpha_2$. We will now show that these operators satisfy the inequality
\begin{align}\label{eq:operatorinequalityrefinement}
Q^{\alpha_1\alpha_2}_{A_2D_2}\leq Q^{\alpha_1}_{A_1D_1}\ .
\end{align}
for all $(\alpha_1,\alpha_2)$. This expresses the fact that the operators $Q^{\alpha_1\alpha_2}_{A_2D_2}$ are a ``refinement'' of $Q^{\alpha_1}_{A_1D_1}$. In particular, their images are orthogonal for different values of $\alpha_1$, and we have $Q^{\alpha_1\alpha_2}_{A_2D_2}Q^{\alpha_1}_{A_1D_1}=Q^{\alpha_1\alpha_2}_{A_2D_2}$ (cf. Lemma~\ref{lem:orderingofprojectors}~\eqref{eq:orderingprojectorssub}). In other words, each of the states $\ket{\Psi^{\alpha_1\alpha_2}_{A_1B_1CD_1}}$ can be obtained by applying a single projection to $\ket{\Psi_{A_1B_1CD_1}}$.

The proof involves the following property of the projection operators.
\begin{remark}[Operator inequalities]\label{rem:orthogonalitypreservation}
Let  $\rho_{ABCD}=\proj{\Psi_{ABCD}}$, $Q^\alpha_{AD}$, and $\rho^{\alpha}_{ABCD}=\proj{\Psi^\alpha_{ABCD}}$ be defined as  in
Lemma~\ref{lem:splitfirst}.  Let $\id_{\supp(\rho)}$ denote the
projector onto the support of the operator $\rho$. Then\footnote{Recall that, according to our convention, the first inequality is an abbreviation for the operator inequality $\id_{\supp(\rho^\alpha_{ADF})}\leq Q^\alpha_{AD} \otimes \id_F$ (see Section~\ref{sec:preliminaries} for more details).}
\begin{align*}
\id_{\supp(\rho^\alpha_{ADF})}\leq Q^\alpha_{AD}\leq
\id_{\supp(\rho_{AD})} 
\end{align*}
for any subsystem $F\subseteq BC$. (By that, we mean that 
$\cH_{BC}$ is the  product $\cH_{BC}\cong\cH_F\otimes\cH_G$ of two systems $F$ and $G$, such   that $BC=FG$.)
\end{remark}
Indeed, the second inequality of this remark gives 
\begin{align*}
Q^{\alpha_1\alpha_2}_{A_2D_2}\leq \id_{\supp(\rho^{\alpha_1}_{A_2D_2})}=\id_{\supp(\rho^{\alpha_1}_{A_1A_2D_1})}\ 
\end{align*}
because $D_2=D_1A_1$,
whereas the first inequality with $F=A_2$ (recall that $B_1=A_2B_2$) leads to
\begin{align*}
\id_{\supp(\rho^{\alpha_1}_{A_1D_1A_2})}\leq Q^{\alpha_1}_{A_1D_1}\ .
\end{align*}
This proves the fundamental property~\eqref{eq:operatorinequalityrefinement}.

It remains to give a proof of the statement made in the remark.
\begin{proof}
According to
Lemma~\ref{lem:orderingofprojectors}~\eqref{eq:supportorderingprojectors},
it suffices to show that
\begin{align*}
\supp(\rho^\alpha_{ADF})\subseteq \supp(Q^\alpha_{AD} \otimes \id_F)
\end{align*}
and 
\begin{align}\label{eq:toprovesupportinclusion}
\supp(Q^\alpha_{AD})\subseteq \supp(\rho_{AD})\ .
\end{align}
The first of these inequalities is a direct consequence of the fact
that $\rho^\alpha_{ADF}=(\id_F\otimes
Q^\alpha_{AD})\rho^\alpha_{ADF}(\id_F\otimes Q^\alpha_{AD})$. To prove
the second inequality, observe that $Q^\alpha_{AD}$ projects onto an
eigenspace of the conditional operator $\tr_{\overline{BC}}
(\sigma_C^\mhalf\proj{\Psi_{ABCD}}\sigma_C^\mhalf)$, and thus
\begin{align*}
\supp (Q^\alpha_{AD})\subseteq\supp\left(
\tr_{\overline{BC}}(\sigma_C^\mhalf\proj{\Psi_{ABCD}}\sigma_C^\mhalf)\right)\ .
\end{align*}
The inclusion~\eqref{eq:toprovesupportinclusion} then follows because
the latter set is contained in $\supp(\rho_{AD})$. This can be
verified for example by using a Schmidt decomposition
$\ket{\Psi_{ABCD}}=\sum_\mu \sqrt{\mu}\ket{\mu_{BC}}\ket{\mu_{AD}}$ of
$\ket{\Psi_{ABCD}}$.  In terms of this decomposition, we have
\begin{align*}
\tr_{\overline{BC}}(\sigma_C^\mhalf\proj{\Psi_{ABCD}}\sigma_C^\mhalf)=\sum_{\mu,\mu'}
\tr(\sigma_C^\mhalf\ket{\mu_{BC}}\bra{\mu'_{BC}}\sigma_C^\mhalf)
\ket{\mu_{AD}}\bra{\mu'_{AD}}\ ,
\end{align*}
and the support of this operator is clearly contained in $\myspan\{\ket{\mu_{AD}}\}=\supp(\rho_{AD})$.
\end{proof}

\subsection{Recombination-rules for split states}
As discussed in Section~\ref{sec:partialrecombination}, we will need a converse to the splitting rule which shows that the entropy of the original state is large if it is large for each split state. Here we show how this works in detail in the most simple case. Again, this section may be omitted, but it is instructive for the slightly more intricate case we will need below (cf.~Lemma~\ref{lem:alphamerge}).

Remarkably, the statement we will prove is generally true for any system $F$ which we do not condition on.
\begin{lemma}\label{lem:recombinationbasic}
Let $\ket{\Psi_{ABCD}}$, $\{\ket{\Psi^\alpha_{ABCD}}\}_{\alpha\in [m]}$ and $\{Q^\alpha_{AD}\}_{\alpha\in [m]}$ be as in Corollary~\ref{cor:split}. Let $F\subseteq ABD$ be an arbitrary subsystem. Then
\begin{align*}
 \min_{\alpha\in [m]}\h{F}{C}{}{\rho^\alpha}{\sigma}-2\log m\leq \h{F}{C}{}{\rho}{\sigma}\ .
\end{align*}
\end{lemma}
\begin{proof}
Let $\lambda:=2^{- \min_{\alpha\in [m]}\h{F}{C}{}{\rho^\alpha}{\sigma}}$.  We then have $\rho^\alpha_{FC}\leq \lambda\sigma_C$ for all $\alpha\in[m]$, or
\begin{align*}
\fr{\rho^\alpha_{FC}}{\sigma_C}\leq \lambda \id_{FC}\ .
\end{align*}
Using the commutativity of $Q^\alpha_{AD}$ and $\sigma_C$, we can rewrite this as 
\begin{align*}
\tr_{\overline{FC}}(Q^\alpha_{AD}\fr{\rho_{ABCD}}{\sigma_C}Q^\alpha_{AD})\leq \lambda\id_{FC}\qquad\textrm{ for all }\alpha\in [m]\ .
\end{align*}
  At this point, we use a statement about operators which we state as Lemma~\ref{lem:alphabetmodifiedlem} in the appendix. It tells us that the previous inequalities imply that 
\begin{align*}
\tr_{\overline{FC}}(Q_{AD}\fr{\rho_{ABCD}}{\sigma_C}Q_{AD})\leq \lambda m^2 \id_{FC}\ ,
\end{align*}
where $Q_{AD}=\sum_{\alpha\in [m]}Q^\alpha_{AD}$. Recall that the operators $Q^{\alpha}_{AD}$ are defined in terms of the eigenspaces of $\tr_{\overline{B C}}(\fr{\rho_{ABCD}}{\sigma_C})$. Their definition implies that $Q_{AD}$ restricted to the support of $\tr_{\overline{B C}}\fr{\rho_{ABCD}}{\sigma_C}$ is equal to the identity. Thus the last inequality simply says
\begin{align*}
\fr{\rho_{FC}}{\sigma_C}\leq \lambda m^2\id_{FC}\ .
\end{align*}
Multiplying from the left and the right by $\sigma_C^\half$ gives the claim.
\end{proof}

\section{Entropy sampling\label{sec:entropysampling}}
We now return to our main problem, i.e., the analysis of a state
$\rho_{X^n E}$ with classical part $X^n=(X_1,\ldots,X_n)$, and the
relation of the entropy $\hmin^\varepsilon(X_{\cS}|E)_{\rho}$ of a
randomly chosen subset $\cS\subset [n]$ to the entropy
$\hmin(X^n|E)_{\rho}$ of all classical parts. We proceed as sketched
in Section~\ref{sec:splittingsamplingrecombining}: In
Section~\ref{sec:recursivesplitting}, we describe the recursive
splitting of the joint min-entropy $\hmin(X^n|E)_\rho$ into a sum of
individual contributions of each random variable. We then discuss how
high-entropy components can be recombined to a state with high
min-entropy (Section~\ref{sec:recombiningev}). In particular, we
relate the smooth min-entropy $\hmin^{\varepsilon}(X_\cS|E)_\rho$ to the probability
weight $\omega(\Gamma)$ of a certain set $\Gamma$ under a given
distribution $\omega$. We then study the behavior of a sampler with
respect to this quantity. For this purpose, we introduce the concept
of a \matrixsampler\ in Section~\ref{sec:averagingsamplersmatrix}. We
then show that with high probability over the choice of $\cS$, the
probability $\omega(\Gamma)$ of interest is large
(Section~\ref{sec:highprobabilityweightsampling}).

We finally combine these components in Section~\ref{sec:samplingandrecombining}, where we state our main result, i.e., the preservation of (smooth) min-entropy rates under sampling. 

\subsection{Splitting\label{sec:recursivesplitting}}

We apply the splitting-chain-rule recursively to a state
$\rho_{X^nE}$, where $X_1,\ldots,X_n$ are random variables on an
alphabet $\cX$. Let $\ket{\Psi_{X^nER}}$ be a purification of
$\rho_{X^nE}$ (for simplicity, we will henceforth often omit
subscripts denoting systems, where there is no potential for
confusion).  Furthermore, let $\sigma_E$ be a nonnegative operator on
$E$. In Figure~\ref{fig:splitstates}, we visualise the set of states
introduced in the following definition by a tree.
\begin{definition}[``Split states'']\label{def:splitstates}
Let $\rho_{X^nER}=\proj{\Psi}$ and let
$\rho^{\alpha^j}_{X^nER}=\proj{\Psi^{\alpha^j}}$ be pure states
recursively defined as follows. Set $\ket{\Psi^{\alpha^0}}:=
\ket{\Psi}$. To obtain $\ket{\Psi^{\alpha^j}}$ for $j\in [n]$ and $\alpha^{j}=(\alpha_j,\alpha^{j-1})\in [m]^j=[m]\times [m]^{j-1}$, apply
Corollary~\ref{cor:split} to the state $\ket{\Psi^{\alpha^{j-1}}}$ with
$A=X_j$, $B=\inc{X}{j}$, $C=E$, $D=\dec{X}{j-1}R$. This gives projectors 
$P_{\inc{X}{j}E}^{\alpha^{j}}$ and $Q_{\dec{X}{j}R}^{\alpha^{j}}$; we define $\ket{\Psi^{\alpha^j}}$ as in
Corollary~\ref{cor:split} as 
$\ket{\Psi^{\alpha^j}}=Q_{\dec{X}{j}R}^{\alpha^j}\ket{\Psi^{\alpha^{j-1}}}$.
\end{definition}

\begin{figure}
\begin{center}
\centerline{\includegraphics[scale=1.0]{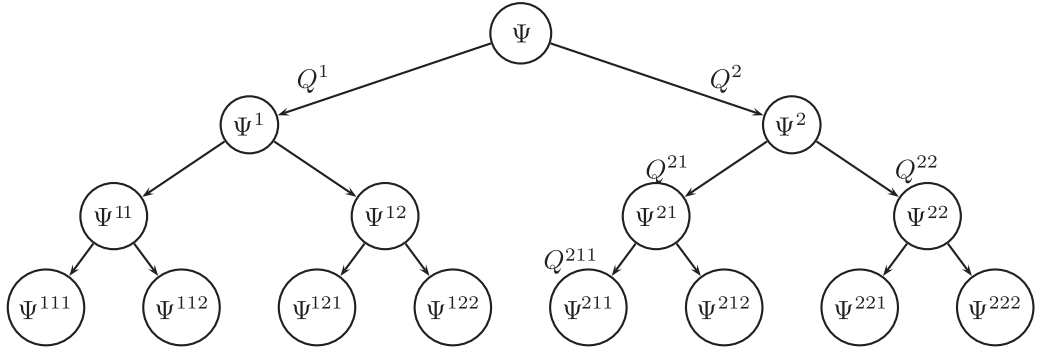}}
\end{center}
\caption{A schematic picture of the states introduced in
  Definition~\ref{def:splitstates}, for $n=3$ and $m=2$. As in
  Figure~\ref{fig:basicsplitting}, the immediate descendants of every
  node give an orthogonal decomposition of the state associated with
  it, and are obtained by applying corresponding projection
  operators. In Lemma~\ref{lem:rhoalphajexpressions}, we will show
  that  the states at level $j$ are orthogonal, for every level $j\in
  [n]$. In particular,  this  means that the leaves form  an orthogonal
  decomposition of the original state. Observe that we label each vertex by the corresponding sequence of splitting operators; in particular, the leaves carry labels $\alpha^n\in [m]^n$.\label{fig:splitstates}}
\end{figure}

Spelling out this recursive definition, we have
\begin{align}
\ket{\Psi^{\alpha^j}}&=Q^{\alpha^j}_{\dec{X}{j}R}\cdots Q^{\alpha^1}_{\dec{X}{1}R}\ket{\Psi}\label{eq:qrecursion}\\
&=\Pt^{\alpha^j}_{\inc{X}{j}E}\cdots
\Pt^{\alpha^1}_{\inc{X}{1}E}\ket{\Psi}\ ,\label{eq:precursion}
\end{align}
where  $\Pt^{\alpha^j}_{\inc{X}{j}E}=\sigma_E^\half
P^{\alpha^j}_{\inc{X}{j}{E}}\sigma_E^\mhalf$. The following auxiliary
result will prove useful. We will apply it to show that the states on each level of the tree in Figure~\ref{fig:splitstates} are mutually orthogonal (by {\em level}, we mean all vertices at a fixed depth of the tree, i.e., distance from the root). In fact, any two states in different subtrees are mutually orthogonal, but we will not need this statement here. The proof of the following lemma relies on the fact that splitting preserves orthogonality. It is analogous to the derivation of~\eqref{eq:operatorinequalityrefinement} in Section~\ref{sec:entropysplitting}.
\begin{lemma}\label{lem:qoperatorsorthogonality}
For all $j\geq k$ and $\alpha^j\in [m]^j$ we have 
\begin{align*}
Q^{\alpha^j}_{\dec{X}{j}R}Q^{\alpha^k}_{\dec{X}{k}R}=Q^{\alpha^k}_{\dec{X}{k}R}Q^{\alpha^j}_{\dec{X}{j}R}=Q^{\alpha^j}_{\dec{X}{j}R}\ .
\end{align*}
Moreover, the operators $\{Q^{\alpha^j}_{\dec{X}{j}R}\}_{\alpha^j\in [m]^j}$
are pairwise orthogonal for a fixed $j\in [n]$.
\end{lemma}
\begin{proof}
Note that the first claim trivially holds for  $j=k$ since
the operators are projectors. Observe that for any $j>1$, we have
\begin{align*}
Q^{\alpha^j}_{\dec{X}{j}R}\leq
\id_{\supp(\rho^{\alpha^{j-1}}_{\dec{X}{j}R})}=
\id_{\supp(\rho^{\alpha^{j-1}}_{X_j\dec{X}{j-1}R})}\leq
Q^{\alpha^{j-1}}_{\dec{X}{j-1}R}\ ,
\end{align*}
where we used Remark~\ref{rem:orthogonalitypreservation} twice (with
$F=X_j$). Inductively, we obtain
\begin{align*}
Q^{\alpha^j}_{\dec{X}{j}R}\leq Q^{\alpha^k}_{\dec{X}{k}R}
\end{align*}
for any $k\leq j$. The first claim therefore follows from
Lemma~\ref{lem:orderingofprojectors}~\eqref{eq:orderingprojectorssub}.

The orthogonality of the operators
$\{Q^{\alpha^j}_{\dec{X}{j}R}\}_{\alpha^j\in [m]^j}$ immediately
follows from the first claim: For $\alpha^j\neq \beta^j\in [m]^j$, let
$k\leq j$ be the minimal index in which they differ, i.e.,
$\alpha_k\neq \beta_k$ and $\alpha^{k-1}=\beta^{k-1}$. We then have by
the first claim
\begin{align*}
Q^{\alpha^j}_{\dec{X}{j}R}Q^{\beta^j}_{\dec{X}{j}R}=Q^{\alpha^j}_{\dec{X}{j}R}Q^{\alpha^k}_{\dec{X}{k}R}Q^{\beta^k}_{\dec{X}{k}R}Q^{\beta^j}_{\dec{X}{j}R}=0\ ,
\end{align*}
since the operators
$Q^{\alpha^k}_{\dec{X}{k}R}=Q^{(\alpha_k,\alpha^{k-1})}_{\dec{X}{k}R}$
and
$Q^{\beta^k}_{\dec{X}{k}R}=Q^{(\beta_k,\alpha^{k-1})}_{\dec{X}{k}R}$
are orthogonal for $\alpha_k\neq \beta_k$.
\end{proof}
As promised, we now establish a few properties of the split states such as their orthogonality and the fact that they are partly classical as the original state.
\begin{lemma}[Properties of the split states]\label{lem:rhoalphajexpressions}
The states introduced in Definition~\ref{def:splitstates} have the
following properties.
\begin{enumerate}[(i)]
\item\label{it:lempropfirst}
The states $\{\ket{\Psi^{\alpha^j}}\}_{\alpha^j\in [m]^j}$ are pairwise orthogonal for a
fixed~$j\in [n]$.
\item\label{it:Psialphasum}
The states $\{\ket{\Psi^{\alpha^n}}\}_{\alpha^n\in [m]^n}$ form an
orthogonal resolution of $\ket{\Psi}$, i.e., 
$\sum_{\alpha^n\in[m]^n}\ket{\Psi^{\alpha^n}}=\ket{\Psi}$.
In particular, $\omega(\alpha^n):=\tr\proj{\Psi^{\alpha^n}}$ defines
a probability distribution on $[m]^n$.
\item\label{it:singleprojectorpsidef}
The state $\ket{\Psi^{\alpha^j}}$ can be obtained by a single
projection on $\dec{X}{j}R$, i.e., $\ket{\Psi^{\alpha^j}}=Q^{\alpha^j}_{\dec{X}{j}{R}}\ket{\Psi}$.
\item\label{it:classicalitysplit}
For every $j\in [n]$ and $\alpha^j\in [m]^j$, the state $\rho^{\alpha^j}_{X^nE}$ is classical on $X^n$.
\item\label{it:hnrhosigma}
For all $\sigma=\sigma_E$, we have $\hcond{E}{\rho^{\alpha^j}}{\sigma}\geq \hcond{E}{\rho}{\sigma}$.
\end{enumerate}
\end{lemma}
The probability distribution $\omega$ (introduced in~\eqref{it:Psialphasum}) on the leaves $[m]^n$ of the tree in Figure~\ref{fig:splitstates} will play an important role in our recombination step. Inequality~\eqref{it:hnrhosigma} can be seen as an expression of the fact that splitting does not affect the part we condition on. 
\begin{proof}
First observe that~\eqref{it:singleprojectorpsidef}
follows inductively from Lemma~\ref{lem:qoperatorsorthogonality} and
expression~\eqref{eq:qrecursion}.  Similarly,  the
orthogonality~\eqref{it:lempropfirst} follows from this lemma
and~\eqref{it:singleprojectorpsidef}. Statement~\eqref{it:Psialphasum} follows by induction over $j$ from~\eqref{eq:identityresolutionpsi}. Statement~\eqref{it:classicalitysplit} follows inductively from Remark~\ref{rem:classicality} applied with $D_1=\dec{X}{j-1}$ and $D_2=R$.  Finally, the claim~\eqref{it:hnrhosigma} directly follows
from~\eqref{it:singleprojectorpsidef} and Lemma~\ref{lem:projectionmonotony} (with $C=\dec{X}{j-1}$ and $B=E$).
\end{proof}

The main reason for introducing the split states $\{\ket{\Psi^{\alpha^j}}\}$
 is the fact that they allow us to split the joint entropy $H(X^n|E)_{\rho}$ into individual contributions according to the splitting-chain-rule (Corollary~\ref{cor:split}). We express this central result as follows.
\begin{theorem}[``Splitting'']\label{thm:splittoffmain}
The split states satisfy 
  \begin{align}
\hcond{E}{\rho^{\alpha^n}}{\sigma}
    + \sum_{j=1}^{n} H(X_j|\inc{X}{j}E)_{\rho^{\alpha^{j}}}
  \geq
    H(\inc{X}{0}|E)_{\fr{\rho}{\sigma}} - \frac{n\log |\cX|}{m} \ .\label{eq:splitoffrecursivesum}
  \end{align}
for any $\alpha^n \in [m]^n$.
\end{theorem}

\begin{proof}
In the following, we sometimes refer to the empty set as $\inc{X}{n}$. By construction and Corollary~\ref{cor:split}, the split states
satisfy the inequalities
\begin{equation} \label{eq:indsplitting}
 H(X_j|\inc{X}{j}E)_{\rho^{\alpha^{j}}} + H(\inc{X}{j}|E)_{\fr{\rho^{\alpha^{j}}}{\sigma}} 
 \geq 
  H(\inc{X}{j-1}|E)_{\fr{\rho^{\alpha^{j-1}}}{\sigma}} - \frac{\Delta_j}{m}\qquad\textrm{for all }j\in [n]\ ,
\end{equation}
 where $\Delta_j=H(\inc{X}{j-1}|E)_{\fr{\rho^{\alpha^{j-1}}}{\sigma}}-H(X_j|\inc{X}{j}E)_{\rho^{\alpha^{j-1}}}-H(\inc{X}{j}|E)_{\fr{\rho^{\alpha^{j-1}}}{\sigma}}$.
Summing these inequalities over all $j\in
[n]$, we get
  \[
    \sum_{j\in [n]} H(X_j|\inc{X}{j}E)_{\rho^{\alpha^{j}}}
  \geq 
    \sum_{j\in [n]}\Bigl( H(\inc{X}{j-1}|E)_{\fr{\rho^{\alpha^{j-1}}}{\sigma}} 
    - H(\inc{X}{j}|E)_{\fr{\rho^{\alpha^{j}}}{\sigma}} \Bigr) - \frac{1}{m}\sum_{j\in [n]} \Delta_j\ .
  \]
Because the rhs is a
  telescoping sum, i.e.,
  \[
    \sum_{j\in [n]}\Bigl( H(\inc{X}{j-1}|E)_{\fr{\rho^{\alpha^{j-1}}}{\sigma}}
      - H(\inc{X}{j}|E)_{\fr{\rho^{\alpha^{j}}}{\sigma}} \Bigr)
  = 
    H(\inc{X}{0}|E)_{\fr{\rho}{\sigma}} - H(\inc{X}{n}|E)_{\fr{\rho^{\alpha^n}}{\sigma}}   \ ,
  \]
this gives
\begin{align}
    H(\inc{X}{n}|E)_{\fr{\rho^{\alpha^n}}{\sigma}} 
    + \sum_{j\in [n]} H(X_j|\inc{X}{j}E)_{\rho^{\alpha^{j}}}
  \geq
    H(\inc{X}{0}|E)_{\fr{\rho}{\sigma}} - \frac{1}{m}\sum_{j\in [n]}\Delta_j \ .\label{eq:vdvein}
\end{align}

Note that $\rho^{\alpha^j}_{X^nE}$ is classical on~$X^n$, according to
Lemma~\ref{lem:rhoalphajexpressions}~\eqref{it:classicalitysplit}.
We can therefore use the  dimension bound~\eqref{it:hzerobound}  of
Lemma~\ref{lem:hminprop} and the positivity of the min-entropy (Lemma~\ref{lem:hminprop}~\eqref{it:classicalcond}) for classical systems to get
\begin{align*}
  \log |\cX| 
& \geq 
  H(X_j \inc{X}{j} | E)_{\fr{\rho^{\alpha^{j-1}}}{\sigma}} 
  - H(\inc{X}{j} | E)_{\fr{\rho^{\alpha^{j-1}}}{\sigma}} \\
& \geq
  H(\inc{X}{j-1} | E)_{\fr{\rho^{\alpha^{j-1}}}{\sigma}} 
  - H(\inc{X}{j} | E)_{\fr{\rho^{\alpha^{j-1}}}{\sigma}} 
  - H(X_j | \inc{X}{j} E)_{\rho^{\alpha^{j-1}}}=\Delta_j\qquad\textrm{ for all }j\in [n]\ .
\end{align*}
The claim follows from this and~\eqref{eq:vdvein}.
\end{proof}

To put the statement of Theorem~\ref{thm:splittoffmain} into a more concise form, it is useful to think of the entropic quantities appearing on the lhs of the inequality~\eqref{eq:splitoffrecursivesum} as attached to the tree given in Figure~\ref{fig:splitstates}. For convenience, we use a slightly modified tree $\mytree{n}$ which has spades attached to the leaves of the original tree (see Figure~\ref{fig:originaltree}). 
\begin{figure}
\begin{center}
\centerline{\includegraphics[scale=1.0]{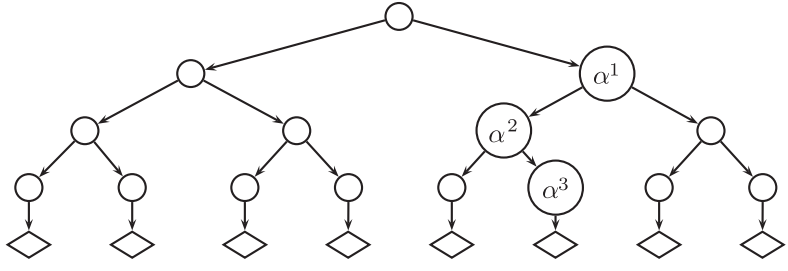}}
\end{center}
\caption{The  tree $\mytree{n}=\mytree{3}$, for $m=2$ and $n=3$.
 Every path from the root to a leaf/spade is specified by an $n$-tuple $\alpha^n\in \{1,2\}^3$. We will attach a weight corresponding to an entropy to every edge in the graph).\label{fig:originaltree}
}
\end{figure}
\begin{figure}
\begin{center}
\centerline{\includegraphics[scale=1.0]{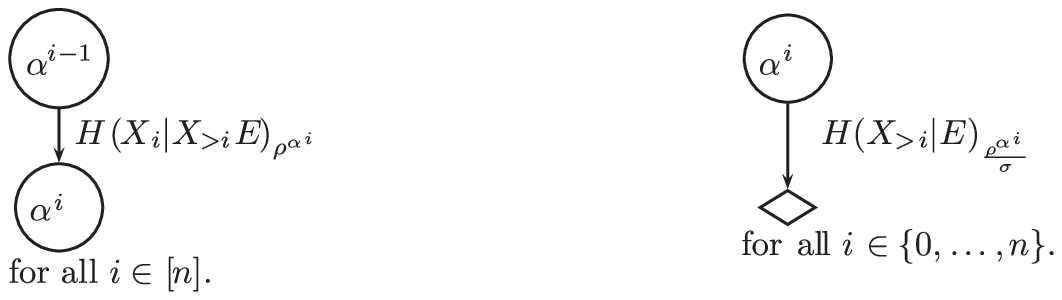}}
\end{center}
\caption{The weighting $\treevsimple{\rho}$ of the edges.\label{fig:treevsimple}}
\end{figure}

We can then attach weights to the edges of $\mytree{n}$ according the rule $\treevsimple{\rho}$ given in Figure~\ref{fig:treevsimple}. For a path $\alpha^n\in [m]^n$ from the root to a spade (i.e., leaf), we define the weight $\treevsimple{\rho}(\alpha^n)$ of the path $\alpha^n$ as the sum of the values on the edges along this path. In particular, for the weighting $\treevsimple{\rho}$ specified by Figure~\ref{fig:originaltree}, the weight $\treevsimple{\rho}(\alpha^n)$ coincides with the lhs of~\eqref{eq:splitoffrecursivesum} in Theorem~\ref{thm:splittoffmain}. 
 
More generally, we slightly abuse notation and define the {\em value} $\mathbf{w}(\mathbb{T})$ of a tree $\mathbb{T}$ with weighting $\mathbf{w}$ as the minimal value of a path from the root to a leaf. Theorem~\ref{thm:splittoffmain} can then be reformulated as follows.

\newtheorem*{varthmsplit}{Theorem~\ref{thm:splittoffmain}$^\prime$}
\begin{varthmsplit}
Let $\mytree{n}$ be the tree introduced in Figure~\ref{fig:originaltree}, and let $\treevsimple{\rho}$ be the weighting specified by Figure~\ref{fig:originaltree}. Then $\treevsimple{\rho}(\mytree{n})\geq    H(\inc{X}{0}|E)_{\fr{\rho}{\sigma}} - \frac{n}{m}\log|\cX|$.
\end{varthmsplit}
We will later be interested in different weightings. We will also show a converse to this statement: If the value of a tree is large, then so is the corresponding entropy.

\subsection{Recombining\label{sec:recombiningev}}
To show that the original state $\rho_{X^nE}$ has a large smooth min-entropy $H^\varepsilon_{\min}(X_\cS|E)$ for a randomly selected subset $\cS\subset [n]$, we will now study how the split states can be recombined. More precisely, we are interested in properties of states $\ket{\bPsi}$ that are obtained by summing up states $\ket{\Psi^{\alpha^n}}$ corresponding to a subset $\Gamma\subset [m]^n$ of leaves of the tree in Figure~\ref{fig:splitstates}.

In Section~\ref{sec:partiallyrecombinedstatesproperties}, we discuss how such a recombined state can be defined recursively, starting from the bottom of the tree. We then use the corresponding intermediate states in Section~\ref{sec:recombininghighentropycomponents} to analyse how a judicious choice of $\Gamma$ yields a recombined state $\ket{\bPsi}$ with a large min-entropy $\hmin(X_{\cS}|E)_{\rhob}$.

\subsubsection{Partially recombined states and properties\label{sec:partiallyrecombinedstatesproperties}}

We are interested in properties of the state
\begin{align}
\ket{\bPsi}=\sum_{\alpha^n\in \Gamma}\ket{\Psi^{\alpha^n}}\label{eq:partiallyrecombinedstateorig}
\end{align}
obtained by summing over a certain subset $\Gamma\subset [m]^n$ of paths. To analyse such a ``partially recombined'' state, we will consider intermediate states attached to a tree. The state $\ket{\bPsi}$ will sit at the root of the tree. We will refer to it as $\ket{\bPsi}=\ket{\bPsi^{\alpha^0}}$ in the following definition, which we  illustrate in Figure~\ref{fig:recombinedstates}.

\begin{definition}[``Recombined states'']\label{def:recombinedstates}
Let $\Gamma\subseteq [m]^n$ be arbitrary, and let $\ket{\Psi^{\alpha^n}}$
for $\alpha^n\in [m]^n$ be the split states introduced in
Definition~\ref{def:splitstates}. We define the recombined states 
\begin{align*}
\ket{\bPsi^{\alpha^j}}=\sum_{\substack{\gamma^n\in\Gamma\\ \gamma^j=\alpha^j} } \ket{\Psi^{\gamma^n}}\ 
\end{align*}
and let $\rhob^{\alpha^j}_{X^nER}=\proj{\bPsi^{\alpha^j}}$ for all
$\alpha^j\in [m]^j$. For simplicity, we omit $\Gamma$ in the notation.
\end{definition}

\begin{figure}
\begin{center}
\centerline{\includegraphics[scale=1.0]{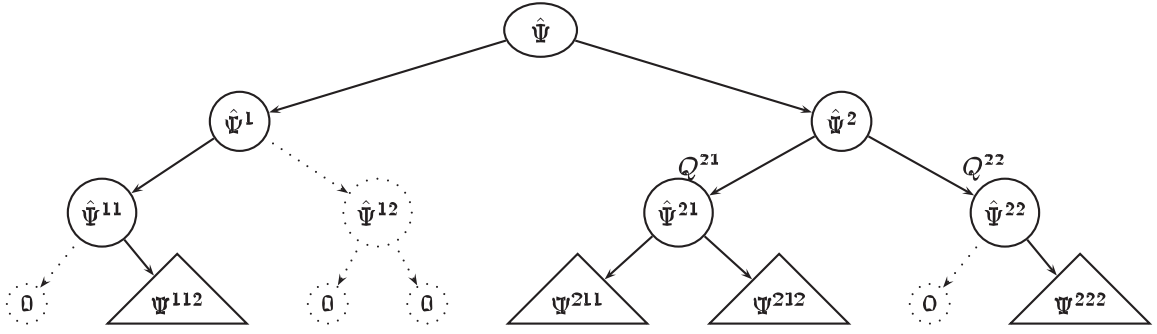}}
\end{center}
\caption{Here we illustrate the partially recombined states of
  Definition~\ref{def:recombinedstates}, 
for $n=3$, $m=2$ and $\Gamma=\{112,211,212,222\}$. We again
 associate every state $\ket{\widehat{\Psi}^{\alpha^j}}$ with the node carrying the label $\alpha^j\in [m]^j$. We start by defining the leaves, i.e., the
states $\ket{\widehat{\Psi}^{\alpha^n}}$ for $\alpha^n\in [m]^n$: For
$\alpha^n\in \Gamma$ (illustrated by triangles), we use the same
leaves as in Figure~\ref{fig:splitstates}, i.e., we set
$\ket{\widehat{\Psi}^{\alpha^n}}=\ket{\Psi^{\alpha^n}}$. On the other
hand, we set $\ket{\widehat{\Psi}^{\alpha^n}}=0$ for
$\alpha^n\not\in\Gamma$. We then work our way up the tree, defining
the state in each node as the sum of its immediate
descendants. The elements at the dotted nodes are equal to zero,
whereas for example $\ket{\widehat{\Psi}^2}=\ket{\widehat{\Psi}^{21}}+\ket{\widehat{\Psi}^{22}}=\ket{\Psi^{211}}+\ket{\Psi^{212}}+\ket{\Psi^{222}}$.
Clearly, the state at the root is equal to the sum of the leaves in
$\Gamma$, i.e.,
$\ket{\widehat{\Psi}}=\sum_{\alpha^n\in\Gamma}\ket{\Psi^{\alpha^n}}$.
We will show in Lemma~\ref{lem:subsetgammastates} that the states at
any given level are orthogonal, and that movement in this diagram is
achieved by the same projection operators as in the tree of
Figure~\ref{fig:splitstates}. Moreover, the entropies of interest
corresponding to this modified tree are at least as large as those
corresponding to Figure~\ref{fig:splitstates}. \label{fig:recombinedstates}
}
\end{figure}
Not surprisingly, the recombined states inherit many properties of the split states. The following lemma summarises these, and is the analog of Lemma~\ref{lem:rhoalphajexpressions}. 

\begin{lemma}\label{lem:subsetgammastates}
 The recombined states
 have the following properties.
\begin{enumerate}[(i)]
\item\label{it:lemrecombinedorthogonality}
The states $\{\ket{\bPsi^{\alpha^j}}\}_{\alpha^j\in [m]^j}$ are
orthogonal for a fixed $j\in [n]$.
\item\label{eq:psibsumid} 
The states $\{\ket{\bPsi^{\alpha^j}}\}_{\alpha^j\in [m]^j}$ form
  a resolution of $\ket{\bPsi^{\alpha^{j-1}}}$, i.e., 
$\sum_{\alpha_j\in [m]}\ket{\bPsi^{\alpha^j}}=\ket{\bPsi^{\alpha^{j-1}}}$.
\item\label{eq:rhobrecursion}
The states satisfy the recursion relation $\ket{\bPsi^{\alpha^j}}=Q^{\alpha^j}_{\dec{X}{j}R}\ket{\bPsi^{\alpha^{j-1}}}$  for all $j\in [n]$ and $\alpha^j\in [m]^j$.
\item\label{it:directbPsiprojections}
For every $j=0,\ldots,n$, there is a projector  $T^{\alpha^j}_{\dec{X}{n}R}$ such that
$\ket{\bPsi^{\alpha^j}}=T^{\alpha^j}_{\dec{X}{n}R}\ket{\Psi}$. In particular, for $\sigma=\sigma_E$ arbitrary, we have $\hcond{E}{\rhob^{\alpha^j}}{\sigma}\geq \hcond{E}{\rho}{\sigma}$.
\item We have \label{eq:rhobentropy}
$H(X_j|\inc{X}{j}E)_{\rhob^{\alpha^j}}\geq
H(X_j|\inc{X}{j}E)_{\rho^{\alpha^j}}$ for all $j\in [n]$ and $\alpha^j\in [m]^j$.

\item For all $\alpha^n\in [m]^n$, we have $\hcond{E}{\rhob^{\alpha^n}}{\sigma}\geq \hcond{E}{\rho^{\alpha^n}}{\sigma}$.\label{it:equalityalphan}
\end{enumerate}
\end{lemma}
The recursion relation~\eqref{eq:rhobrecursion} will be most important in our analysis. It provides a means of studying properties of the corresponding states in a recursive manner, moving up the tree in Figure~\ref{fig:recombinedstates} to the root.

\begin{proof}
The orthogonality~\eqref{it:lemrecombinedorthogonality} of the states
$\{\ket{\bPsi^{\alpha^j}}\}_{\alpha^j}$ is a direct consequence of the
orthogonality of the states $\{\ket{\Psi^{\gamma^n}}\}_{\gamma^n}$
  (cf.\ Lemma~\ref{lem:rhoalphajexpressions}~\eqref{it:lempropfirst}). Identity~\eqref{eq:psibsumid} also follows from the definition
of $\ket{\bPsi^{\alpha^j}}$. 

For the proof of~\eqref{eq:rhobrecursion}, observe that
\begin{align*}
Q^{\alpha^j}_{\dec{X}{j}R}\ket{\Psi^{\gamma^n}}=Q^{\alpha^j}_{\dec{X}{j}R}
Q^{\gamma^n}_{\dec{X}{n}R}\ket{\Psi}=Q^{\alpha^j}_{\dec{X}{j}R}Q^{\gamma^j}_{\dec{X}{n}R}Q^{\gamma^n}_{\dec{X}{n}R}\ket{\Psi}=\begin{cases}
\ket{\Psi^{\gamma^n}}\qquad&\textrm{if } \gamma^j=\alpha^j\\
0 & \textrm{otherwise}\ ,
\end{cases}
\end{align*}
by Lemma~\ref{lem:qoperatorsorthogonality}. Applying this to compute
$Q^{\alpha^j}_{\dec{X}{j}R}\ket{\bPsi^{\alpha^{j-1}}}$ immediately
gives the claim~\eqref{eq:rhobrecursion}.

Defining $T^{\alpha^j}_{\dec{X}{n}R}=\sum_{\substack{\gamma^n\in\Gamma\\
 \gamma^j=\alpha^j}}Q^{\gamma^n}_{\dec{X}{n}R}$ and using the fact
 that  $\ket{\bPsi^{\alpha^j}}=\sum_{\substack{\gamma^n\in\Gamma\\ \gamma^j=\alpha^j}}\ket{\Psi^{\gamma^n}}$
  and 
$\ket{\Psi}=\sum_{\alpha^n} \ket{\Psi^{\alpha^n}}$ proves
 the first part of~\eqref{it:directbPsiprojections} because of
 Lemma~\ref{lem:qoperatorsorthogonality}
and Lemma~\ref{lem:rhoalphajexpressions}. The second part
 of~\eqref{it:directbPsiprojections} follows from Lemma~\ref{lem:projectionmonotony}.

Next we prove~\eqref{eq:rhobentropy}. Note that the statement holds
trivially for $j=n$ and $\alpha^n\in\Gamma$, since in this case
$\ket{\bPsi^{\alpha^n}}=\ket{\Psi^{\alpha^n}}$ by Definition~\ref{def:recombinedstates}. If $j=n$ and $\alpha^n\not\in\Gamma$, then
$\ket{\bPsi^{\alpha^n}}=0$ and $H(X_j|\inc{X}{j}E)_{\rhob^{\alpha^j}}=\infty$ by definition, hence~\eqref{eq:rhobentropy} also holds in this case. Assume  now that $j<n$. We have
$\ket{\bPsi^{\alpha^j}}=R_{\inc{X}{j+1}E}\ket{\Psi^{\alpha^j}}$ for
the operator $R_{\inc{X}{j+1}E}=\sum_{\substack{\gamma^n\in\Gamma\\ \gamma^j=\alpha^j} }
\Pt^{\gamma^n}_{\inc{X}{n}E}\cdots
\Pt^{\gamma^{j+1}}_{\inc{X}{j+1}E}$, by Definition~\ref{def:recombinedstates}
and~\eqref{eq:precursion}. The claim~\eqref{eq:rhobentropy} therefore
follows from Lemma~\ref{lem:projectionmonotony}.

For the proof of statement~\eqref{it:equalityalphan}, we again use the fact that $\ket{\bPsi^{\alpha^n}}=\ket{\Psi^{\alpha^n}}$ if $\alpha^n\in\Gamma$ and $\ket{\bPsi^{\alpha^n}}=0$ otherwise. In particular, we have $\hcond{E}{\rhob^{\alpha^n}}{\sigma}=\hcond{E}{\rho^{\alpha^n}}{\sigma}$ in the former and $\hcond{E}{\rhob^{\alpha^n}}{\sigma}=\infty$ in the latter case. Hence the claim~\eqref{it:equalityalphan} follows.
\end{proof}

We next prove an analog of the basic recombination lemma (Lemma~\ref{lem:recombinationbasic}) for the partially recombined states $\ket{\bPsi^{\alpha^i}}$. In terms of the position of the corresponding states in the described tree, it expresses the fact that the entropies of interest do not decrease significantly when we move from one level up to another level closer to the root.
\begin{lemma} \label{lem:alphamerge}
  For all $\cA \subseteq [n]$ (possibly empty), $i\in [n]$ and $\alpha^{i-1}\in [m]^{i-1}$
\[
\min_{\alpha_i\in [m]}    H(X_\cA|E)_{\fr{\rhob^{\alpha^{i}}}{\sigma}}-2\log m\leq H(X_{\cA}|E)_{\fr{\rhob^{\alpha^{i-1}}}{\sigma}}\ .
\]
\end{lemma}

\begin{proof}
  Let $\alpha^{i-1}\in [m]^{i-1}$ be fixed and let $\lambda:=2^{-\min_{\alpha_{i}\in [m]}
    H(X_\cA|E)_{\frac{\rhob^{\alpha^{i}}}{\sigma}}}$, where $\alpha^i=(\alpha_i,\alpha^{i-1})$.  By definition
  \begin{align}\label{eq:rhobineqsi}
    \frac{\rhob_{X_{\cA}E}^{\alpha^{i}}}{\sigma_E} \leq \lambda
    \id_{X_\cA E} \qquad\textrm{ for all }\alpha_i\in [m]\ .
  \end{align}

To relate this to $\rhob_{X_{\cA}E}^{\alpha^{i-1}}$, we use the
recursion relation~\eqref{eq:rhobrecursion} of Lemma~\ref{lem:subsetgammastates} to rewrite~\eqref{eq:rhobineqsi} as 
  \begin{align*}
    \tr_{\overline{X_{\cA}E}}\bigl(Q^{\alpha^i}_{\dec{X}{i}R}\fr{\rhob^{\alpha^{i-1}}_{X^nER}}{\sigma_E}
Q^{\alpha^i}_{\dec{X}{i}R}\bigr)\leq \lambda \id_{X_{\cA} E} \qquad\textrm{for all }\alpha_i\in [m]\  .
  \end{align*}
  Lemma~\ref{lem:alphabetmodifiedlem} thus implies
  \begin{align}\label{eq:bqineqtr}
    \tr_{\overline{X_{\cA}E}}\bigl(
    Q  \fr{\rhob^{\alpha^{i-1}}_{X^nER}}{\sigma_E}Q\bigr)\leq \lambda m^2\id_{X_{\cA} E} \ ,
  \end{align}
where $Q=\sum_{\alpha_i\in [m]}Q^{\alpha^i}_{\dec{X}{i}R}$.
But
\begin{align*}
    Q  \fr{\rhob^{\alpha^{i-1}}_{X^nER}}{\sigma_E}Q&=
\sigma_E^\mhalf\left(\sum_{\alpha_i}Q^{\alpha^i}_{\dec{X}{i}R}\ket{\bPsi^{\alpha^{i-1}}}\right)\left(\bra{\bPsi^{\alpha^{i-1}}}\sum_{\alpha_i}Q^{\alpha^i}_{\dec{X}{i}R}\right)\sigma_E^\mhalf\\
&=\sigma_E^\mhalf\left(\sum_{\alpha_i}\ket{\bPsi^{\alpha^i}}\right)\left(\sum_{\alpha_i}\bra{\bPsi^{\alpha^i}}\right)\sigma_E^\mhalf\\
&=\sigma_E^\mhalf\ket{\bPsi^{\alpha^{i-1}}}\bra{\bPsi^{\alpha^{i-1}}}\sigma_E^\mhalf\ ,
\end{align*} 
where we used~\eqref{eq:rhobrecursion}
and~\eqref{eq:psibsumid} of Lemma~\ref{lem:subsetgammastates}. Inserting this into~\eqref{eq:bqineqtr} gives
  \begin{align*}
  \fr{\rhob^{\alpha^{i-1}}_{X_\cA E}}{\sigma_E}\leq \lambda m^2\id_{X_{\cA} E} \ ,
  \end{align*}
which concludes the proof.
\end{proof}

\subsubsection{Recombining high-entropy components\label{sec:recombininghighentropycomponents}}

We now study the entropies associated with recombined states, in the
special case where $\Gamma\subset [m]^n$ is chosen as the set of
``high-entropy paths'' for a subset $\cS$. Our main result of this
section is Theorem~\ref{thm:recombining}, which expresses the fact
that the corresponding entropy $\hmin(X_\cS|E)_{\widehat{\rho}}$ is
large.

\begin{figure}
\begin{center}
\centerline{\includegraphics[scale=1.0]{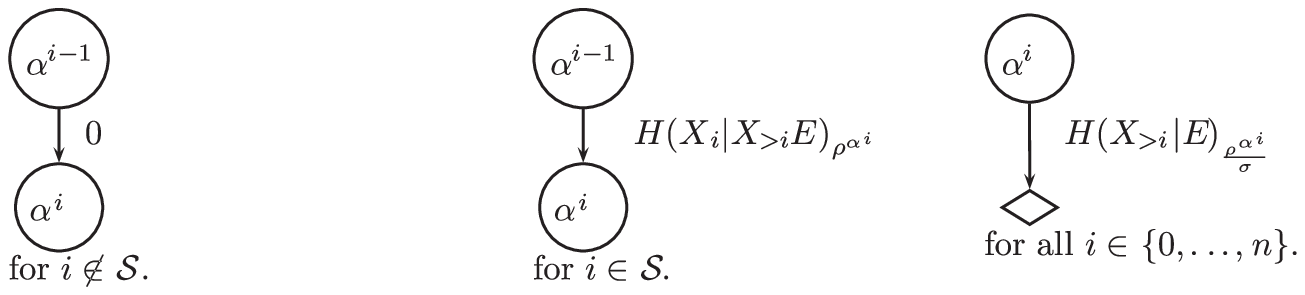}}
\end{center}
\caption{The weighting $\treev{\cS}{\rho}$ of the edges of $\mytree{n}$. The weighting $\treev{\cS}{\hat{\rho}}$ is defined analogously, with $\rho^{\alpha^i}$ replaced by $\hat{\rho}^{\alpha^i}$.\label{fig:simplevsweighting}}
\end{figure}

Let us fix a subset $\cS\subset [n]$. We will be interested in the entropies of variables $X_i$ with $i\in\cS$. That is, we consider the weighting $\treev{\cS}{\rho}$ defined by Figure~\ref{fig:simplevsweighting} of the tree $\mytree{n}$ introduced after Theorem~\ref{thm:splittoffmain}. A given path $\alpha^n\in [m]^n$ in $\mytree{n}$ then has weight
\begin{align*}
\treev{\cS}{\rho}(\mytree{n},\alpha^n)=
\hcond{E}{\rho^{\alpha^n}}{\sigma}  + \sum_{j \in \cS} H(X_j|\inc{X}{j}E)_{\rho^{\alpha^{j}}}
\end{align*}
by definition\footnote{Observe that we now explicitly mention the dependence on the tree $\mathbb{T}_n$ in $\treev{\cS}{\rho}(\mytree{n},\alpha^n)$, as we will be dealing with several different (sub)trees.}. We cannot expect this to be large for all $\alpha^n\in [m]^n$; in particular, the value $\treev{\cS}{\rho}(\mytree{n})$ will in general be small. We therefore introduce the following sets.

\begin{definition}[``$\lambda$-good paths'']\label{def:lambdagoodset}
For $\lambda>0$ and $\cS\subset [n]$, let $\Gamma(\lambda,\cS)\subset [m]^n$ be the set of
$n$-tuples $\alpha^n\in [m]^n$ with
\begin{align}\label{eq:sampleineqv}
\frac{\treev{\cS}{\rho}(\mytree{n},\alpha^n)}{|\cS|\log|\cX|}\geq\lambda\ .
\end{align}
We call $\Gamma(\lambda,\cS)\subset [m]^n$ the set of {\em $\lambda$-good paths for $\cS$}.
\end{definition}
The choice of the normalisation factor $|\cS|\log|\cX|$ will become clearer in the sequel when we relate  the quantity $\frac{\treev{\cS}{\rho}(\mytree{n},\alpha^n)}{|\cS|\log|\cX|}$ to the {\em entropy-rate } $\fr{H(X_{\cS}|E)}{H_0(X_\cS)}=\fr{H(X_{\cS}|E)}{|\cS|\log|\cX|}$.

Let us consider states that arise when recombining only $\lambda$-good
paths. That is, we fix $\lambda>0$, a subset $\cS\subset [n]$ of size $|\cS|=r$, and
let $\Gamma=\Gamma(\lambda,\cS)$ be the set of $n$-tuples specified by Definition~\ref{def:lambdagoodset}. We then define the partially recombined states $\{\rhob^{\alpha^j}\}$ as in Definition~\ref{def:recombinedstates}. 

Note that the recombined states give rise to a weighting $\treev{\cS}{\hat{\rho}}$ of the tree $\mytree{n}$ as in  Figure~\ref{fig:simplevsweighting}. Contrary to the original weighting $\treev{\cS}{\rho}$, this weighting assigns a large weight to every path. That is, we have the statement
\begin{lemma}\label{lem:treevlabelinglowerbound}
$\treev{\cS}{\rhob}(\mytree{n})\geq \lambda|\cS|\log|\cX|$ .
\end{lemma}
In other words, when considering the recombined states, all paths are $\lambda$-good. This is not the case for the original split states.

\begin{proof}
Suppose first that $\alpha^n\in \Gamma(\lambda,\cS)\subset [m]^n$.
Then
\begin{align*}
H(X_j|\inc{X}{j}E)_{\rhob^{\alpha^{j}}} &\geq H(X_j|\inc{X}{j}E)_{\rho^{\alpha^{j}}}\qquad & \textrm{for all $j\in [n]$ by Lemma~\ref{lem:subsetgammastates}~\eqref{eq:rhobentropy} and }\\
 \hcond{E}{\rhob^{\alpha^n}}{\sigma}&\geq \hcond{E}{\rho^{\alpha^n}}{\sigma} \qquad & \textrm{by Lemma~\ref{lem:subsetgammastates}~\eqref{it:equalityalphan}}\ .
\end{align*}
This directly gives $\treev{\cS}{\rhob}(\mytree{n},\alpha^n)\geq \treev{\cS}{\rho}(\mytree{n},\alpha^n)\geq \lambda |\cS|\log |\cX|$ for $\alpha^n\in\Gamma(\lambda,\cS)$.
On the other hand, if $\alpha^n\not\in\Gamma(\lambda,\cS)$, then we have $\rhob^{\alpha^n}=0$ which implies that $\hcond{E}{\rhob^{\alpha^n}}{\sigma}=\infty$ and thus $\treev{\cS}{\rhob}(\mytree{n},\alpha^n)=\infty$.

The claim follows by taking the minimum over $\alpha^n\in [m]^n$.
\end{proof}

Next we apply subadditivity, to go from the weighting $\treev{\cS}{\rhob}$ defined by Figure~\ref{fig:simplevsweighting} to the weighting $\treew{\cS}{\rhob}$ introduced in Figure~\ref{fig:subadditivityweighting}. This weighting assigns the weight
\begin{align*}
\treew{\cS}{\rhob}(\mytree{n},\alpha^n)=  \hcond{E}{\rhob^{\alpha^n}}{\sigma}  + \sum_{j \in \cS} H(X_j|X_{>j \cap \cS} E)_{\rhob^{\alpha^{j}}} 
\end{align*}
to a path $\alpha^n$ in the tree $\mytree{n}$. We then have the inequality
\begin{lemma}\label{lem:subadditivitytrees}
$\treew{\cS}{\rhob}(\mytree{n})\geq \treev{\cS}{\rhob}(\mytree{n})$.
\end{lemma}
\begin{proof}
With subadditivity (Lemma~\ref{lem:hminprop}~\eqref{it:simplesubadditivity}), it is straightforward to show that 
\begin{align*}
H(X_j|X_{>j \cap \cS} E)_{\rhob^{\alpha^{j}}}\geq H(X_j|\inc{X}{j}E)_{\hat{\rho}^{\alpha^j}}
\end{align*}
for all $j\in\cS$ and $\alpha^j\in [m]^j$. The statement follows immediately.
\end{proof}

\begin{figure}
\begin{center}
\centerline{\includegraphics[scale=1.0]{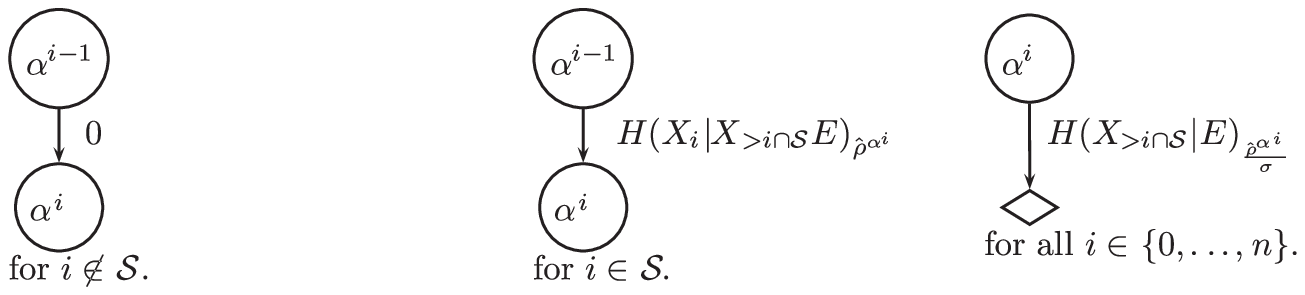}}
\end{center}
\caption{The weighting $\treew{\cS}{\hat{\rho}}$.\label{fig:subadditivityweighting}}
\end{figure}
Our aim is to show that if every path is $\lambda$-good for some $\lambda$, then the entropy $H(X_{\cS}|E)_{\fr{\rhob^{\alpha^0}}{\sigma}}$ is large for the recombined state $\rhob^{\alpha^0}$. This expression can be seen as the value  of the tree $\mytree{0}$ which is defined
in Figure~\ref{fig:zeroanktree}, i.e., we have
\begin{align}\label{eq:treezeroweight}
\treew{\cS}{\hat{\rho}}(\mytree{0})=H(X_{\cS}|E)_{\fr{\rhob^{\alpha^0}}{\sigma}}\ .
\end{align}
\begin{figure}
\begin{center}\centerline{\includegraphics[scale=1.0]{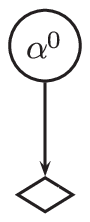}}
\end{center}
\caption{The tree~$\mytree{0}$\label{fig:zeroanktree}}
\end{figure}
To obtain an estimate on this quantity,  we use a sequence of intermediate trees and show the following:
\begin{lemma}
There is a sequence  $\mytree{n-1},\ldots,\mytree{1}$ of intermediate trees such that 
\begin{align}
\treew{\cS}{\hat{\rho}}(\mytree{j-1})\geq \treew{\cS}{\hat{\rho}}(\mytree{j})-2\log m\qquad\textrm{ for all }j\in [n]\ ,\label{eq:treerecursiontoprove}
\end{align}
where $\mytree{0}$ is the tree in Figure~\ref{fig:zeroanktree}, and $\mytree{n}$ is the original tree (see Figure~\ref{fig:originaltree}). In particular,
\begin{align}\label{eq:inductioncompleted}
\treew{\cS}{\hat{\rho}}(\mytree{0})\geq \treew{\cS}{\hat{\rho}}(\mytree{n})-2n\log m\ .
\end{align}
Here $X_{>j \cap \cS}$ denotes  $(X_i)_{i \in \cS, i>j}$ (this is equal to $\emptyset$ if $j>n$). In these expressions, the value of the tree $\mytree{j}$ is equal to 
\begin{align}
\treew{\cS}{\hat{\rho}}(\mytree{j})&=\min_{\alpha^j\in [m]^j}\treew{\cS}{\hat{\rho}}(\mytree{j},\alpha^j)\ \textrm{ where }
\nonumber\\
\treew{\cS}{\hat{\rho}}(\mytree{n},\alpha^j)&=
H(X_{>j\cap
  \cS}|E)_{\fr{\rhob^{\alpha^j}}{\sigma}}+\sum_{\substack{i\leq j\\
    i\in\cS}} H(X_i|X_{>i\cap \cS}E)_{\rhob^{\alpha^i}}\ . \label{eq:valuetreej}
\end{align}
\end{lemma}
\begin{proof}
Note that~\eqref{eq:inductioncompleted} follows  immediately from~\eqref{eq:treerecursiontoprove}. 

We first define the sequence of trees $\mytree{n-1},\mytree{n-2},\ldots,\mytree{0}$. We do this inductively as shown in Figure~\ref{fig:substitutionrule}; that is, we obtain $\mytree{j-1}$ from $\mytree{j}$ by substituting  subtrees corresponding to vertices $\alpha^{j-1}\in [m]^{j-1}$. Clearly, $\mytree{j}$ is a tree characterised as follows: For every $0\leq k\leq j$, every vertex at level $k$ has $m$~immediate descendants, whereas each vertex at level $j$ has one descendant which is a spade. 

\begin{figure}
\begin{center}
\begin{minipage}[t]{0.4\linewidth}
\begin{center}
\centerline{\includegraphics[scale=1.0]{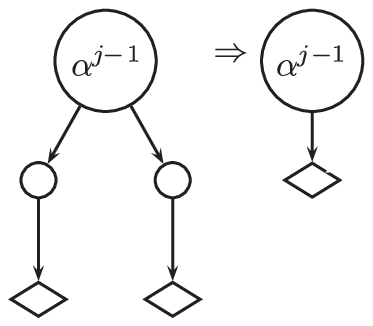}}
\end{center}\vspace{0.11in}
To obtain $\mytree{j-1}$ from $\mytree{j}$, the subtree defined by a vertex $\alpha^{j-1}$ at level $j$ is substituted as shown, for all $\alpha^{j-1}\in [m]^{j-1}$. Note that the vertex $\alpha^{j-1}$ has (in general) $m$ direct descendants; the figure corresponds to $m=2$.
\end{minipage}
\qquad
\begin{minipage}[t]{0.4\linewidth}
\begin{center}
\centerline{\includegraphics[scale=1.0]{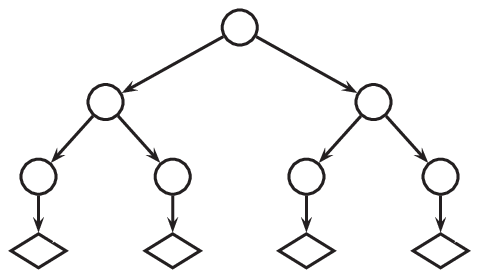}}
\end{center}\vspace{0.1in}
The tree $\mytree{2}$ obtained by applying the substitution rule to the tree $\mytree{3}$ of Figure~\ref{fig:originaltree}.
\end{minipage}
\end{center}
\caption{The substitution rule for obtaining $\mytree{j-1}$ from $\mytree{j}$, $j\in [n]$. The tree $\mytree{j}$ has depth $j+1$, with spades sitting on the $j+1$-st level. \label{fig:substitutionrule}}
\end{figure}

The tree $\mytree{0}$ defined recursively in this way coincides with the definition given above (Figure~\ref{fig:zeroanktree}). Also, it is easy to see that the value of the tree $\mytree{j}$ is given by~\eqref{eq:valuetreej}. We prove the central inequality~\eqref{eq:treerecursiontoprove}.

By definition, it suffices to prove that for all $\alpha^{j-1}\in [m]^{j-1}$, there is an $\alpha_j\in [m]$ such that 
\begin{align*}
\treew{\cS}{\hat{\rho}}(\mytree{j-1},\alpha^{j-1})\geq \treew{\cS}{\hat{\rho}}(\mytree{j},\alpha^j)-2\log m\ , 
\end{align*}
or equivalently 
\begin{align}\label{eq:toprovedeltaineq}
\delta=\min_{\alpha_j\in [m]}\treew{\cS}{\hat{\rho}}(\mytree{j},\alpha^j)-\treew{\cS}{\hat{\rho}}(\mytree{j-1},\alpha^{j-1})\leq 2\log m\ .
\end{align}
Since the two paths to the vertex $\alpha^{j-1}$ are identical in $\mytree{j}$ and $\mytree{j-1}$, the  expression on the lhs~is equal to
\begin{align}\label{eq:subtreeexpr}
\delta=\treew{\cS}{\hat{\rho}}(A)-\treew{\cS}{\hat{\rho}}(B)\ ,
\end{align}
where  $A$ and $B$ are the subtrees defined by $\alpha^{j-1}$ 
on the left in Figure~\ref{fig:substitutionrule}.

By definition, we have
\begin{align*}
\treew{\cS}{\hat{\rho}}(A)&=\begin{cases}
\min_{\alpha_j\in [m]} H(X_{>j\cap \cS}|E)_{\fr{\rhob^{\alpha^j}}{\sigma}}\qquad&\textrm{if }j\not\in\cS\\
\min_{\alpha_j\in [m]} \bigl(H(X_{>j\cap \cS}|E)_{\fr{\rhob^{\alpha^j}}{\sigma}}+H(X_j|X_{>j\cap \cS}E)_{\rhob^{\alpha^j}}\bigr)\qquad&\textrm{if }j\in\cS
\end{cases}\\
\treew{\cS}{\hat{\rho}}(B)&=H(X_{>j-1\cap \cS}|E)_{\fr{\rhob^{\alpha^{j-1}}}{\sigma}}\ .
\end{align*}

We thus have to consider two cases. 
\begin{enumerate}[(i)]
\item
If $j\not\in\cS$, then $X_{>j\cap \cS}=X_{>j-1\cap \cS}$ and~\eqref{eq:toprovedeltaineq} follows with~\eqref{eq:subtreeexpr} once we show that 
\[
\min_{\alpha_j}H(X_{>j\cap
  \cS}|E)_{\fr{\rhob^{\alpha^j}}{\sigma}}-
H(X_{>j\cap S}|E)_{\fr{\rhob^{\alpha^{j-1}}}{\sigma}}\leq 2\log m\ .
\]
This was shown in Lemma~\ref{lem:alphamerge}.
\item
If $j\in\cS$, we have
\[
\delta=\min_{\alpha_j\in [m]}\left(H(X_{>j\cap \cS}|E)_{\fr{\rhob^{\alpha^j}}{\sigma}}+H(X_j|X_{>j\cap \cS}E)_{\rhob^{\alpha^j}}-H(X_{>j-1\cap \cS}|E)_{\fr{\rhob^{\alpha^{j-1}}}{\sigma}}\right)\ .
\]
By the chain-rule (Lemma~\ref{lem:hminprop}~\eqref{it:chainrule}), we have (observe that $X_{>j-1\cap\cS}=X_jX_{>j\cap\cS}$)
\begin{align*}
H(X_{>j\cap \cS}|E)_{\fr{\rhob^{\alpha^j}}{\sigma}}+H(X_j|X_{>j\cap \cS}E)_{\rhob^{\alpha^j}}\leq H(X_{>j-1\cap \cS}|E)_{\fr{\rhob^{\alpha^j}}{\sigma}}
\end{align*}
and thus
\begin{align*}
\delta\leq \min_{\alpha_j\in [m]}\left(H(X_{>j-1\cap \cS}|E)_{\fr{\rhob^{\alpha^j}}{\sigma}}-H(X_{>j-1\cap \cS}|E)_{\fr{\rhob^{\alpha^{j-1}}}{\sigma}}\right)\ .
\end{align*}
The claim~\eqref{eq:toprovedeltaineq} again follows from Lemma~\ref{lem:alphamerge}.
\end{enumerate}
\end{proof}

In summary, we have shown the following:
\begin{lemma}\label{lem:maintreerecombination}
  Let $\cS\subset [n]$ be arbitrary and let
  $\Gamma(\lambda,\cS)\subset [m]^n$ be the set of $\lambda$-good
  paths for $\cS$ as in Definition~\ref{def:lambdagoodset}.  Let
  $\{\rhob^{\alpha^j}\}$ be the corresponding partially recombined
  states as in Definition~\ref{def:recombinedstates}.  Then
\[
\frac{H(X_{\cS}|E)_{\fr{\rhob^{\alpha^0}}{\sigma}}}{|\cS|\log|\cX|}\geq \lambda-\frac{2n\log m}{|\cS|\log|\cX|}\ .
\]
\end{lemma}
\begin{proof}
We have 
\begin{align*}
H(X_{\cS}|E)_{\fr{\rhob^{\alpha^0}}{\sigma}}&= \treew{\cS}{\hat{\rho}}(\mytree{0})\qquad & \text{by \eqref{eq:treezeroweight},}\\
& \geq \treew{\cS}{\hat{\rho}}(\mytree{n})-2n\log m\qquad & \text{\eqref{eq:inductioncompleted},}\\
&\geq \treev{\cS}{\hat{\rho}}(\mytree{n})-2n\log m\qquad & \text{Lemma~\ref{lem:subadditivitytrees} and} \\
& \geq \lambda|\cS|\log|\cX|-2n\log m \qquad & \text{Lemma~\ref{lem:treevlabelinglowerbound}.}
\end{align*}
\end{proof}
We have shown that when recombining only $\lambda$-good paths, one ends up with a state with high entropy on the subset $\cS$ of systems of interest. The recombined state can, however, be far from the original state, if only a few paths are $\lambda$-good (or more precisely, if the share of the $\lambda$-good paths is small). We express this as follows.

\begin{theorem}[``Recombining''] \label{thm:recombining}
There is a probability distribution $\omega$ on $[m]^n$
such that for any subset $\cS\subset [n]$, there
is a subnormalised state $\bar{\rho}_{X^nER}$ with
\begin{align*}
\frac{H(X_{\cS}|E)_{\fr{\bar{\rho}}{\sigma}}}{|\cS|\log|\cX|}&\geq \lambda-\frac{2n\log m}{|\cS|\log|\cX|}\ ,\\
\hcond{E}{\bar{\rho}}{\sigma}&\geq \hcond{E}{\rho}{\sigma}
\end{align*}
at distance
\begin{align*}
\frac{1}{2}\|\bar{\rho}_{X^nER}-\rho_{X^nER}\|\leq \sqrt{1-\omega(\Gamma(\lambda,\cS))}\ ,
\end{align*}
from the original state $\rho_{X^nER}$, where $\Gamma(\lambda,\cS)\subset [m]^n$ is the set of
paths $\alpha^n\in [m]^n$ such that
\begin{align}
\hcond{E}{\rho^{\alpha^n}}{\sigma}  + \sum_{j \in \cS} H(X_j|\inc{X}{j}E)_{\rho^{\alpha^{j}}}
\geq 
  \lambda |\cS|\log|\cX| \qquad\textrm{ for all } \alpha = \alpha^n \in \Gamma(\lambda,\cS)\ .
\end{align}
\end{theorem}
\begin{proof}
Let $\omega$ be the probability distribution introduced in
Lemma~\ref{lem:rhoalphajexpressions}. 
We set $\bar{\rho}=\rhob^{\alpha^0}$ equal to the partially recombined state~\eqref{eq:partiallyrecombinedstateorig}.

The first bound was derived in  Lemma~\ref{lem:maintreerecombination}.
The second bound is identical to the claim~\eqref{it:directbPsiprojections}
of Lemma~\ref{lem:subsetgammastates} for $j=0$. For the bound on the
distance between $\bar{\rho}$ and $\rho$, we use the fact that 
$\rhob^{\alpha^0}=Q\rho Q$, where
$Q=\sum_{\gamma^n\in\Gamma}Q^{\gamma^n}_{\dec{X}{n}R}$ is a
projector (cf. Lemma~\ref{lem:qoperatorsorthogonality}). Applying the  gentle measurement lemma~\cite{winter99,ogawanagaoka02}
\begin{align*}
\frac{1}{2}\|\rho-Q\rho Q\|\leq \sqrt{\tr(\rho)-\tr(Q^2\rho)}\qquad\textrm{ for all subnormalised } \rho\textrm{ and } 0\leq Q\leq \id\  
\end{align*}
gives the claim.
\end{proof}

\subsection{Averaging samplers and  \matrixsamplers\label{sec:averagingsamplersmatrix}}
To argue that an averaging sampler picks $\lambda$-good paths with high probability, it will be necessary to analyse the behavior of a sampler with respect to values attached to a tree. For simplicity, we consider an even simpler situation (which is more general and sufficient for our purposes): We think of values arranged in a matrix, and introduce the concept of a {\em \matrixsampler}.

Consider a modified sampler situation, where instead of a single vector $\beta=(\beta_1,\ldots,\beta_n)\in [0,1]^n$, a family $\{\beta^\alpha=(\beta_1^\alpha,\ldots,\beta_n^\alpha)\in [0,1]^n\}_{\alpha\in [M]}$ of $M$~vectors is given. We would like to approximate the values $\bar{\beta}^\alpha=\frac{1}{n}\sum_{i=1}^n\beta^\alpha_i$ simultaneously by expressions of the form~$\frac{1}{|\cS|}\sum_{i\in\cS}\beta^\alpha_i$. Clearly, a single (small) subset $\cS\subset [n]$ will generally not give a good approximation for each one of the $M$~vectors. However, it is possible to guarantee that  it does so for most vectors, in the following sense.
\begin{definition}\label{def:matrixsampler}
Let $M,n\in\mathbb{N}$.
For any subset $\cS\subset [n]$, matrix $\beta=(\beta^\alpha_i)_{\alpha\in [M],i\in [n]}\in [0,1]^{M\times n}$ and $\xi\in [0,1]$, let $\BAD(\beta,\cS,\xi)\subset [M]$ be the set of $\alpha\in [M]$ such that 
\[
\frac{1}{|\cS|}\sum_{i\in\cS}\beta^{\alpha}_i\leq\frac{1}{n}\sum_{i=1}^n\beta^{\alpha}_i -\xi\ . 
\]
A {\em $(M,n,\xi,\delta,\varepsilon)$-\matrixsampler} is a distribution $P_{\cS}$ over subsets $\cS$ of $[n]$ with the property that for every 
fixed probability distribution $\omega$ on $[M]$,
\begin{align*}
\Pr_{\cS}\left[\omega(\BAD(\beta,\cS,\xi))\geq \delta \right]\leq \varepsilon 
\textrm{ for all }\beta=(\beta^\alpha_i)_{\alpha\in [M],i\in [n]}\in [0,1]^{M\times n}\ .
\end{align*}
A {\em $(n,\xi,\delta,\varepsilon)$-\matrixsampler} is a $(M,n,\xi,\delta,\varepsilon)$-\matrixsampler\ for any $M\in\mathbb{N}$.
\end{definition}
Clearly, a ``standard'' sampler corresponds to $M=1$. In our application, the matrices $\beta\in [0,1]^{M\times n}$ will not be arbitrary, but have a lot of redundancy. This could perhaps be exploited to find better constructions; however, for our purposes, a \matrixsampler\ is sufficient.

We now use Markov's inequality to obtain the following generic construction of a \matrixsampler; again, more optimal constructions may be possible, but the following one is sufficient for our considerations.
\begin{lemma}\label{lem:matrixsampler}
A $(n,\xi,\varepsilon)$-sampler is a  $(n,\xi,\sqrt{\varepsilon},\sqrt{\varepsilon})$-\matrixsampler.
\end{lemma}
\begin{proof}
Let $M\in\mathbb{N}$ be arbitrary. Fix a probability distribution $\omega$ on $[M]$ and let  $\beta=(\beta^\alpha_i)_{\alpha\in [M],i\in [n]}\in [0,1]^{M\times n}$ be arbitrary. Since the probability on the lhs of~\eqref{eq:averagingsamplerdef} is bounded by~$\varepsilon$ for each vector $(\beta^\alpha_1,\ldots,\beta^\alpha_n)$ with $\alpha\in [M]$, it is also bounded if we choose $\alpha$ independently according to $\omega$. That is, we have
\begin{align*}
\varepsilon\geq \Pr_{\cS,\alpha\in [M]}\left[\alpha\in\BAD(\beta,\cS,\xi)\right]=\ExpE_{\cS}\left[\Pr_{\alpha}[\alpha\in\BAD(\beta,\cS,\xi)]\right]\ .
\end{align*}
Markov's inequality $\Pr[Z\geq c]\leq
\textfrac{\ExpE[Z]}{c}$ with $c=\sqrt{\varepsilon}$ applied to
the random variable $Z(\cS)=\Pr_{\alpha}[\alpha\in \BAD(\cS)]$ immediately gives the claim.
\end{proof}

\subsection{Sampling $\lambda$-good paths\label{sec:highprobabilityweightsampling}}
We now apply the concept of a \matrixsampler\ to the situation of interest. Recall Definition~\ref{def:lambdagoodset} of the set
$\Gamma(\lambda,\cS)\subset [m]^n$ of $\lambda$-good
paths for every $\lambda>0$ and $\cS\subset [n]$.
We show that for an appropriate choice of $\lambda$, and a fixed
probability distribution $\omega$ on $[m]^n$, the weight of the $\lambda$-good
paths for $\cS$ is large with high probability if $\cS\subset [n]$ is a random subset which is a \matrixsampler.

\begin{theorem}[``Sampling'']\label{thm:sampling}
Let $\omega$ be an arbitrary probability distribution on
$[m]^n$. Let $P_\cS$ be a probability distribution over subsets of $[n]$ which is a $(n,\xi,\delta,\varepsilon)$-\matrixsampler. Then
\begin{align*}
\Pr_\cS\left[\omega(\Gamma(\lambda,\cS))\geq
  1-\delta\right]&\geq 1-\varepsilon\qquad\textrm{for }\\
\lambda&:=\frac{H(\inc{X}{0}|E)_{\fr{\rho}{\sigma}}}{n\log|\cX|}+\frac{n-|\cS|}{|\cS|n\log|\cX|}\hcond{E}{\rho}{\sigma}- \bigl(\frac{1}{m}  + \xi \bigr)\ ,
\end{align*}
where $\Gamma(\lambda,\cS)$ is the set of $\lambda$-good paths as in Definition~\ref{def:lambdagoodset}, i.e., the set of $\alpha^n\in [m]^n$ with
\begin{align*}
\hcond{E}{\rho^{\alpha^n}}{\sigma}+\sum_{j\in\cS} H(X_j|\inc{X}{j}E)_{\rho^{\alpha^j}}\geq \lambda|\cS|\log|\cX|\ .
\end{align*}
\end{theorem}
\begin{proof}
For every $j\in [n]$ and $\alpha^n\in [m]^n$, we define the quantity
\[
  \beta^{\alpha^n}_j=\frac{H(X_j|\inc{X}{j}E)_{\rho^{\alpha^{j}}}}{\log |\cX|} \ .
\]
 (Note that this depends only on the first $j$~entries of $\alpha^n$.)
Observe that we have $\beta^{\alpha^n}_j\in [0,1]$ by the dimension bound (Lemma~\ref{lem:hminprop}~\eqref{it:hzerobound}). By Definition~\ref{def:matrixsampler} of a \matrixsampler, we therefore get
\begin{align}\label{eq:inequalitybadv}
\Pr_{\cS}\left[\omega(\BAD(\beta,\cS,\xi))\geq \delta\right]\leq \varepsilon\ ,
\end{align}
where 
\begin{align*}
\BAD(\beta,\cS,\xi)=\big\{\alpha^n\in [m]^n\ \big|\ \frac{1}{|\cS|}\sum_{i\in\cS}\beta^{\alpha^n}_i\leq \frac{1}{n}\sum_{i\in[n]}\beta^{\alpha^n}_i-\xi \big\}\ .
\end{align*}
Inequality~\eqref{eq:inequalitybadv} can  be rewritten as
\begin{align}
\Pr_{\cS}\left[\omega(\overline{\BAD(\beta,\cS,\xi)})\geq 1-\delta\right]\geq 1-\varepsilon\ ,\label{eq:inequalitybadvrewritten}
\end{align}
where we write $\overline{\BAD(\beta,\cS,\xi)}=[m]^n\backslash\BAD(\beta,\cS,\xi)$ for the complement of $\BAD(\beta,\cS,\xi)$.

Note that if $\alpha^n\in \overline{\BAD(\beta,\cS,\xi)}$, then 
\begin{align*}
\frac{\log|\cX|}{|\cS|}\sum_{i\in\cS}\beta^{\alpha^n}_i&\geq \frac{\log|\cX|}{n}\sum_{i\in [n]}\beta^{\alpha^n}_i -\xi\log|\cX|\\
&\geq\frac{1}{n}\left(H(\inc{X}{0}|E)_{\fr{\rho}{\sigma}}-\frac{n\log|\cX|}{m}-
\hcond{E}{\rho^{\alpha^n}}{\sigma}\right) -\xi\log|\cX|\ 
\end{align*}
by the definition of $\beta_i^{\alpha^n}$ and Theorem~\ref{thm:splittoffmain}.
This is equivalent to
\begin{align} \label{eq:sampledentropy}
\treev{\cS}{\rho}(\mytree{n},\alpha^n)=\hcond{E}{\rho^{\alpha^n}}{\sigma}  + \sum_{j \in \cS} H(X_j|\inc{X}{j}E)_{\rho^{\alpha^{j}}}
\geq
 \lambda^{\alpha^n} |\cS|\log|\cX|
\end{align}
where
\[
  \lambda^{\alpha^n}
:=
  \frac{H(\inc{X}{0}|E)_{\fr{\rho}{\sigma}}}{n\log|\cX|} -\frac{1}{m}+\frac{n-|\cS|}{|\cS| n\log|\cX|} \hcond{E}{\rho^{\alpha^n}}{\sigma} - \xi\ .
\]
We use Lemma~\ref{lem:rhoalphajexpressions}~\eqref{it:hnrhosigma} (with $j=n$) to bound the second summand
from below, getting 
$\lambda^{\alpha^n}\geq \lambda$ for all $\alpha^n\in [m]^n$.  With~\eqref{eq:sampledentropy}, we conclude that
\begin{align}
\treev{\cS}{\rho}(\mytree{n},\alpha^n)\geq\lambda|\cS|\log|\cX|\qquad\textrm{ for all }\alpha^n\in \overline{\BAD(\beta,\cS,\xi)}\ .
\end{align}
In other words, we have $\overline{\BAD(\beta,\cS,\xi)}\subset \Gamma(\lambda,\cS)$, and the claim follows from~\eqref{eq:inequalitybadvrewritten}.
\end{proof}

\subsection{Sampling and recombining: preservation of smooth entropy rate\label{sec:samplingandrecombining}}
We will now turn our attention to the smooth min-entropy, as introduced in~\cite{Ren05}. We will state and prove our main result in this section; that is, we will show that smooth min-entropy rate is preserved under sampling.

Before discussing our main result, we quickly review an important special case: We will often consider situations where a random variable $Z=f(X,Y)$ is the result of applying a function to two random variables $X$ and $Y$. An example of this is the case where $Z$ is a randomly chosen substring of $X$. To show that the uncertainty about $f(X,Y)$ is large given $Y$ and a quantum system $E$, it suffices to show that with high probability over $Y$, the uncertainty about $f(X,y)$ is large. This is expressed by the following result.
\begin{lemma}\label{lem:hminclassicalmarkov}
Let $\rho_{ZYE}$ be such that
\begin{align*}
\Pr_{y}\left[\hmin^\delta(Z|E,Y=y)\geq k\right]\geq 1-\varepsilon\ .
\end{align*}
Then $\hmin^{\delta+\varepsilon}(Z|YE)\geq k$.
\end{lemma}
\noindent The proof of this lemma is deferred to Appendix~\ref{sec:entropypropertyappendix}.

Recall that the smooth min-entropy-rate $\hminrate^\varepsilon(A|B)_\rho$ is defined as in~\eqref{eq:smoothminentropyrate} as the smooth min-entropy $\hmin^{\varepsilon}(A|B)_{\rho}$ divided by the size $H_0(A)$ of $A$. Our main result is the following
\begin{theorem}\label{thm:mainresultsampling}
Let $\rho_{X^nE}$ be a quantum state where $X^n=(X_1,\ldots,X_n)$ on $\cX^n$ is classical. Let $\cS$ be a random variable over subsets of $[n]$ which is independent of $X^nE$ and a $(n,\xi,\delta,\varepsilon)$-\matrixsampler. Assume that 
 $\kappa=\frac{n}{|\cS|\log |\cX|}\leq 0.15$.
Then
\begin{align*}
\hminrate^{2\sqrt{\delta}+\varepsilon+2\theta+\tau}(X_{\cS}|\cS E)_\rho&\geq \hminrate^\tau(X^n|E)_\rho-\Delta
\qquad\textrm{where }\\
\Delta&=\xi+\frac{2\log\textfrac{1}{\theta}}{n\log|\cX|}+2\kappa\log \textfrac{1}{\kappa}\ ,
\end{align*}
for all $\theta,\tau\geq 0$.
\end{theorem}
We will give concrete parameters below, which show that $\Delta\rightarrow 0$ (in some security parameter), in situations of interest. To put this result into a more convenient form, we choose a certain value of $\theta$, and show how this result applies to general samplers.
\begin{corollary}\label{cor:samplersminentropypreservation}
Let $\rho_{X^nE}$ be a quantum state as in Theorem~\ref{thm:mainresultsampling} and let $\cS$ be a $(n,\xi,\varepsilon)$-sampler.
Assume that $\kappa=\frac{n}{|\cS|\log |\cX|}\leq 0.15$.
Then
\begin{align*}
\hminrate^{\varepsilon'+\tau}(X_{\cS}|\cS E)_\rho&\geq \hminrate^{\tau}(X^n|E)_\rho-3\xi-2\kappa \log\textfrac{1}{\kappa}\textrm{ with}\\
\varepsilon'&=2\cdot 2^{-\xi n\log|\cX|}+3\varepsilon^{\textfrac{1}{4}}
\end{align*}
for all $\tau\geq 0$.
\end{corollary}
\begin{proof}
We choose $\theta=2^{-\xi n \log|\cX|}$. We can then bound $\Delta$ in Theorem~\ref{thm:mainresultsampling} by
\begin{align*}
\Delta\leq 3\xi+2\kappa\log\textfrac{1}{\kappa}\ ,
\end{align*}
and the claim follows from the fact that a
$(n,\xi,\varepsilon)$-sampler is a
$(n,\xi,\sqrt{\varepsilon},\sqrt{\varepsilon})$-\matrixsampler (i.e.,
Lemma~\ref{lem:matrixsampler}).
\end{proof}

In the remainder of this section, we prove Theorem~\ref{thm:mainresultsampling}. We do so in two successive steps. We first show that sampling preserves the entropy rate of a modified (smooth) entropy $\hbmin^\varepsilon(A|B)$. We then use the fact that this modified entropy $\hbmin^\varepsilon$ is essentially equivalent to the smooth min-entropy.  More precisely, we introduce the quantities
\begin{align*}
\hbmin(A|B)_\rho&=\sup_{\sigma_B\geq \rho_B} H(A|B)_{\fr{\rho}{\sigma}}\\
\hbmin^\varepsilon(A|B)_\rho&=\sup_{\substack{\bar{\rho}_{AB}: \|\bar{\rho}_{AB}-\rho_{AB}\|\leq \varepsilon\\
\tr(\bar{\rho}_{AB})\leq 1}}\hbmin(A|B)_{\bar{\rho}}\ 
\end{align*}
for any bipartite state $\rho_{AB}$ and $\varepsilon\geq 0$.  The only difference to the original definition of the (smooth) min-entropy $\hmin^\varepsilon(A|B)_\rho$ (Definition~\eqref{eq:minentropyconvdef}) is that the supremum is restricted to states $\sigma_B$ which are bounded from below by $\rho_B$. These quantities give the bounds
\begin{align}\label{eq:hbminhminbounds}
\hbmin^{2\varepsilon+\delta}(A|B)_\rho+2\log\textfrac{1}{\varepsilon}\geq \hmin^\delta(A|B)_\rho\geq \hbmin^\delta(A|B)_\rho\qquad\textrm{ for all }\varepsilon,\delta\geq 0
\end{align}
on the smooth min-entropy, for all states $\rho_{AB}$. Note that the second inequality follows trivially from the definition; we give a proof of the first inequality in Appendix~\ref{sec:entropypropertyappendix} (Lemma~\ref{lem:hbmin}).

We are ready to combine the recombination
theorem (Theorem~\ref{thm:recombining}) with the sampling
theorem (Theorem~\ref{thm:sampling}). This gives the following main result,
which shows that the min-entropy rate (for the modified entropy
$\hbmin^\varepsilon(A|B)_\rho$) is preserved under sampling.

\begin{lemma}\label{lem:mainresultsamplingentropy}
Consider a quantum state of the form $\rho_{X^nE}$ where $X^n=(X_1,\ldots,X_n)$ on $\cX^n$ is classical. 
Let $P_\cS$ be a probability distribution over subsets of $[n]$ which is a $(n,\xi,\delta,\varepsilon)$-\matrixsampler. Then
\begin{align}
\Pr_{\cS}\left[\frac{\hbmin^{2\sqrt{\delta}}(X_{\cS}|E)_\rho}{|\cS|\log|\cX|}\geq \frac{\hbmin(X^n|E)_\rho}{n\log|\cX|}-c\right]&\geq 1-\varepsilon\qquad\textrm{where }\label{eq:subsetsamplerfir}\\
c& = \xi+\frac{1}{m}+\frac{2n\log m}{|\cS|\log|\cX|}\ \nonumber
\end{align}
for any $m\in\mathbb{N}$.
In particular, inequality~\eqref{eq:subsetsamplerfir} is true for the choice
\begin{align*}
c=\xi+2\kappa\log\textfrac{1}{\kappa}\ 
\end{align*}
if $\kappa=\frac{n}{|\cS|\log |\cX|}\leq 0.15$.
\end{lemma}
\begin{proof}
We first show the second part, assuming that the first statement is true. It is obtained by choosing a specific value of $m\in\mathbb{N}$. Consider the function $f(m):=\frac{1}{m}+2\kappa\log m$. We are interested in the minimum value of this function for $m\in\mathbb{N}$. It is easy to see that the function is minimised for $m_{\min}=\frac{\ln 2}{2\kappa}$, where $\ln$ denotes the natural logarithm. However, since this is not necessarily an integer, we use the value $f(\frac{1}{\ln 2}\cdot m_{\min})$. Clearly, for $\kappa$ small enough ($\kappa\leq \textfrac{1}{2}(1-\ln 2)\approx 0.15$), there is an integer $m_0\in [m_{\min}, \frac{1}{\ln 2}\cdot m_{\min}]$, and this integer satisfies
\begin{align*}
f(m_0)\leq f(\frac{m_{\min}}{\ln 2})=2\kappa\log\textfrac{1}{\kappa}\ .
\end{align*}
The second claim immediately follows from this.

We rephrase the first claim more explicitly: We have to show that the following holds with probability at least $1-\varepsilon$ over the choice of $\cS$. There is a subnormalised state $\bar{\rho}_{X^nER}$ (depending on $\cS$) at distance
\[
\frac{1}{2}\|\rho_{X^nER}-\bar{\rho}_{X^nER}\|\leq \sqrt{\delta}
\]
from the original state $\rho_{X^nER}$ and a state $\sigma_E\geq \rho_E$ (which happens to be independent of $\cS$) such that
\begin{align*}
H(X_\cS|E)_{\fr{\bar{\rho}}{\sigma}}&\geq
\frac{|\cS|}{n}\hbmin(\inc{X}{0}|E)_{\rho}-|\cS|\log|\cX|(\frac{1}{m}+\xi)-2n\log m\\
\hcond{E}{\bar{\rho}}{\sigma}&\geq 0\ .
\end{align*}
Let $\sigma_E$ be a state which achieves the supremum in the definition of $\hbmin(\inc{X}{0}|E)_{\fr{\rho}{\sigma}}$, i.e., we have $\hbmin(\inc{X}{0}|E)_{\rho}=H(\inc{X}{0}|E)_{\fr{\rho}{\sigma}}$ and $\sigma_E\geq \rho_E$. Then  $\hcond{E}{\rho}{\sigma}\geq 0$ by definition, and we can bound the quantity
$\lambda$ in Theorem~\ref{thm:sampling} by
\[ \lambda
\geq
  \frac{H(\inc{X}{0}|E)_{\fr{\rho}{\sigma}}}{n\log|\cX|}  - \bigl(\frac{1}{m}
  + \xi \bigr)\ . 
\] 
According to Theorem~\ref{thm:sampling}, the set $\Gamma(\lambda,\cS)$
of $\lambda$-good paths has weight at least $1-\delta$ with respect to
the distribution $\omega$ (defined by Theorem~\ref{thm:recombining}),
with probability at least $1-\varepsilon$ over the choice of
$\cS$. The recombination theorem (Theorem~\ref{thm:recombining})
therefore guarantees the existence of a state $\bar{\rho}_{X^nER}$
with the required properties, except with probability $\varepsilon$
over the choice of $\cS$.
\end{proof}

We can now prove our main result.

\begin{proof}[Proof of Theorem~\ref{thm:mainresultsampling}]
First observe that Lemma~\ref{lem:mainresultsamplingentropy} can be adapted using the triangle inequality to include an additional parameter $\gamma\geq 0$, thus
replacing the probability in question by
\[
\Pr_{\cS}\left[\frac{\hbmin^{2\sqrt{\delta}+\gamma}(X_{\cS}|E)_\rho}{|\cS|\log|\cX|}\geq \frac{\hbmin^\gamma(X^n|E)_\rho}{n\log|\cX|}-c\right]\geq 1-\varepsilon\ .
\]
Setting $\gamma=2\theta+\tau$ and $\Delta=c+\frac{2\log\textfrac{1}{\theta}}{n\log|\cX|}$, this implies that 
\[
\Pr_{\cS}\left[\frac{\hmin^{2\sqrt{\delta}+2\theta+\tau}(X_{\cS}|E)_\rho}{|\cS|\log|\cX|}\geq \frac{\hmin^\tau(X^n|E)_\rho}{n\log|\cX|}-\Delta\right]\geq 1-\varepsilon\ 
\]
because of  the relations~\eqref{eq:hbminhminbounds} between $\hbmin$ and $\hmin$. The claim of the theorem follows from this inequality and Lemma~\ref{lem:hminclassicalmarkov}.
\end{proof}

\section{Recursive sampling and the bounded storage model\label{sec:bsmapplicationdetailed}}
We will now consider the random subset sampler and analyse recursive sampling from a string. This will result in a concrete protocol for the bounded storage model which achieves significant key expansion.

 The basic building block is the subprotocol $\textsf{Samp}$ described in Figure~\ref{prot:sample}. This protocol outputs  a random substring of a given string $Z$.  The effect of protocol $\textsf{Samp}$ is the following.
\begin{figure}
\begin{myprotocol}{\textsf{Samp}$(Z,r,\cS)$}
\item  {\bf Input}:  $L$-bit string $Z$, parameter $r$ and  $\lceil\log\binom{rL^{\textfrac{1}{4}}}{r}\rceil$ independent random bits $\cS$.
\item {\bf Output}: $L^{\textfrac{3}{4}}$-bit substring of $Z$. 
\item {\bf Procedure}: Partition $Z$ into $n=\frac{L}{t}$ blocks $Z=(X_1,\ldots,X_n)$ of
$t=\frac{L^\textfrac{3}{4}}{r}$ bits each. Use random bits to pick a subset $\cS\subset [n]$ of size $|\cS|=r$ at random. Output $X_\cS$, i.e., the concatenation of the corresponding blocks.
 \end{myprotocol}
\caption{The subprotocol $\textsf{Samp}$. We slightly abuse notation by identifying the random bits with the subset we choose at random.\label{prot:sample}}
\end{figure}
\begin{lemma}\label{lem:sampprocedure}
Let $r$ be fixed, let $\rho_{ZE}$ be a quantum state where $Z$ is an $L$-bitstring with $L\geq r^4$. Let $\cS$ be an independent, uniform $\lceil\log\binom{rL^\textfrac{1}{4}}{r}\rceil$-bit string. (In particular, these are less than $r\log L$ bits.) Let $Z'$ be the $L^{\textfrac{3}{4}}$-bitstring $Z'=\textsf{Samp}(Z,r,\cS)$. Then
\begin{align*}
\hminrate^{\varepsilon'+\tau}(Z'|E\cS)_{\rho}\geq \hminrate^{\tau}(Z|E\cS)_{\rho}-5\frac{\log r}{r^{\textfrac{1}{4}}}\ ,
\end{align*}
for all $\tau\geq 0$, where $\varepsilon'=5\cdot 2^{-\textfrac{\sqrt{r}}{8}}$.
\end{lemma}
Note that we could have used any $(n,\xi,\varepsilon)$-sampler in place of the subset sampler. However, for concreteness and simplicity, we restrict our attention to this sampler; in practice, more efficient constructions may be used, and the analysis is analogous.
\begin{proof}
To prove the bound $r\log L$ on the number of bits consumed, we use the inequality $\binom{p}{q}\leq \left(\frac{p e}{q}\right)^q$ on the binomial coefficients. It implies that
\begin{align*}
\lceil \log \binom{rL^{\textfrac{1}{4}}}{r}\rceil\leq \lceil\log (L^{\textfrac{1}{4}}e)^r\rceil\leq r(\frac{\log L}{4}+\log e)\leq r\log L\ .
\end{align*}

We express everything in terms of the number $r$ of subblocks we sample, and the length $L^{\textfrac{3}{4}}$ of the final string. That is, we sample $r<n$ subblocks from $n=\frac{L}{t}=L^\textfrac{1}{4}r$ blocks of size $t=\frac{L^{\textfrac{3}{4}}}{r}$ bits each, obtaining a substring of  $rt=L^{\textfrac{3}{4}}$ bits.

We know from Lemma~\ref{lem:subsetsampler} that a randomly chosen subset $\cS$ of size $r<n$ is a $(n,\xi,e^{-\textfrac{r\xi^2}{2}})$-sampler, for every $\xi\in [0,1]$. We choose 
$\xi=\frac{1}{r^\textfrac{1}{4}}$ such that $e^{-\textfrac{r\xi^2}{2}}=e^{-\textfrac{\sqrt{r}}{2}}$. The parameter $\varepsilon'$ in Corollary~\ref{cor:samplersminentropypreservation} then takes the form 
\begin{align*}
\varepsilon'&=2\cdot 2^{-\xi n\log|\cX|}+3(e^{-\textfrac{\sqrt{r}}{2}})^{\textfrac{1}{4}}\\
&=2\cdot 2^{-\textfrac{L}{r^{\textfrac{1}{4}}}}+3e^{-\textfrac{\sqrt{r}}{8}}\\
&\leq 5\cdot 2^{-\textfrac{\sqrt{r}}{8}}\ .
\end{align*}
Similarly, we have $\kappa =\frac{r}{L^{\textfrac{1}{2}}}$.  
Because the function $\kappa\mapsto 2\kappa\log\textfrac{1}{\kappa}$ is monotonically increasing for small enough $\kappa$, its value is maximised for  small values of $L$; that is, we can use our lower bound  $L\geq r^4$ on $L$ get
\begin{align*}
2\kappa\log\textfrac{1}{\kappa}&=2\frac{r}{L^{\textfrac{1}{2}}}\log \frac{L^{\textfrac{1}{2}}}{r}\leq 2\frac{\log r}{r}\ .
\end{align*}
We conclude that 
\begin{align*}
3\xi+2\kappa\log\textfrac{1}{\kappa}\leq \frac{3}{r^{\textfrac{1}{4}}}+\frac{2}{r}\log r\leq 5\frac{\log r}{r^{\textfrac{1}{4}}}\ .
\end{align*}
The claim then follows from Corollary~\ref{cor:samplersminentropypreservation}.
\end{proof}

Note that the length $L$ is only reduced to $L^{\textfrac{3}{4}}$ by the protocol $\textsf{Samp}$. To reduce the length of the output even further, we use the protocol recursively. That is, we randomly sample substrings $f$ times, each time sampling a substring of the already obtained string. In each
step, the original string is partitioned into a certain number of
blocks, out of which $r$ are chosen at random (throughout, $r$ will be
a fixed parameter). We call the resulting protocol $\textsf{ReSamp}$; see Figure~\ref{fig:itsampler}.  

To understand the effect of this recursive protocol, observe that the quantities in Lemma~\ref{lem:sampprocedure} describing the effect of $\textsf{Samp}$ are all additive in the following sense: repeated application of $\textsf{Samp}$ simply requires addition of the parameters. Moreover, with the chosen parameters, only the number of bits consumed depends on the length of the involved bitstrings. The analysis is therefore particularly simple, and the effect of the procedure~$\textsf{ReSamp}$ is described by the following lemma.
\begin{lemma}
Let $\rho_{ZE}$ be such that $Z$ is an $L$-bit string.
Let $f$ and $L$ be such that $L^{(\textfrac{3}{4})^f}\geq r^4$. Let $\cS$ be independent random bits and let $Z'=\textsf{ReSamp}(Z,f,r,\cS)$. Then
$Z'$ is a $L^{(\textfrac{3}{4})^f}$-bit substring of $Z$, with
\begin{align}\label{eq:recursiveresolved}
\hminrate^{\varepsilon+\gamma}(Z'|E\cS)_{\rho}\geq \hminrate^{\varepsilon}(Z|E\cS)_{\rho}-5f\frac{\log r}{r^{\textfrac{1}{4}}}\ 
\end{align}
where $\gamma=5f\cdot 2^{-\textfrac{\sqrt{r}}{8}}$. The generation of $Z'$ 
consumes less than $f r\log L$ independent random bits from $\cS$.

In particular, if $L=2^r$, then approximately $f\approx \frac{1}{\log \textfrac{4}{3}}\log\left(\frac{r}{4\log r}\right)$ applications of the subprotocol~$\textsf{Samp}$ are sufficient to produce a substring $Z'$ of $Z$ of length $\lesssim r^4$, while preserving the min-entropy-rate (for large~$r$). This consumes less than $r^3$~independent random bits.
\end{lemma}
\begin{figure}
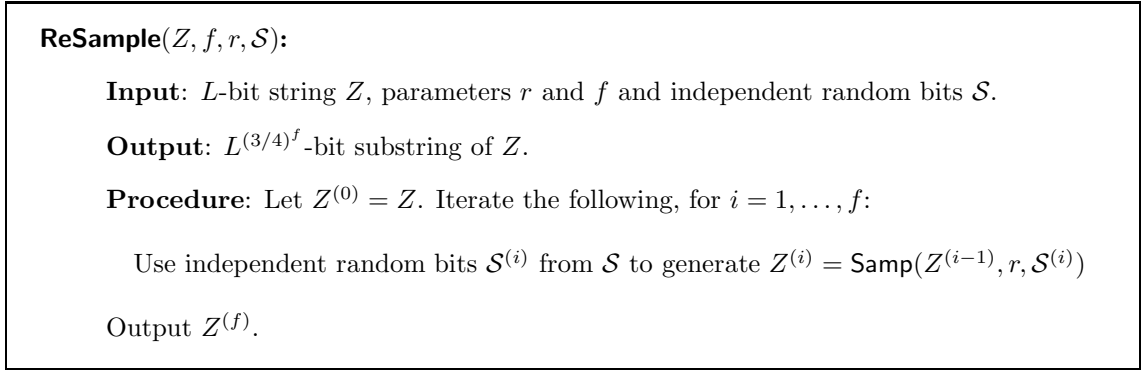

\begin{myprotocol}{\textsf{ReSample}$(Z,f,r,\cS)$}
\item  {\bf Input}:  $L$-bit string $Z$, parameters $r$ and $f$ and independent random bits $\cS$.
\item {\bf Output}: $L^{(\textfrac{3}{4})^f}$-bit substring of $Z$. 
\item {\bf Procedure}: Let $Z^{(0)}=Z$. Iterate the following, for $i=1,\ldots,f$:
\begin{center}
Use independent random bits $\cS^{(i)}$ from $\cS$ to generate $Z^{(i)}=\textsf{Samp}(Z^{(i-1)},r,\cS^{(i)})$
\end{center}
Output $Z^{(f)}$.
 \end{myprotocol}
\caption{The protocol $\textsf{ReSamp}$. It calls the subprotocol $\textsf{Sample}$ $f$~times, each time producing a substring of the already generated string. We will determine the amount of randomness this protocol needs below. Note that the output of this recursive protocol can be computed with limited storage. In particular, it is unnecessary to store the intermediate substrings $Z^{(i)}$.\label{fig:itsampler}}
\end{figure}

\begin{proof}
Let $L^{(i)}$ be the length of the string $Z^{(i)}$; by the definition of the protocol $\textsf{Samp}$, we have $L^{(i)}=(L^{(i-1)})^{\textfrac{3}{4}}$ and hence $L^{(f)}=L^{(\textfrac{3}{4})^f}$, as claimed.
Let $\cS^{(i)}$ be the random bits from $\cS$ used to generate $Z^{(i)}$.
 According to Lemma~\ref{lem:sampprocedure}, we have for all $i=1,\ldots,f$
\begin{align*}
\hminrate^{\varepsilon^{(i-1)}+\delta}(Z^{(i)}|E\cS^{(1)}\cdots \cS^{(i)})_\rho\geq \hminrate^{\varepsilon^{(i-1)}}(Z^{(i-1)}|E\cS^{(1)}\cdots \cS^{(i-1)})_\rho-5\frac{\log r}{r^{\textfrac{1}{4}}}\ ,
\end{align*}
where $\delta=4\cdot2^{-\textfrac{\sqrt{r}}{8}}$ and $\varepsilon^{(i-1)}$ is arbitrary. Defining
$\varepsilon^{(0)}=\varepsilon$ and $\varepsilon^{(i)}=\varepsilon^{(i-1)}+\delta$, this implies the claim~\eqref{eq:recursiveresolved}.

To compute the number of random bits consumed in this procedure\footnote{Note
that we are only interested in the approximate number of bits we need;
see~\cite{Dingetal04} for a dense encoding of subsets of $[n]$ into
bitstrings.}, observe that in the $i$-th step, a random subset $\cS^{(i)}$ of $[n^{(i)}]$ of size $r$ is chosen, where $n^{(i)}=r(L^{(i-1)})^{\textfrac{1}{4}}$; this consumes $\lceil\binom{n^{(i)}}{r}\rceil$ bits.

The total number of bits we need is bounded by the number $f$ of applications times the maximal number of bits 
$\max_{n^{(i)}}\lceil \log \binom{n^{(i)}}{r}\Big\rceil=\lceil \log \binom{rL^{\textfrac{1}{4}}}{r}\Big\rceil\leq r\log L$ consumed in a single step. (Here we used the fact that 
 $n^{(i)}$ is a decreasing sequence and the bound on the binomial coefficient shown in Lemma~\ref{lem:sampprocedure}.) This gives the upper bound $fr\log L$, as claimed.
\end{proof}

In summary, we have found a procedure that generates a random substring of a string of $2^r$ bits, with the following parameters:
\begin{center}
\begin{tabular}{|l|c|c|c|}
\hline
 & original $Z$ & substring $Z'$ & seed $\cS$\\
\hline
length (bits) & $2^r$ & $r^4$ & $r^3$\\
entropy-rate & $R$ (arbitrary) & $R-\textfrac{1}{r^{\Omega(1)}}$ &\\
\hline
\end{tabular}
\end{center}
with error $poly(\log(r))e^{-\Omega(\sqrt{r})}$. It is computable with $poly(r)$ bits of storage. Subsequent to this sampling procedure, privacy amplification may be used to extract a secret key from the substring~$Z'$. In conclusion, this gives a sample-and-hash procedure for key expansion in the bounded storage model, expanding an initial key of $r^3$ bits to  approximately $r^4(R-\textfrac{1}{r^{\Omega(1)}})$ bits. ($R\in [0,1]$ is usually assumed to be constant.)

\appendix

\section{Additional proofs related to entropies~\label{sec:entropypropertyappendix}}
\begin{proof}[Proof of Lemma~\ref{lem:hminclassicalmarkov}]
Let $\rho_{YZE}=\sum_y P_Y(y)\proj{y}\otimes\rho_{ZE|Y=y}$. Let 
\[
\cG:=\left\{y\in\cY\ |\ \hmin^\delta(Z|E,Y=y)\geq k\right\}\ .
\]
For every $y\in\cG$, there is a subnormalised state $\bar{\rho}^y_{ZE}$ and a state $\sigma^y_E$ such that
\begin{align*}
\bar{\rho}^y_{ZE}&\leq 2^{-k}\sigma^y_E\\
\|\bar{\rho}^y_{ZE}-\rho_{ZE|Y=y}\|&\leq \delta\ 
\end{align*}
by definition.
For every $y\not\in\cG$, we choose arbitrary states $\bar{\rho}^y_{ZE}$ and $\sigma^y_E$ satisfying
\begin{align*}
\bar{\rho}^y_{ZE}\leq 2^{-k}\sigma^y_E\ .
\end{align*}
It is easy to verify that the two states
\begin{align*}
\bar{\rho}_{YZE}&=\sum_{y} P_Y(y) \proj{y}\otimes\bar{\rho}^y_{ZE}\\
\sigma_{YE}&=\sum_{y} P_Y(y)\proj{y}\otimes\sigma^y_E
\end{align*}
satisfy
\begin{align*}
\bar{\rho}_{YZE}\leq 2^{-k}\bar{\sigma}_{YE}\\
\|\bar{\rho}_{YZE}-\rho_{YZE}\|\leq \delta+\varepsilon\ .
\end{align*}
The claim follows from this.
\end{proof}
We next prove the non-trivial inequality in~\eqref{eq:hbminhminbounds}.

\begin{lemma}\label{lem:hbmin}
Let $\rho_{AB}$ be a subnormalised state. We have
\begin{align}\label{eq:hbminlowerbound}
\hbmin^{2\varepsilon+\delta}(A|B)_\rho\geq \hmin^{\delta}(A|B)_\rho-2\log\textfrac{1}{\varepsilon}\ .
\end{align}
\end{lemma}
\begin{proof}
Clearly, it suffices to show that the claim of the lemma is true for $\delta=0$. 

By definition, there is a normalised state $\sigma_B$ such that
\begin{align*}
\rho_{AB}\leq 2^{-\hmin(A|B)_\rho}\sigma_B\  .
\end{align*}
This implies that
\begin{align*}
\fr{\rho_{AB}}{\rho_B}\leq 2^{-\hmin(A|B)_\rho}\fr{\sigma_B}{\rho_B}\  .
\end{align*}
Let $P_B$ denote the projector onto the eigenspaces of $\fr{\sigma_B}{\rho_B}$ corresponding to eigenvalues smaller than or equal to $\textfrac{1}{\varepsilon^2}$. Applying this on both sides of the previous inequality gives (with $P_B\leq \id_B$)
\begin{align*}
P_B\fr{\rho_{AB}}{\rho_B}P_B\leq \frac{2^{-\hmin(A|B)_\rho}}{\varepsilon^2} \id_B\ .
\end{align*}
Multiplying from both sides by $\rho_B^\half$, we obtain
\begin{align}
\bar{\rho}_{AB}\leq  \frac{2^{-\hmin(A|B)_\rho}}{\varepsilon^2} \rho_B\ ,\label{eq:hbminl}
\end{align}
where we introduced the operator $\bar{\rho}_{AB}=\rho_B^\half P_B\fr{\rho_{AB}}{\rho_B}P_B\rho_B^\half$. We claim that
\begin{align}
\bar{\rho}_{B}\leq \rho_B\label{eq:barrhob}\ ,\\
\tr(\bar{\rho}_{AB})\leq 1\label{eq:barrhosv}\textrm{ and }\\
\frac{1}{2}\|\bar{\rho}_{AB}-\rho_{AB}\|\leq \varepsilon\ .\label{eq:barrhobsec}
\end{align}
Note that~\eqref{eq:hbminl}--\eqref{eq:barrhobsec} imply the claim~\eqref{eq:hbminlowerbound} (for $\delta=0$).

Inequality~\eqref{eq:barrhob} directly follows from the fact that $\tr_A(\fr{\rho_{AB}}{\rho_B})=\fr{\rho_B}{\rho_B}\leq \id_B$. To prove~\eqref{eq:barrhobsec}, let $\ket{\Psi_{ABC}}$ be a purification of $\rho_{AB}$ and consider the purification
\begin{align*}
\ket{\bar{\Psi}_{ABC}}=\rho_B^\half P_B\rho_B^\mhalf\ket{\Psi_{ABC}}
\end{align*}
of $\bar{\rho}_{AB}$. Using the Schmidt decomposition 
$\ket{\Psi_{ABC}}=\sum_{\lambda}\sqrt{\lambda}\ket{\lambda}_{AC}\ket{\lambda}_B$, it is straightforward to verify that 
\begin{align}
\spr{\Psi_{ABC}}{\bar{\Psi}_{ABC}}=\tr(P_B\rho_B)\qquad\textrm{ and }\qquad\spr{\bar{\Psi}_{ABC}}{\bar{\Psi}_{ABC}}\leq \tr(P_B\rho_B)\leq \tr(\rho_B)\  .\label{eq:Psioverlaps}
\end{align}
Note that the latter of these inequalities proves~\eqref{eq:barrhosv}.
The Cauchy-Schwarz inequality $\sum_{i=1}^d |\lambda_i|\leq \sqrt{d}\sqrt{\sum_{i=1}^d \lambda_i^2}$ implies that for any two pure states $\ket{\chi}$,$\ket{\varphi}$, we have
\begin{align*}
\big\|\proj{\varphi}-\proj{\chi}\big\|\leq \sqrt{2}\big\|\proj{\varphi}-\proj{\chi}\big\|_2
=\sqrt{2}\sqrt{|\spr{\varphi}{\varphi}|^2-2|\spr{\varphi}{\chi}|^2+|\spr{\chi}{\chi}|^2}\ ,
\end{align*}
where $\|A\|_2=\sqrt{\tr(A^\dagger A)}$, 
since the difference $\proj{\varphi}-\proj{\chi}$ has rank at most~$2$. Applying this to $\ket{\Psi_{ABC}}$ and $\ket{\bar{\Psi}_{ABC}}$ and using~\eqref{eq:Psioverlaps} gives
\begin{align*}
\bigl\|\proj{\bar{\Psi}_{ABC}}-\proj{\Psi_{ABC}}\bigr\|&\leq \sqrt{2}\sqrt{\tr(\rho_B)^2-\tr(P_B\rho_B)^2}\\
&=\sqrt{2}\sqrt{\left(\tr(\rho_B)-\tr(P_B\rho_B)\right)\left(\tr(\rho_B)+\tr(P_B\rho_B)\right)}\\
&\leq 2\sqrt{\tr(P_B^\bot \rho_B)}\ .
\end{align*}
Here $P_B^\bot=\id_B-P_B$ projects onto the orthogonal complement of the image of $P_B$. In the last inequality, we have used the  assumption that $\rho_{AB}$ is subnormalised and thus $\tr(\rho_B)\leq 1$.
Since the trace distance is non-increasing under partial traces, the claim~\eqref{eq:barrhobsec} follows once we show that
\begin{align*}
\tr(P_B^\bot\rho_B)\leq \varepsilon^2\ .
\end{align*}
Note that we have $\frac{1}{\varepsilon^2}P_B^\bot\leq P_B^\bot \fr{\sigma_{B}}{\rho_B}P_B^\bot$ by definition. Inserting this into the expression of interest gives
\begin{align*}
\tr(P_B^\bot\rho_B)&\leq \varepsilon^2 \tr\left(P_B^\bot \fr{\sigma_{B}}{\rho_B}P_B^\bot\rho_B\right)\leq\varepsilon^2\tr\left(\fr{\sigma_B}{\rho_B}\rho_B\right) \leq \varepsilon^2\tr(\sigma_B) =\varepsilon^2\ ,
\end{align*}\
as claimed.

\end{proof}

\section{Additional lemmas\label{sec:additional}}

\begin{lemma}\label{lem:alphabetmodifiedlem}
  Let $\{Q^\alpha\}_{\alpha\in [m]}$ be a family of Hermitian operators
  on a Hilbert space $A\otimes B$ and suppose that $\tr_B(Q^\alpha \rho_{AB} Q^\alpha)\leq \sigma_A$, for any $\alpha \in [m]$.
  Then
  \begin{align*}
    \tr_B(Q \rho_{AB} Q) \leq m^2 \sigma_A
  \end{align*}
  for $Q := \sum_\alpha Q^\alpha$.  In particular, if $\{Q^\alpha\}_{\alpha\in [m]}$ resolves the
  identity on $\supp(\rho_{AB})$ then $\rho_A \leq m^2 \sigma_A$.
\end{lemma}

\begin{proof}
Let $\ket{\varphi_A}\in A$ be arbitrary. By definition, we have
\begin{align}
\tr\left(\tr_B(Q\rho_{AB}Q) \proj{\varphi_A}\right)&=\tr\left(Q\rho_{AB}Q(\proj{\varphi_A}\otimes\id_B)\right)\nonumber\\
&=\sum_{\alpha,\beta} \tr\left(Q^\alpha\rho_{AB}Q^\beta(\proj{\varphi_A}\otimes\id_B)\right)\ .\label{eq:Qrhoabineq}
\end{align}
By the cyclicity of the trace and the fact that $\proj{\varphi_A}\otimes\id_B$ is a projector, we have
\begin{align*}
\tr\left(Q^\alpha\rho_{AB}Q^\beta(\proj{\varphi_A}\otimes\id_B)\right)=\tr(Z^\alpha (Z^\beta)^\dagger)\ ,
\end{align*}
with $Z^\alpha=(\proj{\varphi_A}\otimes\id_B)Q^\alpha \rho_{AB}^\half$. The operator-Cauchy-Schwarz-inequality \[
\tr(EF)\leq\sqrt{\tr(E^\dagger E)\tr(F^\dagger F)}
\] applied to $E=Z^\alpha$ and $F=(Z^\beta)^\dagger$ therefore gives (with the cyclicity of the trace)
\begin{align}
\tr\left(Q^\alpha\rho_{AB}Q^\beta(\proj{\varphi_A}\otimes\id_B)\right)\leq \sqrt{\tr\left((Z^\alpha)^\dagger Z^\alpha\right)\tr\left((Z^\beta)^\dagger Z^\beta\right)}\ .\label{eq:qalphaineq}
\end{align}
It is straightforward to verify that 
\begin{align}
\tr\left((Z^\alpha)^\dagger Z^\alpha\right)=\tr\left(\tr_B(Q^\alpha\rho_{AB}Q^\alpha)\proj{\varphi_A}\right)\leq \tr(\sigma_A\proj{\varphi_A})\qquad\textrm{ for all }\alpha\in [m]\ ,\label{eq:zalphaopineq}
\end{align}
where we used the assumption in the last inequality. Combining~\eqref{eq:Qrhoabineq} with \eqref{eq:qalphaineq} and \eqref{eq:zalphaopineq} gives
\begin{align*}
\tr\left(\tr_B(Q\rho_{AB}Q) \proj{\varphi_A}\right)\leq \tr(m^2\sigma_A\proj{\varphi_A})\ .
\end{align*}
Since $\ket{\varphi_A}\in A$ was arbitrary, the claim follows.
\end{proof}

\begin{lemma}\label{lem:orderingofprojectors}
Let $P$ and $P'$ be two projectors on a Hilbert space $\cH$. Then
\begin{enumerate}[(i)]
\item\label{eq:supportorderingprojectors}
If $\supp P\subseteq \supp P'$, then $P\leq P'$.
\item\label{eq:orderingprojectorssub}
If $P\leq P'$, then $PP'=P'P=P$.
\end{enumerate}
\end{lemma}
\begin{proof}
Both statements follow immediately from the fact that
\begin{align*}
\supp P'=\supp P\oplus (\supp P)^{\bot}\ ,
\end{align*}
where $(\supp P)^{\bot}$ is the orthogonal complement of $\supp P$ in $\supp P'$. This identity implies $P'=P+\id_{(\supp P)^{\bot}}$, where $\id_{(\supp P)^{\bot}}$ is the projector onto $(\supp P)^{\bot}$. Thus $PP'=P'P=P$.
\end{proof}


\newcommand{\etalchar}[1]{$^{#1}$}

\end{document}